\documentclass[11pt]{article}%
\usepackage[letterpaper, top=1in, bottom=1in, left=1in, right=1in]{geometry}
\usepackage{amsfonts, amsmath, amsthm, amssymb}
\usepackage{enumitem, graphicx, stmaryrd, color, bbm}
\usepackage{url, algorithm}
\usepackage{float}
\usepackage{footmisc}
\usepackage[T1]{fontenc} 
\floatstyle{ruled}

\usepackage[pagebackref=true,colorlinks]{hyperref}
	\hypersetup{linkcolor=[rgb]{.7,0,.7}}
	\hypersetup{citecolor=[rgb]{.5,0,.5}}
	\hypersetup{urlcolor=[rgb]{.7,0,.7}}

\newtheorem{theorem}{Theorem}[section] 
\newtheorem{lemma}[theorem]{Lemma}

\newtheorem{definition}[theorem]{Definition}
\newtheorem{observation}[theorem]{Observation}

\newtheorem{corollary}[theorem]{Corollary}
\newtheorem{remark}{Remark}					

\newtheorem{proposition}[theorem]{Proposition}
\newtheorem{example}[theorem]{Example}

\newtheorem{fact}[theorem]{Fact}

\newcommand{\abs}[1]{\left|#1\right|}		
\newcommand{\st}{\,|\,} 				 		
\newcommand{\E}{\mathop{\mathbb{E}}}  		
\newcommand{\R}{\mathbb{R}}  				
\newcommand{\N}{\mathbb{N}} 			  		
\newcommand{\F}{\mathbb{F}}					
\newcommand{\C}{\mathbb{C}}					
\newcommand{\Z}{\mathbb{Z}}  				

\newcommand{\mc}{\mathcal}
\newcommand{\mb}{\mathbb}

\newcommand{\eps}{\mathop{\epsilon}}
\newcommand{\set}[1]{\left\{ #1 \right\}}   
\newcommand{\ip}[1]{\left\langle #1 \right\rangle}     
\newcommand{\brac}[1]{\left( #1 \right)}    
\newcommand{\sqbrac}[1]{\left[ #1 \right]}  

\newcommand{\Stab}{\textnormal{Stab}} 				
\newcommand{\Var}{\textnormal{Var}}				

\newcommand{\norm}[1]{\left\vert #1 \right\vert}			
\newcommand{\Norm}[1]{\left\Vert #1 \right\Vert}			

\newcommand{\val}{\textnormal{val}} 	 			

\newcommand{\ts}{\textsuperscript}
\newcommand{\supp}{\textnormal{supp}}

\newcommand{\T}{\textnormal{T}}
\newcommand{\Win}{\textnormal{Win}}
\newcommand{\Lose}{\textnormal{Lose}}
\newcommand{\ind}{\mathbbm{1}}

\let\leq=\leqslant 
\let\geq=\geqslant %

\begin{document}

\title{An Analytical Approach to Parallel Repetition \\ via CSP Inverse Theorems}
\author{
Amey Bhangale\thanks{Department of Computer Science and Engineering, University of California, Riverside. Supported by the Hellman Fellowship award and NSF CAREER award 2440882.}\qquad\quad
Mark Braverman\thanks{Department of Computer Science, Princeton University. Research supported in part by the NSF Alan T. Waterman Award, Grant No. 1933331.}\qquad\quad
Subhash Khot\thanks{Department of Computer Science, Courant Institute of Mathematical Sciences, New York University. Research supported by NSF Award CCF-2515155, and the Simons Investigator Award.}\\[1.5ex]
Yang P. Liu\thanks{Department of Computer Science, Carnegie Mellon University.}\qquad\quad
Dor Minzer\thanks{Department of Mathematics, Massachusetts Institute of Technology. Research supported by NSF CCF award 2227876 and NSF CAREER award 2239160.}\qquad\quad
Kunal Mittal\thanks{Department of Computer Science, Courant Institute of Mathematical Sciences, New York University. Research supported by NSF Award CCF-2007462, and Simons Investigator Awards to Subhash Khot and Ran Raz.}}
\date{}			 					
\maketitle

\begin{abstract}
Let $\mathcal{G}$ be a $k$-player game with value $<1$, whose query distribution
is such that no marginal on $k-1$ players admits a non-trivial Abelian embedding.
We show that for every $n\geq N$, the value of the $n$-fold parallel repetition of $\mathcal{G}$ is
$$ \text{val}(\mathcal{G}^{\otimes n}) \leq \frac{1}{\underbrace{\log\log\cdots\log}_{C\text{ times}} n}, $$
where $N=N(\mathcal{G})$ and $1\leq C\leq k^{O(k)}$ are constants. As a consequence, we obtain a parallel repetition theorem for all $3$-player games whose query distribution is pairwise-connected.
Prior to our work, only inverse Ackermann decay bounds were known for such games~\cite{V96}.

As additional special cases, we obtain a unified proof for all known parallel repetition theorems, albeit with weaker bounds:
\begin{enumerate}
    \item A new analytic proof of parallel repetition for all 2-player games \cite{Raz98,Hol09,DS14}.
    \item A new proof of parallel repetition for all $k$-player playerwise connected games \cite{DHVY17,GHMRZ22}.
    \item Parallel repetition for all $3$-player games (in particular $3$-XOR games) whose query distribution has no non-trivial Abelian embedding into $(\mathbb{Z}, +)$ \cite{BKM23ghz,BBKLM25}.
    \item Parallel repetition for all 3-player games with binary inputs~\cite{HR20,GHMRZ21, GHMRZ22, GMRZ22}.
\end{enumerate}
\end{abstract}

\newpage

\vspace*{-3.5\baselineskip}
\tableofcontents

\newpage

\section{Introduction}
\label{sec:intro}

In a $k$-player game, a verifier samples a \emph{question} $X = (X^1, \dots, X^k)$ from a distribution $Q$ over $\mc{X} := \mc{X}^1 \times \dots \times \mc{X}^k$, where each $\mc{X}^i$ is a finite set, and gives $X^i$ to the $i$\textsuperscript{th} player. Then, for each $i \in \{1, \dots, k\}$, the $i$\textsuperscript{th} player answers with $A^i := f^i(X^i)$ for some function $f^i: \mc{X}^i \to \mc{A}^i$, and sends this to the verifier. The values $A^i\in \mc{A}^i$ are known as the \emph{answers}, and let $\mc{A} := \mc{A}^1 \times \dots \times \mc{A}^k$. The verifier accepts if $V((X^1, \dots, X^k), (A^1, \dots, A^k)) = 1$, where $V: \mc{X} \times \mc{A} \to \{0, 1\}$ is a predicate known to all players. This defines a game $\mc{G} = (\mc{X}, \mc{A}, Q, V)$, and we define the game's value, denoted $\val(\mc{G})$, as the maximum acceptance probability (with respect to the distribution $Q$) over all possible player strategies. See Definitions \ref{defn:multiplayer_games}, \ref{defn:game_value} for more precise definitions.

A natural question that arises is: how does the value of a game behave under parallel repetition? The \emph{$n$-fold parallel repetition} of $\mc{G}$, which we denote as $\mc{G}^{\otimes n}$, is the game where the verifier samples $n$ questions $(X_1, \dots, X_n) \sim Q^{\otimes n}$ independently, and sends the $i$\textsuperscript{th} coordinate of all of $(X_1, \dots, X_n)$ to the $i$\textsuperscript{th} player simultaneously; each player then needs to respond $n$ answers, one for each instance of the game. The players win the repeated game if they win each one of the $n$ instances.
See Definition~\ref{defn:game_parrep} for a more precise definition.

The parallel repetition of two-player games is by now well understood.
Originally proposed by~\cite{FRS94} as a means of amplifying the advantage in interactive protocols, it is now known that the value of a two-player game decays exponentially under parallel repetition whenever $\val(\mc{G}) < 1$.
This exponential decay was first established by Raz~\cite{Raz98} using information-theoretic techniques. Subsequent works have simplified Raz’s proof and strengthened the quantitative bounds~\cite{Hol09,BRRRS09, Rao11,RR12,DS14,BG15}.
With the exception of~\cite{DS14}, most of these works follow the information-theoretic framework introduced by Raz. The work of~\cite{DS14}, in contrast, introduces an analytic approach applicable in the case where the game~$\mathcal{G}$ is a projection game (which is the case in the majority of applications, especially those pertaining to probabilistically-checkable-proofs).

Although it might appear plausible that the na\"{i}ve bound $\val(\mathcal{G}^{\otimes n}) \le \val(\mathcal{G})^n$ should hold, this turns out to be false~\cite{For89,Fei91,FV02,Raz11}.
This failure is also well understood: it is intimately connected to the geometry of high-dimensional Euclidean tilings (see~\cite{FKO07,KORW08,AK09,BMtilesym}), and in particular to the existence of bodies of volume~1 and surface area~$\Theta(\sqrt{n})$ that tile~$\mathbb{R}^n$.

Parallel repetition of two-player games has found numerous applications across several domains, including interactive proofs~\cite{BGKW88}, communication complexity~\cite{PRW97,BBCR13,BRWY13}, quantum information~\cite{CHTW04,BBLV13}, and hardness of approximation~\cite{FGLSS96,ABSS97,ALMSS98,AS98,BGS98,Fei98,Has01,Khot02a,Khot02b,GHS02,DGKR05,DRS05}.
The reader is referred to the survey~\cite{Raz10} for more details.

Parallel repetition of $k$-player games for $k \ge 3$ is much more poorly understood. Even for $k = 3$, the best general bound on $\val(\mc{G}^{\otimes n})$ by \cite{V96} is very weak: approximately $(\log^{\ast q} n)^{-1}$, where $q = \abs{\supp(Q)}$ is the number of questions in $\mc{G}$, and $\log^{\ast q} n$ is defined recursively as the number of times one needs to apply $\log^{\ast (q-1)}$ starting from $n$ to get down to a constant. This was proven by directly invoking the Density-Hales-Jewett theorem \cite{FK91, DHJ12}. More recently, there has been renewed interest in proving parallel repetition theorems for restricted classes of $k$-player games for $k \ge 3$. \cite{DHVY17} used information-theoretic techniques to prove parallel repetition theorems with exponential decay, i.e., $\val(\mc{G}^{\otimes n}) \le \exp(-\Omega(n))$ for connected games (see Definition \ref{defn:connection_graph}), which are those where the following graph on question tuples is connected: draw an edge between $X, X' \in \supp(Q)$ if $X$ and $X'$ differ in a single coordinate. After this, a series of works \cite{HR20,GHMRZ21,GHMRZ22, GMRZ22} established parallel repetition theorems for all $3$-player games with binary inputs, i.e., $\mc{X}^1 = \mc{X}^2 = \mc{X}^3 = \{0, 1\}$, with polynomial decay, i.e., $\val(\mc{G}) \le n^{-\Omega(1)}$. These works contained three key technical points (among many others) that we mention. First, the work \cite{GHMRZ22} proved polynomial decay for all $k$-player playerwise connected games, which are games where the projection of the connectivity graph (defined in Definition \ref{defn:coord_conn_graph}) to single coordinates are all connected graphs. Second, this was combined with previous works \cite{HR20,GHMRZ21} studying a particular GHZ game, whose question distribution is supported on $\{(x_1, x_2, x_3) \in \{0,1\}^3 : x_1 + x_2 + x_3 \equiv 0 \pmod{2} \}$, and is not playerwise connected (because the connectivity graph has no edges); more recently, parallel repetition with exponential decay was established for the GHZ game \cite{BKM23ghz} and was even extended to all 3-XOR games \cite{BBKLM25, BBKLM25} with a certain assumption on the underlying distribution.
Third, the works~\cite{GHMRZ22, GMRZ22}, proved a polynomial decay bound for games whose questions are supported on $\set{(0,0,1),(0,1,0),(1,0,0)}\subseteq \set{0,1}^3$; such games are intimately connected to the length-3 density Hales-Jewett problem, bounds for which have been significantly improved very recently~\cite{BKLM24c}.

Proving improved bounds on the parallel repetition of multiplayer games can lead to several interesting applications.
It is known that a strong parallel repetition theorem for a certain class of multiplayer games implies super-linear lower bounds for non-uniform Turing machines~\cite{MR21}.
Also, parallel repetition of multiplayer games in the large alphabet regime is equivalent to many important problems in high-dimensional extremal combinatorics, like the density Hales-Jewett problem, and that of square free sets in finite fields~\cite{FV02, HHR16, Mit25}.
Additionally, as mentioned in~\cite{DHVY17}, it is suspected that the techniques used to prove multiplayer parallel repetition bounds may lead to an improved understanding of multipary communication complexity in the number-on-forehead (NOF) model, a problem intimately connected to circuit lower bounds. 

The primary goal of this work is to present a new analytic framework for proving parallel repetition theorems.
This framework utilizes recent inverse theorems for $k$-wise correlations over high-dimensional distributions (e.g., from \cite{BKLM24a,BKLM24b}). While our main theorem requires a certain assumption on the input question distribution, it is sufficiently general and is able to reprove all currently known parallel repetition theorems, and much more. However, we get bounds that are weaker than the best known ones in many cases (a finite number of repeated logarithms), but still superior to the general bounds obtained by Verbitsky \cite{V96}.

\subsection{Our Results}
\label{subsec:results}

To state our main result, we need to introduce a few notions to describe the assumptions we make on the underlying input distribution.
A distribution is pairwise connected if the support of its projection to any two coordinates forms a connected bipartite graph:
\begin{definition}\label{defn:pair_conn} (Pairwise-Connected)
	Let $k\in \N$, let $\Sigma_1,\dots,\Sigma_k$ be finite sets, and let $\mc S\subseteq \Sigma_1\times\dots\times\Sigma_k$.
	We say that $\mc S$ is \emph{pairwise-connected} if for every $1\leq i<j\leq k$, the bipartite projection graph $\mc S_{i,j}$ is connected: this has vertex set $\Sigma_i \cup \Sigma_j$, and edge set \[\set{(x_i,x_j): \exists y\in \prod_{t\in [k]\setminus\set{i,j}} \Sigma_t,\ (x_i,x_j, y)\in \mc S} \subseteq \Sigma_i\times \Sigma_j.\]
	We say that a distribution $\mu$ over $\Sigma_1\times\dots\times\Sigma_k$ is pairwise-connected if $\supp(\mu)$ is pairwise-connected.
\end{definition}

We define the notion of Abelian embeddings, and games/distributions that have no-Abelian-embeddings and no-marginal-Abelian-embeddings.

\begin{definition}\label{defn:ab_emb} (Abelian Embeddings)
	Let $k\in \N$, let $\Sigma_1,\dots,\Sigma_k$ be finite sets, and let $\mc S\subseteq \Sigma_1\times\dots\times\Sigma_k$.
	An Abelian embedding of $\mc S$ is a tuple $(G,\sigma_1,\dots,\sigma_k)$, where $G$ is an Abelian group, and $\sigma_i:\Sigma_i\to G$, $i\in [k]$ are mappings such that for each $(a_1,\dots,a_k)\in \mc S$, it holds that $\sum_{i=1}^k\sigma_i(a_i)=0_G$.
	Such an Abelian embedding is called non-trivial if not all of the maps $\sigma_i$ are constant.
	
	We say that $\mc S$ has no-Abelian-embeddings if it admits no non-trivial Abelian embedding.
    Similarly, we say that a distribution $\mu$ over $\Sigma_1\times\dots\times\Sigma_k$ has no-Abelian-embeddings if $\supp(\mu)$ has no-Abelian-embeddings.
\end{definition}

\begin{definition} (Marginal Abelian Embeddings)\label{defn:pair_conn_marginal}
	Let $k\in \N$, let $\Sigma_1,\dots,\Sigma_k$ be finite sets, and let $\mu$ be a distribution over $\Sigma_1\times \dots \times \Sigma_k$ such that:
	\begin{enumerate}
		\item $\mu$ is a pairwise-connected distribution (see Definition~\ref{defn:pair_conn}).
		\item There exists two distinct indices $i_1,i_2 \in [k]$, such that the $(k-1)$-marginals $\mu_{-i_1}$ and $\mu_{-i_2}$ admit no-Abelian-embeddings (see Definition~\ref{defn:ab_emb}).
	\end{enumerate}
	Then, we say that $\mu$ is pairwise-connected with no-marginal-Abelian-embeddings.
\end{definition}

Our main theorem is a parallel repetition theorem for distributions with no-marginal-Abelian-embeddings. In the next section, we give several applications.
\begin{theorem}
\label{thm:main} \label{thm:pairconn_parrep}
	Let $\mc G = (\mc X, \mc A, Q, V)$ be a $k$-player game with $\val(\mc G)<1$, such that the distribution $Q$ is pairwise-connected with no-marginal-Abelian-embeddings (see Definition~\ref{defn:pair_conn_marginal}).
	Then, there exists a constant $C\in \N,\ C\leq k^{O(k)}$, such that for every sufficiently large $n\in \N$, \[\val(\mc G^{\otimes n}) \leq \frac{1}{\log \log\cdots\log n},\]
	where the number of logarithms is $C$.
\end{theorem}

To understand the above definitions and our main theorem better, we make a few simple remarks:

\begin{remark}\label{remark:marginal_pairconn_defn}
Let $\mc S\subseteq \Sigma_1\times \dots \times \Sigma_k$, and let $\mu$ be a distribution with support $\mc S$.
Then,
\begin{enumerate}
    \item Suppose that $\mc S$ has no-Abelian-embeddings, then for every $j\in [k]$, the \emph{projection} of $\mc S$ onto coordinate $j$ equals $\Sigma_j$.
    Otherwise, a non-trivial Abelian embedding is obtained as follows: the map $\sigma_j$ assigns arbitrary non-zero values to elements of $\Sigma_j$ that never occur in $\mc S$, and everything else maps to zero.

    \item\label{remtim} Suppose that $\mu$ is not pairwise connected, then it admits a non-trivial Abelian embedding into any (non-trivial) Abelian group $G$, as follows: Suppose for some $i,j\in [k]$, there are partitions $\Sigma_i = \Sigma_i^{(1)}\sqcup \Sigma_i^{(2)},\ \Sigma_j = \Sigma_j^{(1)}\sqcup \Sigma_j^{(2)}$ such that the graph $\mc S_{i,j}$ has edges contained in $\brac{\Sigma_i^{(1)}\times \Sigma_j^{(1)}} \cup \brac{\Sigma_i^{(2)}\times \Sigma_j^{(2)}}$.
    Then, the map $\sigma_i$ assigns a value $g\not=0$ to each element of $\Sigma_i^{(1)}$, and  $\sigma_j$ assigns value $-g$ to each element of $\Sigma_j^{(1)}$; everything else maps to zero.
    
    \item Suppose that $\mu$ is such that some marginal $\mu_T$, for $T\subseteq [k]$, has a non-trivial Abelian embedding. Then, $\mu$ also has a non-trivial Abelian embedding.
    This simply follows by extending the embedding by choosing $\sigma_j$ as the zero map for $j\not\in T$.

    \item The two points above imply that if $\mu$ has no-Abelian embeddings, then it is pairwise-connected with no-marginal-Abelian-embeddings.
    Hence, Theorem~\ref{thm:main} gives bounds for all games with no-Abelian-embeddings; this special case is also proven directly in Section~\ref{sec:ab_emb}.

    \item Any set $\mc S$ that is connected (Definition~\ref{defn:connection_graph}), or even coordinate-wise connected (see Definition~\ref{defn:coord_conn_graph}) has no-Abelian-embeddings. This follows from Lemma~\ref{lemma:conn_implies_nae}.
    
	\item With the above, we get that for $k=3$, Definition~\ref{defn:pair_conn_marginal} is equivalent to $\mu$ being pairwise-connected.
    Hence, Theorem~\ref{thm:main} gives bounds for all 3-player pairwise-connected games.
    
	\item Note that the second condition in Definition~\ref{defn:pair_conn_marginal} holds if $\mu_{-i}$ admits no-Abelian-embeddings for each $i\in[k]$.
    In this case, by the above observations, it also follows that $\mu$ is pairwise-connected, making that condition redundant in the definition.
    
    A slightly more careful analysis also shows that pairwise-connectivity follows if three of the $(k-1)$-marginals $\mu_{-i}$ admit no-Abelian-embeddings (unlike two, as required in Definition~\ref{defn:pair_conn_marginal}).
\end{enumerate}
\end{remark}

We also show a simple (and naturally arising) game for which our theorem improves upon the state-of-the-art.

\begin{example} (Random 3-CNF Game; \cite[Example 1.5]{GHMRZ22})
Consider a random 3-CNF formula $\varphi = (C_1,C_2,\dots,C_m)$, with $m$ clauses over $d$ variables, generated as follows: Sample each clause independently and uniformly from the set of all $(2d)^3 = 8d^3$ clauses; that is, each clause contains three variables, either negated or not, each chosen randomly.

For a fixed formula $\varphi$, we consider the following 3-player game ${\mc G}_{\varphi}$.
The verifier samples $r\in [m]$ uniformly at random, and gives variables corresponding to the literals in $C_r$ to the three players respectively, with each player getting one of the three variables.
The players answer back with a value for the variable they get, and the verifier accepts if these values satisfy the clause $C_r$.

It was shown in~\cite{GHMRZ22} that with high probability:
\begin{enumerate}
    \item If $m=\omega(d)$, the value of the game ${\mc G}_{\varphi}$ is close to 7/8, and hence $< 1$.
    \item If $m=\omega(d^2\log d)$, the game ${\mc G}_{\varphi}$ is connected, and its parallel repetition has exponential decay~\cite{DHVY17}.
    On the other hand, if $m=o(d^2)$, the game is not connected.

    \item If $m=\omega(d^{1.5}\sqrt{\log d})$, the game ${\mc G}_{\varphi}$ is playerwise connected, and its parallel repetition has polynomial decay~\cite{GHMRZ22}.
    On the other hand, if $m=o(d^{1.5})$, the game is not playerwise connected.
\end{enumerate}
Our Theorem~\ref{thm:main} can be used to close the gap above, and give an effective parallel repetition bound when $m=\omega(d\log d)$.
This follows by observing that the game ${\mc G}_{\varphi}$ is pairwise connected when $m=\omega(d\log d)$:
For any two players, the bipartite projection graph (Definition~\ref{defn:pair_conn}) is simply a bipartite graph on $[d]\times [d]$ obtained by choosing $m$ edges uniformly and independently; when $m=\omega(d\log d)$, this is connected with high probability~\cite{Pal64}.
\end{example}

\subsection{Special Cases}
\label{subsec:special}

In this section, we explain how our parallel repetition theorem (Theorem \ref{thm:main}) applies in the setting of all previously known parallel repetition theorems with ``effective bounds'', i.e., not $\log^*$ type.

\paragraph{2-player games and $k$-player connected games.} \cite{DHVY17} studied parallel repetition for connected games---this generalizes the case of two-player games because we can assume without loss of generality that any two-player game is connected when proving parallel repetition theorems.
By Lemma~\ref{lemma:conn_implies_nae} and Remark~\ref{remark:marginal_pairconn_defn}, it follows that every connected game has no-Abelian-embeddings, and is hence pairwise-connected with no-marginal-Abelian-embeddings.
Thus, Theorem~\ref{thm:main} applies to such games.

\paragraph{Playerwise connected games.} These games were studied by \cite{GHMRZ22} as a step towards proving parallel repetition theorems for all $3$-player games with binary inputs.
It follows by Lemma~\ref{lemma:conn_implies_nae} and Remark~\ref{remark:marginal_pairconn_defn} that all playerwise connected games have no-Abelian-embeddings, and are hence pairwise-connected with no-marginal-Abelian-embeddings. Thus, Theorem \ref{thm:main} applies to such games.

\paragraph{Pairwise-connected $3$-player games.} \cite{BBKLM25}, building off \cite{BKM23ghz}, established parallel repetition theorems for $3$-XOR games whose question distributions have no non-trivial Abelian embeddings into the integers under addition, i.e., $(\mb{Z}, +)$.
By Remark~\ref{remark:marginal_pairconn_defn}, all such games are pairwise-connected, and our Theorem \ref{thm:main} in fact applies to \emph{all} $3$-player games that are pairwise connected.

\paragraph{$3$-player games with binary inputs.} The work \cite{GHMRZ22} established parallel repetition theorems with polynomial decay for all $3$-player games whose question distribution was binary, i.e., $\mc{X} = \{0, 1\}^3$. Our Theorem \ref{thm:main} recovers this result, with weaker bounds. Indeed, by Remark~\ref{remark:marginal_pairconn_defn}, it suffices to argue that we can reduce to the case of pairwise-connected games. Towards this, let $\mc{S}$ be the support of $Q$, and without loss of generality consider the projection of $\mc{S}$ onto the first two coordinates. It has between $0$ and $4$ edges. The cases of $3$ and $4$ edges lead to connected graphs. $0$ edges is impossible because $|\mc{S}| \ge 1$. In the case of one edge, both players 1 and 2 know their input deterministically, and thus the number of players can be reduced. In the case of $2$ edges, either they share an endpoint or not. If they share an endpoint, again some player knows their input deterministically, and thus the number of players can be reduced. If the edges don't share an endpoint, they must either be $\{(0,0), (1,1)\}$ or $\{(0,1), (1,0)\}$. In the former case, players 1 and 2 have the same input, so they can be merged into a single player. In the latter case, they always have opposite inputs, so they can be merged into a single player again (eg. imagine always negating the input of player 2, so that now they have the same input as player 1). Thus, Theorem \ref{thm:main} establishes parallel repetition theorems for all $3$-player games over binary alphabets.

\subsection{Organization}

In Section~\ref{sec:overview}, we give an overview of our proofs.
In Section~\ref{sec:prelim}, we establish some preliminaries.
Then, in Section~\ref{sec:ab_inv_rr} and Section~\ref{sec:pair_inv_rr}, we state the CSP inverse theorems relevant to this work, define generalized random restrictions, and prove several results regarding these.
In Section~\ref{sec:multiplayer_games}, we formally define the notions of multiplayer games and parallel repetition.
Then, in Section~\ref{sec:ab_emb}, we prove parallel repetition for games with no-Abelian-embeddings, an important special case of our main theorem.
Finally, in Section~\ref{sec:paircon_parrep}, we prove our main theorem.

\section{Overview}
\label{sec:overview}

In this section, we give a proof overview for our main result: parallel repetition for all $k$-player games which are pairwise connected and such that any projection to $k-1$ coordinates have no-Abelian-embeddings. We start by giving the general setup for our parallel repetition proofs. Then, we discuss the main inverse theorems from prior works underlying our proofs; this is the only place where the structure of the support of the query distribution is used. Then, we overview the proof in the simplified case where the distribution itself has no-Abelian-embeddings.
Finally, we discuss the additional pieces required to obtain the main result.

\subsection{Setup for Parallel Repetition Proofs}
\label{subsec:setup}

The information-theoretic proofs of parallel repetition, such as the ones for two-player games \cite{Raz98,Hol09}, at a very high-level, take the following approach. 
Consider the game $\mc G^{\otimes n}$, and for each $i\in [n]$, let $\Win_i$ be the event that the $i$\textsuperscript{th} coordinate is won.
Observe that for any permutation $i_1,i_2,\dots,i_n$ of $[n]$, we can write
\[\Pr[\Win_1 \land  \Win_2 \land\cdots\land  \Win_n] =  \prod_{k=1}^n  \Pr[\Win_{i_k}\mid\Win_{i_1}\land\cdots\land\Win_{i_{k-1}}]. \]
We know that $\Pr[\Win_i] \le \val(\mc{G}) < 1$ for all $i$, where the probability is over the random questions. Now, let $E$ be the event $\Win_1$. From here, it would be very natural to prove that there is some $i \in \{2, \dots, n\}$ such that $\Pr[\Win_i \mid E] < 1$ still, and even stronger, that
\[ \Pr[\Win_i \mid E] \le \val(\mc{G}) + o_n(1). \]
If this were true, we could condition on $\Win_i$ for another coordinate $i$, and continue inductively. Roughly speaking, the key lemma in these approaches takes the following form:

\paragraph{Claim:} For any event $E$ with $\Pr[E] \ge \alpha$, there is some coordinate $i \in [n]$ such that \[ \Pr[\Win_i \mid E] \le 1 - \Omega_{\mc{G}}(1). \]

The rate of decay in the resulting parallel repetition theorem then depends on the smallest $\alpha$ (in terms of $n$) for which we can establish such a \textbf{Claim} (we include a more precise version of this discussion in Lemma \ref{lemma:prod_set_hard_coor}).
For example, \cite{Raz98,Hol09} establishes such a claim for 2-player games, as long as $\alpha \ge \exp(-\Omega(n))$, which leads to an exponential decay rate. In our theorems, one should think of $\alpha \approx \frac{1}{\log\log \cdots\log n}$, where the number of logarithms is some large constant depending on the game $\mc G$. The remainder of this overview is devoted to explaining how to establish the \textbf{Claim} in our setting, for this choice of $\alpha$.

\subsection{CSP Inverse Theorems}
\label{subsec:cspinverse}

Our results rely on certain theorems about correlations of functions under $k$-ary $n$-dimensional distributions. Formally, let $\Sigma_1, \dots, \Sigma_k$ be finite sets and let $\mu$ be a distribution over $\Sigma_1 \times \dots \times \Sigma_k$. Let $\brac{f_i: \Sigma_i^n \to \mb{C}}_{i\in [k]}$ be $1$-bounded functions (that is, $|f_i(x)| \le 1$ for all $x \in \Sigma_i^n$). The reader should think of the $\Sigma_i$ and $\mu$ as being fixed while sending $n \to \infty$. A very general question considered by previous works was: under what conditions on $\mu$ and $f_i$ can we guarantee that
\begin{equation}
\left|\mb{E}_{(x_1, \dots, x_k) \sim \mu^{\otimes n}}\left[f_1(x_1) \dots f_k(x_k) \right] \right| \leq o_n(1)? \label{eq:kcsp}
\end{equation}
This question, naturally, has applications to analyzing dictatorship tests, property testing, and additive combinatorics. A recent line of work has partially answered this question for increasingly more general distributions $\mu$ \cite{BKM1,BKM2,BKM3,BKM4,BKM5}. For the purposes of this work, we only require the following theorems proven in \cite{BKLM24b}.

\paragraph{Projections and pairwise connectivity.} Before stating the hypotheses on $\mu$, we introduce some basic notions pertaining to the connectivity of $\mu$. Towards this, given a distribution $\mu$ on $\Sigma_1 \times \dots \times \Sigma_k$ and a subset $S \subseteq [k]$, we can naturally define the \emph{projection} of $\mu$ onto $S$ as the distribution on $\prod_{s \in S} \Sigma_s$ where we simply restrict the coordinates of $\mu$ onto $S$. Now, we say that $\mu$ is \emph{pairwise connected} if the support of $\mu_S$ for any $|S| = 2$ forms a connected graph; see Definition~\ref{defn:pair_conn} for a formal definition.

\paragraph{Distributions with no-Abelian-embeddings.} For a formal definition of Abelian embeddings, the reader is referred to Definition~\ref{defn:ab_emb}, although it is not strictly necessary for this overview. \cite{BKLM24b} proves that if $\mu$ has no-Abelian-embeddings, then \eqref{eq:kcsp} holds unless all the functions $f_i$ have non-negligible noise stability, i.e., each of the $f_i$ has non-negligible Fourier mass on low-degree terms. This is formally stated in Theorem \ref{thm:bklm_inverse}.

\paragraph{Distributions with no-marginal-Abelian-embeddings.} We say that a distribution $\mu$ has no-marginal-Abelian-embeddings if any projection onto $k-1$ coordinates has no-Abelian-embeddings; note that this is more general than the above definition (see Remark~\ref{remark:marginal_pairconn_defn}).\footnote{Everything we say going forward will work with the slightly more general Definition~\ref{defn:pair_conn_marginal} as well.} In this context, \cite{BKLM24b} proved that if \eqref{eq:kcsp} fails for such $\mu$, then each of the $f_i$ satisfies the following property: there is a product function $P_i$ such that $f_iP_i$ has non-negligible Fourier mass on low-degree terms; here, a product function $P_i:\Sigma_i^n\to\C$ takes the form $P_i(x_1,x_2,\dots,x_n) =  P_{i,1}(x_1)\cdot P_{i,2}(x_2)\cdots P_{i,n}(x_n)$, i.e., it is the product of functions depending only on a single coordinate. Later, we call such functions not \emph{product pseudorandom} (see Definition \ref{defn:prod_pseudo}). The precise statement of this result is given in Theorem \ref{thm:pair_conn_inv}.

\subsection{Games with No Abelian Embeddings}
\label{subsec:noabelian}

In this section, we overview a simplified version of Theorem \ref{thm:main}, for games whose query distribution $Q$ has no-Abelian-embeddings; this appears as Theorem \ref{thm:noabemb_parrep}. Recall from the discussion in Section \ref{subsec:setup} that our goal is to prove that given some event $E$ with $\Pr[E] \ge \alpha$ for $\alpha$ not too small, there is some coordinate $i\in [n]$ with $\Pr[\Win_i \mid E] \le \val(\mc{G}) + o(1)$. It turns out that we can assume that $E$ is a product event $E = E^1 \times \dots \times E^k$, i.e., the product of events $E^j$ that depend only on player $j$.\footnote{In Section~\ref{subsec:setup}, $E$ was the event of winning a few coordinates; roughly speaking, we can choose to work with the event $E$ which fixes some winning questions and answers in these coordinates, and this is a product event.}

A common strategy used in prior works to establish such a claim is to use an \emph{embedding strategy}: prove that the players can obtain value close to $\Pr[\Win_i\mid E]$ on using some honest strategy on a single copy of the game.
Then, this quantity must then be bounded by $\val(\mc G)$, as desired.
Here, we shall consider the following randomized strategy for a single coordinate of the game $\mc G$.
Let $\tilde{X} \sim Q$ be the query in a single copy of the game $\mc G$; the players think of this as the input in coordinate $i$ of the game $\mc G^{\otimes n}$, and do the following: each player $j \in [k]$ randomly fills out the remaining questions for coordinates $i' \neq i$ conditioned on event $E^j$, and then outputs based off of their strategy for $\mc G^{\otimes n}$ in coordinate $i$.

We want to show that the success probability of the above strategy is close to $\Pr[\Win_i\mid E]$.
For this, we first try to analytically express the probability 
$\Pr[\Win_i\mid E]$.
Towards this, we define functions $F^{j,\tilde{X}^j}: (\mc{X}^j)^{\otimes (n-1)} \to \{0, 1\}$ as follows. To define $F^{j,\tilde{X}^j}(X)$, think of $X$ as being assignments to the $n-1$ coordinates of player $j$'s input other than $i$. Then, $F^{j,\tilde{X}^j}(X) = 1$ if filling in the $i$\textsuperscript{th} coordinate with $\tilde{X}^j$ gives an input in $E^j$, and otherwise $F^{j,\tilde{X}^j}(X) = 0$. Also, for answers $\tilde{A} \in \mc{A}$, define the functions $F^{j, \tilde{X}^j, \tilde{A}^j}: (\mc{X}^j)^{\otimes (n-1)} \to \{0, 1\}$ as follows: $F^{j, \tilde{X}^j, \tilde{A}^j}(X) = 1$ if any only if $F^{j,\tilde{X}^j}(X) = 1$ and when $X$ is filled in with $\tilde{X}^j$ in coordinate $i$, player $j$'s strategy outputs $\tilde{A}^j$. Note that these are precisely the functions defined in Definition~\ref{defn:abemb_fns_pseudo} and Definition~\ref{defn:pairwise_fns_to_make_pseudo}. Then, the quantity $\Pr[\Win_i\mid E]$, which is winning probability of coordinate $i$ when conditioned on $E$, roughly\footnote{For this to be exact, $\tilde{X}$ must also be drawn conditioned on $E$.} equals
\begin{equation}
\E_{\tilde{X} \sim Q} \sum_{\tilde{A} : V(\tilde{X}, \tilde{A}) = 1} \sqbrac{\frac{\ \E_{(X^1, \dots, X^k) \sim Q^{\otimes(n-1)}}\sqbrac{F^{1,\tilde{X}^1, \tilde{A}^1}(X^1)\cdot F^{2,\tilde{X}^2, \tilde{A}^2}(X^2) \cdots F^{k,\tilde{X}^k, \tilde{A}^k}(X^k) }}{\E_{(X^1, \dots, X^k) \sim Q^{\otimes(n-1)}}\sqbrac{F^{1,\tilde{X}^1}(X^1)\cdot F^{2,\tilde{X}^2}(X^2) \cdots F^{k,\tilde{X}^k}(X^k) }}}.
\label{eq:winningprob}
\end{equation}
Observe that for fixed $\tilde{X}$ and $\tilde{A}$, the expressions in the numerator and denominator are all $k$-wise correlations over $Q^{\otimes(n-1)}$, where $Q$ by assumption is a distribution with no-Abelian-embeddings. Thus, by our above discussion (also see Theorem \ref{thm:bklm_inverse}), the expressions in the numerator and denominator behave as if the $k$ players are acting independently/honestly (that is, according to the honest strategy described above), as long as all the functions $F^{j,\tilde{X}^j}$ and $F^{j, \tilde{X}^j,\tilde{A}^j}$ have low noise stability, i.e., almost all their mass is on high-degree Fourier coefficients (except the constant term).
Hence, if this high-degree property holds for all the functions, we are done.

Thus, we have reduced the problem to ensuring that each $F^{j,\tilde{X}^j}$ and $F^{j, \tilde{X}^j,\tilde{A}^j}$ has low noise stability; note that the number of such functions is a constant (depending on $\mc G$). To achieve this, we apply random restrictions to the overall game. Precisely, Lemma \ref{lemma:random_rr_noise_stab} states that taking a random restriction with a randomly chosen probability makes any small set of functions high-degree, with high probability. Thus, in our proof, we first apply such a random restriction to the game (note that a random restriction of a game is still a valid game with many coordinates) to make the relevant functions $F^{j,\tilde{X}^j}$ and $F^{j, \tilde{X}^j,\tilde{A}^j}$ have low noise stability, and then use the inverse theorem to argue that after this the players are acting as if they have no communication, and thus $\Pr[\Win_i \mid E] \le \val(\mc{G})+o(1)$ as desired.

As a technical step, we need to prove that the questions to the players in coordinate $i$ have distribution essentially the same as $Q$, even when conditioning on $E$ and a typical random restriction as above.
Very roughly speaking, this holds because the random restriction leaves a large number of coordinates free.
We carry out this step by an information-theoretic argument, as in~\cite{Raz98, Hol09}; see Lemma~\ref{lemma:lemma:distr_change_in_i_fixp} for formal details.
We note however that the use of information theory in this step is only for simplicity, and not necessary at all---in the proof of our actual main theorem, this step becomes more far more complicated, and we use an analytic argument with no information theory.

\subsection{Games with No Marginal Abelian Embeddings}
\label{subsec:nomarginal}

We now describe the modifications needed to extend our the discussion in the above section to our main theorem (Theorem \ref{thm:main}) about games with no-marginal-Abelian-embeddings. All the discussion up through \eqref{eq:winningprob} proceeds identically. In the case of no-Abelian-embeddings, Theorem \ref{thm:bklm_inverse} said that the expressions in the numerator and denominator of \eqref{eq:winningprob} behave as if the $k$ players are independent as long as the function $F^{j,\tilde{X}^j}$ and $F^{j, \tilde{X}^j,\tilde{A}^j}$ have low noise stability. However, in the more general setting, this is no longer true and what we require instead is that the functions are \emph{product pseudorandom} (see Definition \ref{defn:prod_pseudo}), i.e., the product of $F^{j,\tilde{X}^j}$ and any product-function has low noise stability. This is much more difficult to ensure and leads to technical complications.

Fortunately, previous work \cite{BKLM24a,BKLM24c} introduced a type of random restriction that can handle functions that are not product pseudorandom. In this work, we call this a \emph{generalized random restriction} (see Definition \ref{defn:grr}), which takes the following form: some coordinates are randomly restricted, and the remaining coordinates are partitioned into sets $T_1, T_2, \dots, T_m \subseteq [n]$. Now, we force that all coordinates in a set $T_i$ take the same value. Note that this naturally takes an $n$-dimensional function and turns it into an $m$-dimensional function. Furthermore, the restriction and these sets $T_1,\dots,T_m$ are randomly chosen so as to ensure that the overall distribution on $Q^{\otimes n}$ is approximately preserved.

Ultimately, we prove that if the functions $F^{j,\tilde{X}^j}$ or $F^{j, \tilde{X}^j,\tilde{A}^j}$ are not product pseudorandom, then we can apply such a generalized random restriction to increase some $\ell_2$ energy potential function, which cannot happen forever. In this way, we ensure that the relevant functions are all product pseudorandom eventually.\footnote{For some intuition, we demonstrate a simple example where generalized restrictions are useful to make functions product-pseudorandom. Suppose $f:\set{0,1}^n\to \set{-1,1}$ is a character over $\F_2$, given by $f(x) = (-1)^{\sum_{i\in S}x_i}$, for some $S\subseteq [n]$.
If $S=[n]$, any usual random restriction of $f$, down to any number of coordinates, is still a character (possibly with a $-1$ sign), and hence not product-pseudorandom. On the other hand, there always exists a generalized random restriction, down to $m= \lfloor n^{1/3}\rfloor$ coordinates, that makes $f$ product-pseudorandom: If $\abs{S}\leq n/2$, we do a usual random restriction and fix the input $x_i$ randomly for each $i\in S$; the function after the restriction is a constant.
If $\abs{S}>n/2$, we randomly pick pairs of coordinates $T_1=\set{i_1,j_1}, T_2=\set{i_2,j_2}, \dots, T_m=\set{i_m,j_m}$ inside $S$, force both coordinates inside each $T_i$ to take the same value, and randomly restrict the coordinates in $S\setminus \cup_{i\in [m]}T_i$; after any such restriction, the function becomes a constant, since coordinates inside each $T_i$ \emph{cancel out}. Moreover, for $m= \lfloor n^{1/3}\rfloor $, as $T_i$'s were chosen randomly, the overall distribution is approximately preserved on average.
}
With this, we carry out a proof that very roughly resembles the no-Abelian-embeddings case, however, it turns out that making the argument actually work is much more technical, for reasons described next.

\paragraph{Generalized Random Restrictions.}
Working with generalized random restrictions is far more technically challenging than usual random restrictions.
They only preserve the original distribution approximately, and hence extra care is needed to analyze the error terms for each of the steps in the proof.

Moreover, it turns out that in this case, to be able to describe the honest strategy for one copy of the game, we need to be able to sample generalized random restrictions conditioned on the event $E$. A priori it is not even clear what this means, however the natural definition, via Bayes' rule, works for us.
This leads to our notion of conditional generalized random restrictions (see Definition~\ref{defn:cond_grr}).

\paragraph{Making sure random restrictions do not reveal too much information.}
A major technical hurdle that arises with the use of generalized random restrictions is that the restrictions needed depend on the set of functions we wish to make product pseudorandom.

As mentioned in the no-Abelian-embeddings case, for technical reasons, we still need to prove that the questions to the players in coordinate $i$---the coordinate we are embedding into---have distribution essentially the same as $Q$, even when conditioning on $E$ and a typical generalized random restriction.
Since the generalized random restriction here may depend both on the coordinate $i$ being considered and the event $E$, this may not even be true.\footnote{For example, suppose the inputs to the players are binary, and that $E$ says that the inputs in the first $n/2$ coordinates sum to zero (mod 2), and also that the inputs in the second $n/2$ coordinates sum to zero (mod 2). Now, for each coordinate $i$, the corresponding generalized random restriction may randomly fix the inputs to all coordinates except $i$ in the block of size $n/2$ that $i$ lies in.}
We get around this hurdle by observing that whenever this is false, it must be because the generalized random restriction, along with the question in coordinate $i$, \emph{split} the mass of $E$ unevenly. Hence, if we perform the generalized random restriction and fix question in coordinate $i$ randomly, we get an $\ell_2$ increment with respect to the mass of the set $E$, which cannot happen forever; see Section~\ref{sec:pairwise_rr_information} for formal details of this step.

\paragraph{Hard coordinates after a generalized random restriction.} Due to the aforementioned step, we don't end up proving that $\Pr[\Win_i\mid E]$ itself is small; instead, we only prove that after performing a generalized random restriction, we get many coordinates that are hard to win.
Finally obtaining a parallel repetition bound from this statement requires carrying out a very careful induction, where we alternate between steps of finding a hard coordinate and performing a generalized random restriction; see Section~\ref{sec:pairwise_induction} for formal details of this step.
\section{Preliminaries}
\label{sec:prelim}

Let $\N = \set{1,2,\dots}$ denote the set of natural numbers. For $n\in \N$, we use $[n]$ to denote the set $\set{1,2,\dots,n}$.
For any set $S\subseteq [n]$, we use $\overline{S}$ to denote its complement $[n]\setminus S$.

We use $\log(\cdot)$ to denote the function $\log_2(\cdot)$, and $\exp(\cdot)$ to denote the function $2^{(\cdot)}$.


\subsection{Probability Distributions}
We will use calligraphic letters to denote sets, capital letters to denote random variables and small letters to denote values.

Let $P$ be a distribution (over an underlying \emph{finite} set $\Omega$, which is usually clear from context). 
We use $\supp(P) = \{\omega \in \Omega : P[\omega] > 0\}$ to denote the \emph{support} of the distribution $P$.
For a random variable $X$, we use $P_X$ to denote the distribution of $X$, that is, $P_X[x] = P[X=x]$.
For random variables $X$ and $Y$, we use $P_{XY}$ to denote the joint distribution of $X$ and $Y$.

For an event $E$ with $P[E]>0$, we use $P|E$ to denote the conditional probability distribution $P$ conditioned on $E$.
Similarly, we use $(P|E)_X = P_{X|E}$ to denote the distribution of $X$ conditioned on the event $E$, given by \[P_{X|E}[x] = \frac{P\sqbrac{X=x\land E}}{P[E]}. \]
Suppose $R$ is a random variable, and $r$ is such that $P_R[r]>0$.
We will frequently use the shorthand $P_{X|r}$ to denote the distribution $P_{X|R=r}$.

Let $P_X$ and $Q_X$ be distributions over set $\mc X$.
The $\ell_1$-distance between $P_X$ and $Q_X$ is defined as 
\[\Norm{P_X-Q_X}_1 = \sum_{x\in \mc X}\norm{P_X[x]-Q_X[x]}.\]

Next, we state a useful lemma that was used in previous works on two-player parallel repetition. Informally, it says that if we condition a product distribution $P$ on some event $E$ with nonnegligible probability, then the marginals of $P$ and its conditioning are very close on average.

\begin{lemma}\cite{Raz98, Hol09}\label{lemma:pdt_distr_cond}
	Let $V=(V_1,\dots,V_n)$ be a random variable, and let $F$ be an event in some finite probability space $(\Omega, P)$.
	Suppose that $P_V = P_{V_1}\times\dots\times P_{V_n}$ is a product distribution.
	Then, we have
	\[ \frac{1}{n}\sum_{i=1}^n \Norm{P_{V_i|F}-P_{V_i}}_1 \leq \sqrt{\frac{2}{n}\log_2{\frac{1}{P[F]}}}.\]
\end{lemma}
\begin{proof}[Proof Sketch]
	It holds that
	\[ D\brac{P_{V|F}||P_V} = \sum_v P[V=v|F]\cdot \log_2\brac{\frac{P[V=v|F]}{P[V=v]}} \leq \log_2\brac{\frac{1}{P[F]}} ,\]
	where we use $D$ to denote the relative entropy (or the KL divergence).
	Also, we have
	\[ D\brac{P_{V|F}||P_V} \geq \sum_{i=1}^n D\brac{P_{V_i|F}||P_{V_i}} \geq \frac{1}{2}\cdot \sum_{i=1}^n\Norm{P_{V_i|F}-P_{V_i}}_1^2 \geq   \frac{1}{2n}\cdot \brac{\sum_{i=1}^n\Norm{P_{V_i|F}-P_{V_i}}_1}^2.\]
	The first inequality above holds as relative entropy is super-additive when the second distribution is a product distribution; the second inequality follows from Pinsker's inequality.
\end{proof}
 

\subsection{Noise Operators}

Let $(\Sigma, \nu)$ be a probability space, with $\Sigma$ a finite set.
Define an inner product on the space $L^2(\Sigma, \nu)$ as follows: for $f,g:\Sigma\to\C$, define 
\[\ip{f,g}_{\nu} = \E_{x\sim \nu}\sqbrac{f(x)\overline{g(x)}}.\]

\begin{definition} (Noise Distribution)
	For a parameter $\rho\in[0,1]$, and $y\in \Sigma$, the distribution of inputs $\rho$-correlated to $y$, denoted $y'\sim \T_{\rho}y$, is defined as follows: take $y'=y$ with probability $\rho$, and else sample $y'\sim \nu$.
\end{definition}

We also view $\T_\rho$ as a map on the space of functions $L^2(\Sigma, \nu)$, given by \[ (\T_\rho g)(y) = \E_{y'\sim \T_{\rho}y}\sqbrac{g(y')}.\]

We consider the tensorization $\T_{\rho}^{\otimes n}$ of this operator, which acts on $L^2(\Sigma^n, \nu^{\otimes n})$ by applying $\T_\rho$ independently across the $n$ coordinates;
for ease of notation, we drop the $\otimes n$ superscript and simply call this operator $\T_{\rho}$ as well.

\begin{definition} (Noise Stability)
	For a function $g:\Sigma^n\to \C$, define its $\rho$-noise-stability as
	\[ \Stab_{\rho}^\nu [g] = \ip{g,\T_\rho g}_{\nu^{\otimes n}}.\]
	This satisfies $\Stab_{\rho}^\nu [g] = \Norm{\T_{\sqrt{\rho }}g}_2^2$; in particular, $\Stab_{\rho}^\nu [g]$ is always a non-negative real number.
	We often drop the superscript $\nu$ when it is clear from context.
\end{definition}

\begin{definition} (Expectation)
	For a function $f\in L^2(\Sigma^n, \nu^{\otimes n})$, we use $\nu(f)$ to denote its expectation $\E_{x\sim \nu^{\otimes n}}[f(x)]$.
\end{definition}


\section{Abelian Embeddings, Inverse Theorems, and Random Restrictions}\label{sec:ab_inv_rr}


\subsection{Abelian Embeddings and Inverse Theorems}
\label{subsec:abeliandef}
In this section we state a result of Bhangale, Khot, Liu and Minzer~\cite{BKLM24b} which says that products of high-degree\footnote{We note that the condition $\Stab_{1-\delta}[f]\leq \delta$ in the theorem serves as a convenient proxy for $f$ to have high Fourier degree. The reader is referred to~\cite{BKLM24b} for more details.} functions behave pseudorandomly under distributions with no-Abelian-embeddings:

\begin{theorem}\label{thm:bklm_inverse} (Inverse Theorem under No-Abelian-Embeddings; \cite[Theorem 1]{BKLM24b})

Let $k\in \N$ and let $\Sigma_1,\dots, \Sigma_k$ be finite sets.
Let $\mu$ be a distribution over $\Sigma_1\times\dots\times\Sigma_k$ with no-Abelian-embeddings (see Definition~\ref{defn:ab_emb}).
Then, for every sufficiently large $n\in \N$ and every $\epsilon>0$, there exists $\delta = \delta(\epsilon)>0$, such that the following holds:

If 1-bounded functions $f_i:\Sigma_i^n\to \C$, $i\in [k]$, satisfy 
\[ \abs{\E_{(x_1,\dots,x_k)\sim \mu^{\otimes n}}\sqbrac{\prod_{i=1}^k f_i(x_i)}} \geq \epsilon, \] 
then $\Stab_{1-\delta}^{\mu_i}[f_i]\geq \delta$ for each $i\in [k]$, where $\mu_i$ is the marginal of $\mu$ on the $i$\textsuperscript{th} coordinate.

Quantitatively, when $\epsilon\leq o_n(1)$, we can take $\delta = \exp(-\exp(\cdots \exp(\epsilon^{-1})))$, where the number of exponentials is at most $k^{O(k)}$.\footnote{\label{footnote:constant_wrt_n}The quantitative dependence in~\cite{BKLM24b} is $\delta = \exp(-\exp(\cdots \exp(\epsilon^{-O_{\alpha}(1)})))$, where the number of exponentials is $C\leq k^{O(k)}$, and $\alpha\in (0,1]$ is such that $\min_{x\in \supp(\mu)}\mu[x]\geq \alpha$. For our purposes, $\alpha$ is a constant independent of $n$, and hence when $\epsilon\leq o_n(1)$, the choice $\delta = \exp(-\exp(\cdots \exp(\epsilon^{-1})))$ works, where the number of exponentials is $C+1 \leq k^{O(k)}$.}
\end{theorem}

\begin{corollary}\label{corr:bklm}
Let $k\in \N$ and let $\Sigma_1,\dots, \Sigma_k$ be finite sets.
Let $\mu$ be a distribution over $\Sigma_1\times\dots\times\Sigma_k$ with no-Abelian-embeddings (see Definition~\ref{defn:ab_emb}).
Then, for every sufficiently large $n\in \N$ and every $\delta>0$, there exists $\epsilon = \epsilon(\delta)>0$, such that the following holds:

Let $f_i:\Sigma_i^n\to [0,1]$, $i\in [k]$ be functions, and for each $i$, let $\mu_i$ be the marginal of $\mu$ on the $i$\textsuperscript{th} coordinate.
If $\Stab_{1-\delta}^{\mu_i}\sqbrac{f_i-\mu_i(f_i)} < \delta$ for each $i\in [k]$, it holds that
\[ \abs{\E_{(x_1,\dots,x_k)\sim \mu^{\otimes n}}\sqbrac{\prod_{i=1}^k f_i(x_i)} - \prod_{i=1}^k \mu_i(f_i)} \leq \epsilon. \]
Quantitatively, when $\delta\leq o_n(1)$, we can take
$\epsilon = \frac{1}{\log\log\cdots\log\brac{\frac{1}{\delta}}}$, where the number of logarithms at most $k^{O(k)}$.
\end{corollary}
\begin{proof}
	Let $\delta >0$, and let $\epsilon>0$ be such that $(\frac{\epsilon}{k},\delta)$ satisfy Theorem~\ref{thm:bklm_inverse}.
	Then, we have
	\begin{align*}
		&\abs{\E_{(x_1,\dots,x_k)\sim \mu^{\otimes n}}\sqbrac{\prod_{i=1}^k f_i(x_i)} - \prod_{i=1}^k \mu_i(f_i)}
		\\&\qquad\qquad= \abs{\sum_{t=1}^k \E_{(x_1,\dots,x_k)\sim \mu^{\otimes n}}\sqbrac{\prod_{i=1}^{t-1}\mu_i(f_i)\cdot \brac{f_t(x_t)-\mu_t(f_t)}\cdot  \prod_{i=t+1}^k f_i(x_i)}}
		\leq k\cdot \frac{\epsilon}{k} = \epsilon,
	\end{align*}
	where we first used the triangle inequality and then used Theorem~\ref{thm:bklm_inverse} for each term.
	Note that $f_t(x_t)-\mu_t(f_t)$ is 1-bounded for each $t\in [k]$.
	
	The bound on $\epsilon$ is now $\frac{k}{\log\log\cdots\log\brac{\frac{1}{\delta}}}$, and for $\delta\leq o_n(1)$, this is as desired (by possibly increasing the number of logarithms by 1).
\end{proof}

The above corollary is very useful in calculating expectations containing high-degree functions; however, in most applications, the functions we care about are not high-degree.
Next, we introduce the notion of random restrictions, which helps make arbitrary functions high-degree.


\subsection{Random Restrictions}

In this section we define the notions of a \emph{restriction} and \emph{random restriction}. A restriction takes a subset of coordinates $I \subseteq [n]$ and sets them to some fixed values.
\begin{definition}\label{defn:restriction} (Restriction)
	Let $\Sigma$ be a finite alphabet and let $n\in \N$.
	A restriction on $\Sigma^n$ is a tuple $\rho=(I,z)$ where $I\subseteq [n]$ and $z\in \Sigma^{I}$.
	
	For a function $f:\Sigma^n\to \C$ and a restriction $\rho=(I,z)$ we define the restricted function $f_{I\to z}:\Sigma^{\bar{I}}\to \C$ as $f_{I\to z}(y) = f(y,z)$; we shall also use $f_\rho$ to denote this function.
\end{definition}
The notation $I\to z$ above indicates that the variables in the set $I$ are fixed to value $z$.
A random restriction is a restriction where the subset of coordinates $I$ is chosen randomly and the values they are fixed to are chosen from the base distribution.
\begin{definition} ($p$-random restriction)
	Let $p\in [0,1],\ n\in \N$, and let $(\Sigma, \nu)$ be a probability space.
	
	We use $I\sim_p [n]$ to denote a random set $I\subseteq [n]$, where each $i\in [n]$ is included in $I$ with probability $p$ independently.
	
	A $p$-random restriction on $(\Sigma^n, \nu^{\otimes n})$ is a random tuple $\rho=(I,z)$, where $I \sim_{1-p} [n]$, and $z\sim \nu^{\otimes I}$.
	For a function $f\in L^2(\Sigma^n, \nu^{\otimes n})$, its $p$-random restriction is then the function $f_{I\to z}:\Sigma^{\bar{I}}\to \C$.
	We shall also sometimes think $f_{I\to z}:\Sigma^n\to \C$ by ignoring the input coordinates in the set $I$.
\end{definition}


\subsection{Noise Stability under Random Restrictions}
Next, we prove a simple lemma about the average noise stability of functions under random restrictions.

\begin{lemma}\label{lemma:rr_noise_stab} (Noise Stability under Random Restrictions)
	Let $f\in L^2(\Sigma^n, \nu^{\otimes n})$.
	Then, for $p,\delta\in [0,1]$, we have
	\[ \E_{I\sim_{1-p}[n],\ z\sim\nu^{\otimes I}}\sqbrac{\Stab_{1-\delta}\sqbrac{f_{I\to z}-\nu(f_{I\to z})}} = \Stab_{1-p\delta}[f]-\Stab_{1-p}[f].\]
\end{lemma}
\begin{proof}
	For $\rho = 1-\delta$, we have
	\begin{align*}
		&\E_{I\sim_{1-p}[n]} \E_{z\sim\nu^{\otimes I}}\sqbrac{\Stab_{1-\delta}\sqbrac{f_{I\to z}-\nu(f_{I\to z})}} 
		\\&\qquad\qquad = \E_{I\sim_{1-p}[n]} \E_{z\sim\nu^{\otimes I}} \ip{f_{I\to z}-\nu(f_{I\to z}),\  \T_{\rho}^{\otimes \bar{I}}\brac{f_{I\to z}-\nu(f_{I\to z})} }_{\nu^{\otimes \bar{I}}}
		\\&\qquad\qquad = \E_{I\sim_{1-p}[n]}\E_{z\sim\nu^{\otimes I}}\sqbrac{\ip{f_{I\to z},\  \T_{\rho}^{\otimes \bar{I}}f_{I\to z} }_{\nu^{\otimes \bar{I}}} - \abs{\nu(f_{I\to z})}^2 } 
		\\&\qquad\qquad = \E_{I\sim_{1-p}[n]}\E_{z\sim\nu^{\otimes I}}\sqbrac{\ip{f_{I\to z},\  \T_{\rho}^{\otimes \bar{I}}f_{I\to z} }_{\nu^{\otimes \bar{I}}} } - \E_{I\sim_{1-p}[n]} \E_{z\sim\nu^{\otimes I}}\sqbrac{\abs{\nu(f_{I\to z})}^2}.
	\end{align*}
	Note that the second term in the above expression is the same as the first term  under the substitution $\rho=0$; hence, it suffices to show that the first term equals $\Stab_{1-p(1-\rho)}[f]$.
	This is done as follows:
	\begin{align*}
	\E_{I\sim_{1-p}[n]} \E_{z\sim\nu^{\otimes I}} \sqbrac{\ip{f_{I\to z},\  \T_{\rho}^{\otimes \bar{I}}f_{I\to z} }_{\nu^{\otimes \bar{I}}} }
		&= \E_{I\sim_{1-p}[n]}  \E_{z\sim\nu^{\otimes I}}  \E_{x\sim \nu^{\otimes \bar{I}}}\E_{y\sim \T_\rho^{\otimes \bar{I}}x} f(x,z)\cdot \overline{f(y,z)}
		\\&= \E_{I\sim_{1-p}[n]} \ip{f,\ \T_{\rho}^{\otimes \bar{I}}\T_{1}^{\otimes I}f}_{\nu^{\otimes n}}
		\\&= \ip{f,\ \brac{p\cdot\T_{\rho} + (1-p)\cdot \T_1}^{\otimes n}f}_{\nu^{\otimes n}}
		\\&= \ip{f,\ \T_{p\rho + 1-p}^{\otimes n}\ f}_{\nu^{\otimes n}}
		\\&= \Stab_{1-p(1-\rho)}[f]. \qedhere
	\end{align*}
\end{proof}

With the above, we show that random restrictions of arbitrary functions are essentially high-degree, when the restriction parameter is chosen appropriately.

\begin{lemma}\label{lemma:random_rr_noise_stab}
	Let $f\in L^2(\Sigma^n, \nu^{\otimes n})$, and let $\delta\in (0,1),\ T\in \N$.
	Let $p\in [0,1]$ be chosen uniformly at random from the set $\set{1,\delta,\delta^2,\dots,\delta^{T-1}}$, and let $I\sim_{1-p}[n],\ z\sim\nu^{\otimes I}$.
	Then, for every $\eta > 0$, we have
		\[ \Pr\sqbrac{\Stab_{1-\delta}\sqbrac{f_{I\to z}-\nu(f_{I\to z})}\geq \eta\cdot \Var[f]} \leq \frac{1}{\eta T}. \]
\end{lemma}
\begin{proof}
	By Lemma~\ref{lemma:rr_noise_stab}, we have 
	\begin{align*}
		\E_{p,I,z}\sqbrac{\Stab_{1-\delta}\sqbrac{f_{I\to z}-\nu(f_{I\to z})}} &= \E_p\sqbrac{\Stab_{1-p\delta}[f]-\Stab_{1-p}[f]} 
		\\&= \frac{1}{T}\sum_{t=0}^{T-1}\brac{{\Stab_{1-\delta^{t+1}}[f]-\Stab_{1-\delta^t}[f]}}
		\\&= \frac{1}{T}\brac{\Stab_{1-\delta^T}[f]-\Stab_{0}[f]}
		\\&\leq \frac{1}{T}\brac{\Stab_{1}[f]-\Stab_{0}[f]} = \frac{\Var[f]}{T}.
	\end{align*}
	The result follows by Markov's inequality.
\end{proof}


\subsection{Coordinate-wise Connected Implies No Abelian Embeddings}

In this subsection, we show that any \emph{sufficiently connected} set $\mc S$ has no-Abelian-embeddings. We start by defining connected and coordinate-wise connected subsets $\mc{S} \subseteq \Sigma_1 \times \dots \times \Sigma_k$, which are notions that have been studied in previous works \cite{DHVY17,GHMRZ22}.

\begin{definition}\label{defn:connection_graph}  (Connectivity Sets)
	Let $k\in \N$, let $\Sigma_1,\dots,\Sigma_k$ be finite sets, and let $\mc S\subseteq \Sigma_1\times\dots\times\Sigma_k$.
	Define the (simple, undirected) connection graph $\mc H(\mc S)$ as follows: The vertex set is $\mc S$, and there is an edge between $x,y\in \mc S$ if and only if they differ in exactly one coordinate;
	that is, there exists $j\in[k]$ such that $x_{-j}=y_{-j}$ and $x_j\not=y_j$.\footnote{By $x_{-j}$ we mean $(x_1,\dots,x_{j-1},x_{j+1},\dots,x_k)\in \Sigma_{-j} = \prod_{i\in [k], i\not=j}\Sigma_i$.}
	
	We say that $\mc S$ is \emph{connected} if the graph $\mc H(\mc S)$ is connected.
\end{definition}

Next, we define an even weaker\footnote{Observe that any set $\mc S$ that is connected must also be coordinate-wise connected, while the converse is not necessarily true.} notion of connectivity, where we only require the projection of the above graph with respect to each of the $k$ coordinates to be connected.

\begin{definition}\label{defn:coord_conn_graph} (Coordinate-wise Connected Sets)
	Let $k\in \N$, let $\Sigma_1,\dots,\Sigma_k$ be finite sets, and let $\mc S\subseteq \Sigma_1\times\dots\times\Sigma_k$.	
	For $j\in [k]$, define the (simple, undirected) connection graph $\mc H_j(\mc S)$ as follows:
	The vertex set is $\Sigma_j$, and there is an edge between $x_j, x_j'\in \Sigma_j$ if and only if there exists $x_{-j}\in \Sigma_{-j}$ such that both $(x_j,x_{-j})\in \mc S$ and $(x_j',x_{-j})\in \mc S$.
	
	We say that $\mc S$ is connected with respect to coordinate $j\in [k]$ if $\mc H_j(\mc S)$ is connected.
	We say that $\mc S$ is \emph{coordinate-wise connected} if $\mc S$ is connected with respect to each $j\in [k]$.	
\end{definition}

We show that coordinate-wise connected sets have no-Abelian-embeddings:

\begin{lemma}\label{lemma:conn_implies_nae}
	Let $k\in \N$, let $\Sigma_1,\dots,\Sigma_k$ be finite sets, and let $\mc S\subseteq \Sigma_1\times\dots\times\Sigma_k$.
	Let $(G,\sigma_1,\dots,\sigma_k)$ be an Abelian embedding of $\mc S$, and let $j\in [k]$ be such that $\mc S$ is connected with respect to coordinate $j\in [k]$.
	Then, the map $\sigma_j$ is constant.
	
	In particular, if $\mc S$ is coordinate-wise connected, then it has no-Abelian-embeddings.
\end{lemma}
\begin{proof}
	Let $(G,\sigma_1,\dots,\sigma_k)$ be an Abelian embedding of $\mc S$, and let $j\in [k]$ be such that $\mc S$ is connected with respect to coordinate $j\in [k]$.
	
	Consider any edge $\set{x_j,x'_j}$ in the graph $\mc H_j(\mc S)$; there exists $x_{-j}\in \Sigma_{-j}$ such that $x=(x_j,x_{-j})\in \mc S$ and $x'=(x'_j,x_{-j})\in \mc S$.
	By the definition of Abelian embeddings, we have \[\sum_{i=1}^k \sigma_i(x_i) = 0_G = \sum_{i=1}^k \sigma_i(x'_i),\] and hence $\sigma_j(x_j) = \sigma_j(x'_j)$.
	The lemma now follows by observing that $\mc H_j(\mc S)$ is connected. 
\end{proof}


\section{Product Pseudorandomness, Inverse Theorems, and Generalized Random Restrictions}\label{sec:pair_inv_rr}

\subsection{Product Pseudorandomness and Inverse Theorems}

We define the crucial notion of a product pseudorandom function.
We say that a function is product-pseudorandom if its random restriction down to $n'$ coordinates has nontrivial correlation to a product function with small probability.
\begin{definition}\label{defn:prod_pseudo}($(n',\gamma)$-product pseudorandomess)
Let $(\Sigma, \mu)$ be a probability space, $n\in \N$.
For $n'\leq n$ and $\gamma>0$, we say that a function $f:\Sigma^n\to \C$ is $(n',\gamma)$-product pseudorandom if for any $\delta \in \sqbrac{\frac{n'}{n}, 1}$, the probability that a random restriction of $f$ down to $\bar{I}\sim _\delta [n]$ is $\gamma$-correlated to a product function is less than $\gamma$.
Precisely,

\[ \Pr_{I\sim_{1-\delta}[n],z\sim \mu^{\otimes I}}\sqbrac{\exists \set{P_i:\Sigma\to \C,\ \Norm{P_i}_{\infty} \leq 1}_{i\in \bar{I}} \text{ with } \abs{\E_{x\sim \mu^{\otimes \bar{I}}}\sqbrac{f_{I\to z}(x) \prod_{i\in \bar{I}}P_i(x_i)}}\geq \gamma} < \gamma.\]
\end{definition}

A key input to our analysis is an inverse theorem of \cite{BKLM24a,BKLM24b} which analyzes the correlations of product-pseudorandom functions over certain pairwise-connected distributions.
\begin{theorem}\label{thm:pair_conn_inv} (Inverse Theorem; {\cite[Lemma 1.4]{BKLM24b}})

	Let $k\in \N$, let $\Sigma_1,\dots,\Sigma_k$ be finite sets, and let $\mu$ be a distribution over $\Sigma_1\times \dots \times \Sigma_k$ such that:
	\begin{enumerate}
		\item $\mu$ is a pairwise-connected distribution (see Definition~\ref{defn:pair_conn}).
		\item The the marginal of $\mu$ on coordinates $\set{1,2,\dots,k-1}$, denoted $\mu_{-k}$, admits no-Abelian-embeddings (see Definition~\ref{defn:ab_emb}).
	\end{enumerate}
	Then, for every sufficiently large $n\in \N$ and every $\epsilon>0$, there exists $\delta = \delta(\epsilon)>0$, such that the following holds:
	If 1-bounded functions $f_i:\Sigma_i^n\to \C$, $i\in [k]$, satisfy \[ \abs{\E_{(x_1,\dots,x_k)\sim \mu^{\otimes n}}\sqbrac{\prod_{i=1}^k f_i(x_i)}} \geq \epsilon, \] 
	then $f_1$ is not $(\delta n, \delta)$-product pseudorandom (over the probability space $(\Sigma_1^n, \mu_1^{\otimes n})$, where $\mu_1$ denotes the marginal distribution of $\mu$ on coordinate $1$).
	
	Quantitatively, when $\epsilon\leq o_n(1)$, we can take $\delta = \exp(-\exp(\cdots \exp(\epsilon^{-1})))$, where the number of exponentials is at most $k^{O(k)}$.\footref{footnote:constant_wrt_n}
\end{theorem}

Note that in the above theorem, by symmetry, if $\mu_{-k}$ has no-Abelian-embeddings, then each of $f_1, \dots, f_{k-1}$ are not product-pseudorandom. Thus, if any two $\mu_{-i}$ and $\mu_{-j}$ have no-Abelian embeddings, then all the functions $f_1,\dots,f_k$ are not product-pseudorandom. This motivates the notion of connectivity/pseudorandomness defined in Definition \ref{defn:pair_conn_marginal}.

\begin{corollary}\label{corr:pair_conn_bklm}
Let $k\in \N$, let $\Sigma_1,\dots,\Sigma_k$ be finite sets, and let $\mu$ be a distribution over $\Sigma_1\times \dots \times \Sigma_k$ that is pairwise-connected with no-marginal-Abelian-embeddings (see Definition~\ref{defn:pair_conn_marginal}).
	Then, for every sufficiently large $n\in \N$ and every $\delta>0$, there exists $\epsilon = \epsilon(\delta)>0$, such that the following holds:

Let $f_i:\Sigma_i^n\to [0,1]$, $i\in [k]$ be functions such that for each $i\in [k]$, the function $f_i-\mu_i(f_i)$ is $(\delta n, \delta)$-product pseudorandom (over the space $(\Sigma_i^n, \mu_i^{\otimes n})$, where $\mu_i$ denotes the marginal distribution of $\mu$ on coordinate $i$).
Then,
\[ \abs{\E_{(x_1,\dots,x_k)\sim \mu^{\otimes n}}\sqbrac{\prod_{i=1}^k f_i(x_i)} - \prod_{i=1}^k \mu_i(f_i)} \leq \epsilon. \]
Quantitatively, when $\delta\leq o_n(1)$, we can take
$\epsilon = \frac{1}{\log\log\cdots\log\brac{\frac{1}{\delta}}}$, where the number of logarithms at most $k^{O(k)}$.
\end{corollary}
\begin{proof}
	Let $\delta >0$, and let $\epsilon>0$ be such that $(\frac{\epsilon}{k},\delta)$ satisfy Theorem~\ref{thm:pair_conn_inv}.
	Then, we have
	\begin{align*}
		&\abs{\E_{(x_1,\dots,x_k)\sim \mu^{\otimes n}}\sqbrac{\prod_{i=1}^k f_i(x_i)} - \prod_{i=1}^k \mu_i(f_i)}
		\\&\qquad\qquad= \abs{\sum_{t=1}^k \E_{(x_1,\dots,x_k)\sim \mu^{\otimes n}}\sqbrac{\prod_{i=1}^{t-1}\mu_i(f_i)\cdot \brac{f_t(x_t)-\mu_t(f_t)}\cdot  \prod_{i=t+1}^k f_i(x_i)}}
		\leq k\cdot \frac{\epsilon}{k} = \epsilon,
	\end{align*}
	where we first used the triangle inequality and then used Theorem~\ref{thm:bklm_inverse} for each term.
	Note that $f_t(x_t)-\mu_t(f_t)$ is 1-bounded for each $t\in [k]$; also note that	 since there exist two distinct indices $i_1,i_2 \in [k]$, such that $\mu_{-i_1}$ and $\mu_{-i_2}$ admit no-Abelian-embeddings, the theorem is applicable to each term.
	
	The bound on $\epsilon$ is now $\frac{k}{\log\log\cdots\log\brac{\frac{1}{\delta}}}$, and for $\delta\leq o_n(1)$, this is as desired (by possibly increasing the number of logarithms by 1).
\end{proof}

Similar to how Corollary~\ref{corr:bklm} is only applicable to high-degree functions, the above corollary is only applicable to functions that are product-pseudorandom.
In most applications, the functions we care about are not product-pseudorandom, and we use a suitable generalization of the notion of random restrictions to make arbitrary functions product-pseudorandom.
In contrast to Lemma~\ref{lemma:random_rr_noise_stab}, the random restriction here shall depend on the functions we are working with.


\subsection{Generalized Random Restrictions}

In this section we introduce the key notion of a generalized random restriction which appeared in the prior works \cite{BKLM24a,BKLM24c}. To start, we define a specific type of restriction of a function/distribution which enforces that certain subsets of coordinates all take the same value.
\begin{definition}
	Let $\Sigma$ be a finite alphabet, let $n\in \N$, and let $f:\Sigma^n\to \C$ be a function.
	
	Let $T\subseteq [n]$.
	We define the function $f_{=T}:\Sigma^{n-\abs{T}+1}\to \C$ as follows: for $y\in \Sigma$ and $z\in \Sigma^{[n]\setminus T}$, let $x\in \Sigma^n$ be the vector with $x_i=z_i$ for $i\in [n]\setminus T$, and let $x_i=y$ for each $i\in T$; then, $f_{=T}(y,z) := f(x)$. 
	
	For disjoint sets $T_1,\dots,T_m \subseteq [n]$, we write $f_{=T_1,\dots,T_m} = \brac{\dots\brac{f_{=T_1}}_{=T_2}\dots}_{=T_m}$ as the naturally defined function on $n- \sum_{i=1}^m \abs{T_i} + m$ coordinates.
\end{definition}

With this definition in hand, we are ready to define our notion of generalized restrictions, which both enforces that certain subsets of coordinates take the same value, while fixing the values of other coordinates (i.e., a restriction as in Definition \ref{defn:restriction}).

\begin{definition}\label{defn:gr}
	(Generalized Restrictions)
	Let $\Sigma$ be a finite alphabet, and $n\in \N$.
	A generalized restriction $\rho$ on $\Sigma^n$ is a tuple $\rho = (T_1,T_2,\dots,T_m, I, z)$ where the sets $T_1,T_2,\dots,T_m,I \subseteq [n]$ form a disjoint partition of $[n]$, and $z\in \Sigma^{I}$.
	
	For a generalized restriction $\rho$, we call $m$ (also denoted $m(\rho)$) as the number of \emph{free variables/coordinates} in $\rho$, and define the event
	\[ E_\rho = \set{x\in \Sigma^n : x_i=x_j\ \forall k\in [m],\ i,j\in T_k \text{ and } x_i = z_i\ \forall i\in I }. \]
	
	For a function $f:\Sigma^n\to \C$, and a generalized restriction $\rho = (T_1,T_2,\dots,T_m, I, z)$, we define the restricted function $f_\rho: \Sigma^m \to \C$ as $f_\rho(y) = f(x)$, where 
	\[ x_i = \begin{cases} y_j,& i\in T_j,\ j\in [m]	 \\z_i,&i\in I  \end{cases}.\]	
	Note that the above is the same as the function $\brac{f_{I\to z}}_{=T_1,T_2,\dots, T_m}$.
\end{definition}
Thus, the restriction $\rho$ sets all input variables in each $T_j$ to be equal, and sets variables in coordinates $I$ to value $z$.
Note that this generalizes the usual notions of restrictions (see Definition~\ref{defn:restriction}) which is the same as above with each $T_j$ of size 1.
We observe that generalized restrictions are closed under composition:
\begin{observation} (Composition of generalized restrictions)
	Let $\Sigma$ be a finite alphabet, let $n\in \N$, and let $f:\Sigma^n\to \C$ be a function.
	Let $\rho = (T_1,\dots,T_m,I,z)$ be a generalized random restriction on $\Sigma^n$ with $m\leq n$ free coordinates, and let $\rho'=(S_1,\dots,S_k, J,w) $ be a generalized random restriction on $\Sigma^m$ with $k\leq m$ free coordinates.
	
	Then, the function $(f_\rho)_{\rho'}:\Sigma^k\to \C$ equals the function $f_{\rho'\circ \rho}$, where  $\rho'\circ \rho$ is a generalized random restriction on on $\Sigma^n$ with $k$  free coordinates, defined as $\rho'\circ \rho = (U_1,\dots,U_k, K,y)$, where:
	\begin{enumerate}
		\item For each $i\in [k]$, it holds that $U_i = \bigcup_{j\in S_i} T_j$.
		\item $K = \bigcup_{j\in J}T_j \cup I$.
		\item $y\in \Sigma^K$ is given as follows: for $i\in I$, it holds that $y_i=z_i$; for $i\in T_j$, for some $j\in J$, it holds that $y_i = w_j$.
	\end{enumerate}
\end{observation} 
\begin{proof}
	This is verified by simply following the definitions.
\end{proof}

Finally, we define a generalized random restriction. Formally, this is a distribution over generalized restrictions such that the overall original distribution $\mu^{\otimes n}$ is preserved on average.

\begin{definition}\label{defn:grr} (Generalized Random Restriction)
	Let $(\Sigma,\mu)$ be a probability space, and let $n\in \N$.
	Let $1\leq m\leq n$, and let $\epsilon \in [0,1]$.
	An $(m,\epsilon)$-generalized random restriction on $\Sigma^n$ is a distribution $\mc R$ over generalized restrictions on $\Sigma^n$ (see Definition~\ref{defn:gr}) satisfying:
	\begin{enumerate}
		\item Each $\rho\in\supp(\mc R)$ has at least $m$ free coordinates.
		\item The distribution obtained by first sampling $\rho\sim \mc R$ and then sampling $x\in \Sigma^n$ conditioned on the restriction $\rho$, is close (in $\ell_1$-norm) to the distribution $\mu^{\otimes n}$. Formally, \[ \Norm{\E_{\rho \sim \mc R}\sqbrac{\mu^{\otimes n}| E_\rho} - \mu^{\otimes n}}_1 \leq \epsilon. \]
	\end{enumerate}	
\end{definition}
Note that property 2 justifies the use of the term \emph{random restriction}. Property 2 additionally implies that $\abs{\E_{\rho\sim \mc R}\sqbrac{\mu(f_\rho)}-\mu(f)}\leq \epsilon$ for every 1-bounded $f:\Sigma^n\to \C$.

\subsection{Product Pseudorandomness under Generalized Random Restrictions}

In this subsection, we show that arbitrary functions can be made product-pseudorandom under suitable generalized random restrictions.
We start by showing an increment lemma, which we later iterate.
The proof of this lemma follows Lemma 7.1 in \cite{BKLM24c}.

\begin{lemma}\label{lemma:uniformization_increment} (Increment Lemma)
	Let $(\Sigma,\mu)$ be a probability space with $\abs{\Sigma}\geq 2$, let $c = \frac{1}{100\abs{\Sigma}^2}$, and let $n\in \N$ be sufficiently large.
	
	For any $\gamma \in [n^{-c}, 1]$, the following holds:
	Let $g:\Sigma^n\to \C$ be a 1-bounded function such that $g-\mu(g)$ is not $(\sqrt{n},\gamma)$-product pseudorandom.
	Then, there exists a $(n^c, n^{-0.01})$-generalized random restriction $\mc R$ on $\Sigma^n$ (see Definition~\ref{defn:grr}) such that \[ \E_{\rho\sim \mc R}\abs{\mu(g_\rho)}^2 \geq \abs{\mu(g)}^2 +\frac{\gamma^3}{16}.\]
\end{lemma}

\begin{proof}
	Let $\delta = 1/\sqrt{n}$ and $f=g-\mu(g)$; then, we know that $f$ is 2-bounded, and not $(\delta n,\gamma)$-product pseudorandom.
	That is, for some $\delta'\geq \delta$,
	\[ \Pr_{I\sim_{1-\delta'}[n],z\sim \mu^{\otimes I}}\sqbrac{\exists \set{P_i:\Sigma\to \C,\ \Norm{P_i}_{\infty} \leq 1}_{i\in \bar{I}} \text{ with } \abs{\E_{x\sim \mu^{\otimes \bar{I}}}\sqbrac{f_{I\to z}(x) \prod_{i\in \bar{I}}P_i(x_i)}}\geq \gamma} \geq  \gamma.\]
	This implies that with probability at least $\gamma/2$ over $I\sim_{1-\delta'}[n]$, it holds that
	\[ \Pr_{z\sim \mu^{\otimes I}}\sqbrac{\exists \set{P_i:\Sigma\to \C,\ \Norm{P_i}_{\infty} \leq 1}_{i\in \bar{I}} \text{ with } \abs{\E_{x\sim \mu^{\otimes \bar{I}}}\sqbrac{f_{I\to z}(x) \prod_{i\in \bar{I}}P_i(x_i)}}\geq \gamma} \geq  \frac{\gamma}{2}.\]
	By a Chernoff bound (Fact~\ref{fact:chernoff}), we know $\Pr_{I\sim_{1-\delta'}[n]}[\abs{I} \geq \brac{1-\delta/2}n] \leq e^{-\delta n/8} < \gamma/2$.
	Hence, there exists $I\subseteq [n]$, with $\abs{I}\leq \brac{1-\delta/2}n$, and such that 
	\[
		\Pr_{z\sim \mu^{\otimes I}}\sqbrac{\exists \set{P_i:\Sigma\to \C,\ \Norm{P_i}_{\infty} \leq 1}_{i\in \bar{I}} \text{ with } \abs{\E_{x\sim \mu^{\otimes \bar{I}}}\sqbrac{f_{I\to z}(x) \prod_{i\in \bar{I}}P_i(x_i)}}\geq \gamma} \geq  \frac{\gamma}{2}.
	\]

	We fix such an $I$, and relabel $\bar{I}$ as $[m]$, with $m\geq \delta n/2$.
	Also, we define the set $\mc G$ as the set of all $z$ for which the condition in the above equation holds; that is
	\[ \mc G = \set{z\in \Sigma^I : \exists \set{P_i:\Sigma\to \C,\ \Norm{P_i}_{\infty} \leq 1}_{i\in \bar{I}} \text{ with } \abs{\E_{x\sim \mu^{\otimes \bar{I}}}\sqbrac{f_{I\to z}(x) \prod_{i\in \bar{I}}P_i(x_i)}}\geq \gamma}.\]

	The required distribution $\mc R$ is now defined as follows:
	Let $z\sim \mu^{\otimes I}$ be chosen randomly.
	Consider the following cases:
\begin{itemize}
\item If $z\not \in \mc G$: then all the coordinates in $[m]$ are kept alive; formally, output the restriction $\rho=(T_1,\dots,T_m,I,z)$ where each $T_i = \set{i} \subseteq [m]=\bar{I}$.

\item Suppose $z\in \mc G$:
	Let $\set{P_i:\Sigma\to \C,\ \Norm{P_i}_{\infty} \leq 1}_{i\in [m]}$ be such that 
	\[\abs{\E_{x\sim \mu^{\otimes m}}\sqbrac{f_{I\to z}(x) \prod_{i\in [m]}P_i(x_i)}}\geq \gamma. \]
	Let $v_1,\dots,v_m:\Sigma\to \R/\Z$ be such that $P_j(x)=e^{2\pi i v_j(x)}$ for all $j\in [m], x\in \Sigma$.\footnote{It is without loss of generality that $\abs{P_j(x)}=1$ for each $j\in [m], x\in \Sigma$. This is because there always exist such functions $P_1,\dots,P_m$ maximizing the quantity $\abs{\E_{x\sim \mu^{\otimes m}}\sqbrac{f_{I\to z}(x) \prod_{i\in [m]}P_i(x_i)}}$, since the expression is \emph{multilinear} in the $P_i(x_i)$'s.}
	
	Let $s = \abs{\Sigma}\geq 2$ and $r=\lceil m^{\frac{1}{16s^2}}\rceil $.
	By the pigeonhole principle, we can find disjoint sets $S_1,S_2,\dots, S_{r} \subseteq [m]$, each of size $\abs{S_i} = \lceil\sqrt{m}/2\rceil$, such that for each $i\in [r]$ and $j,j'\in S_i$, we have $\Norm{v_j-v_{j'}}_{\infty} \leq m^{-\frac{1}{2s}}$.
	Consider arbitrary indices $a_1\in S_1,\ a_2\in S_2,\dots,a_{r}\in S_{r}$, and for each $i\in [r]$, let $1\leq k_i \leq m^{\frac{1}{4s}}$ be such that $ \Norm{k_i v_{a_i}}_{\infty} \leq m^{-\frac{1}{8s^2}}$; note that such a $k_i$ exists by applying the pigeonhole principle on the vectors $v_{{a_i}},2v_{a_i}, 3v_{a_i}, \dots, \lfloor m^{\frac{1}{4s}}\rfloor \cdot v_{a_i}$.
	
	Now, for each $i\in [r]$, let $T_i$ be a uniformly random subset of $S_i$ of size $k_i$.
	We perform the following (generalized) random restriction: For each $i\in [r]$, force $x_j=x_{j'}$ for each $j,j'\in T_i$, and then do a uniform random restriction on coordinates $J:=[m]\setminus (T_1\cup T_2\cup\dots T_r)$, given by $u\sim \mu^{\otimes J}$.
	
	Finally the random restriction is given by $\rho = (T_1,\dots,T_r,\ I\cup J,\ (z,u))$.
\end{itemize}

	We show that $\mc R$ satisfies the three conditions of the lemma statement. 
\begin{enumerate}
	\item By definition, each $\rho\in \supp(\mc R)$ contains at least $r \geq m^{\frac{1}{16s^2}} \geq \brac{\frac{\delta n}{2}}^{\frac{1}{16s^2}}$ free coordinates.
	This is at least $n^c$.
	\item It suffices to show that for every choice of $z\in \mu^{\otimes I}$, it holds that \[ \Norm{\E\sqbrac{\mu^{\otimes n}\st E_\rho, z} -  \brac{\mu^{\otimes n}\st z} }_1 \leq n^{-0.01}. \]
	In the first case when $z\not\in \mc G$, this holds with error zero.
	
	In the second case $z\in \mc G$, denoting $\rho' = (T_1,\dots,T_r, J,u)$ as the generalized random restriction on $\Sigma^m=\Sigma^{\bar{I}}$, we can bound the expression on the left hand side using Lemma~\ref{lemma:small_set_force_equality} as
	\begin{equation}\label{eqn:avg_distr}
		\Norm{\E_\rho \sqbrac{\mu^{\otimes m} \st E_{\rho'} }-\mu^{\otimes m}}_1 \leq \sum_{i=1}^r \frac{Ck_i}{\sqrt{\abs{S_i}}} \leq \sqrt{2}C\cdot r\cdot m^{\frac{1}{4s}-\frac{1}{4}} \leq  O\brac{ m^{\frac{1}{16s^2}+\frac{1}{4s}-\frac{1}{4}}} \leq n^{-0.01}.
	\end{equation}
	For the last inequality, we used that $s\geq 2$ and $m\geq \delta n/2 = \sqrt{n}/2$.
\end{enumerate}
	
	It remains to show the third condition, which is the $\ell_2$ increment.
	For this, we consider the second case of the generalized random restriction (which occurs with probability $\gamma/2$) and consider some fixed $z\in \mc G$.
	Let $P = P_1\cdot P_2\cdots P_m$, and $Q = \brac{P_{J\to u}}_{=T_1,\dots,T_r}$.
	First, we show that for every choice of $T_1,\dots,T_r,u$, the restricted function $Q$ is nearly constant: for any $x\in \Sigma^{r}$:
	\begin{align*}
		\abs{Q(x)-\prod_{j\in J} P_j(u_j)} &=  \abs{\brac{\prod_{j\in J} P_j(u_j) }\cdot \sqbrac{\exp\brac{2\pi i \sum_{t=1}^r \sum_{j\in T_t}v_j(x_t)}-  1}}
		\\&\leq  \abs{\sum_{t=1}^r \sum_{j\in T_t}v_j(x_t)} \leq \sum_{t=1}^r \Norm{\sum_{j\in T_t} v_j}_{\infty}.
	\end{align*}
	Now, for any $i\in [r]$, we have
	\[ \Norm{\sum_{j\in T_i} v_j}_{\infty} \leq \Norm{k_i v_{a_i}}_{\infty} + \sum_{j\in T_i} \Norm{v_j-v_{a_i}}_{\infty} \leq m^{-\frac{1}{8s^2}}+k_i\cdot m^{-\frac{1}{2s}} \leq 2\cdot m^{-\frac{1}{8s^2}} .\]
	This implies that $\abs{Q(x)-\prod_{j\in J} P_j(u_j)}\leq 2\cdot  r\cdot m^{-\frac{1}{8s^2}} \leq 2\cdot m^{-\frac{1}{16s^2}}$.
	
	Now, using the above and Equation~\ref{eqn:avg_distr}, along with the fact that $f$ is 2-bounded, we get
	\begin{align*}
		\gamma &\leq \abs{\E_{x\sim \mu^{\otimes m}}\sqbrac{f_{I\to z}(x)P(x)}}
		\\&\leq \abs{\E_{T_1,\dots,T_r,u}\E_{x\sim \mu^{\otimes r}}\sqbrac{(f_{I\to z, J\to u})_{=T_1,\dots T_r}(x)\cdot \brac{P_{J\to u}}_{=T_1,\dots,T_r}(x)}} + 2\cdot n^{-0.01}
		\\&\leq \E_{T_1,\dots,T_r,u}\abs{\E_{x\sim \mu^{\otimes r}}\sqbrac{(f_{I\to z, J\to u})_{=T_1,\dots T_r}(x)}} + 4\cdot m^{-\frac{1}{16s^2}} + 2\cdot n^{-0.01}
		\\&= \E_{T_1,\dots,T_r,u} \abs{\mu(f_\rho)} + 4\cdot m^{-\frac{1}{16s^2}} + 2\cdot n^{-0.01}
		\\&= \E_{T_1,\dots,T_r,u} \abs{\mu(g_\rho)-\mu(g)} + 4\cdot m^{-\frac{1}{16s^2}} + 2\cdot n^{-0.01}
		\\&\leq \E_{T_1,\dots,T_r,u} \abs{\mu(g_\rho)-\mu(g)} + 6\cdot n^{-\frac{1}{50s^2}}.
	\end{align*}
	Using that $\gamma \geq n^{-c} \geq 12\cdot n^{-\frac{1}{50s^2}}$, and Cauchy-Schwarz, we get
	\[  \E_{T_1,\dots,T_r,u} \abs{\mu(g_\rho)-\mu(g)}^2 \geq \brac{\frac{\gamma}{2}}^2 = \frac{\gamma^2}{4}. \]
	Now, averaging over $z$, and observing that $\Pr[z\in \mc G] \geq \gamma/2$, we get
	\begin{align*}
		\frac{\gamma^3}{8} &\leq \E_{\rho\sim \mc R} \abs{\mu(g_\rho)-\mu(g)}^2
		\\&= \E_{\rho\sim \mc R}\abs{\mu(g_\rho)}^2 + \abs{\mu(g)}^2 - \overline{\mu(g)}\cdot \E_{\rho\sim \mc R}\sqbrac{\mu(g_\rho)}  - \mu(g)\cdot \overline{\E_{\rho\sim \mc R}\sqbrac{\mu(g_\rho)}}
		\\&\leq \E_{\rho\sim \mc R}\abs{\mu(g_\rho)}^2 + \abs{\mu(g)}^2 - 2\abs{\mu(g)}^2 + 2\cdot n^{-0.01}
		\\&\leq \E_{\rho\sim \mc R}\abs{\mu(g_\rho)}^2 - \abs{\mu(g)}^2  + \frac{\gamma^3}{16}. \qedhere
	\end{align*}
\end{proof}

Now, we iterate the increment lemma, and show that any small collection of functions can be made product pseudorandom under generalized random restrictions.

\begin{proposition}\label{prop:uniformization} (Uniformization)
	Let $(\Sigma,\mu)$ be a probability space with $\abs{\Sigma}\geq 2$, let $c = \frac{1}{100\abs{\Sigma}^2}$, and let $r,n\in \N$.
	Let $g_1,\dots,g_r:\Sigma^n\to \C$ be a collection of  1-bounded functions.

	Let $0< \delta,\gamma \leq 1$ be such that $\frac{1}{\delta\gamma^3} \leq \frac{1}{200r\log(1/c)}\cdot \log\log n$, and let $T = \big\lceil\frac{50r}{\delta \gamma^3}\big\rceil \in \N$.
	Then, there exists a $(n^{c^T},\ n^{-c^T})$-generalized random restriction $\mc R$ on $\Sigma^n$ (see Definition~\ref{defn:grr}), 
	such that the following holds with probability at least $1-\delta$ over $\rho\sim \mc R$: 
	
	Suppose $m(\rho)$ denotes the number of free coordinates in $\rho$.
	Then, for every $i\in [r]$, the function $(g_i)_\rho:\Sigma^{m(\rho)}\to \C$ is such that $(g_i)_\rho-\mu((g_i)_\rho)$ is $(\sqrt{m(\rho)}, \gamma)$-product pseudorandom.
\end{proposition}
\begin{proof}
	For any generalized restriction $\rho$, we say that a restriction $\rho$ is bad for $i\in [r]$ if the function $(g_i)_\rho:\Sigma^{m(\rho)}\to \C$ is such that $(g_i)_\rho-\mu((g_i)_\rho)$ is not $(\sqrt{m(\rho)}, \gamma)$-product pseudorandom.
	We say $\rho$ is bad if it is bad for some $i\in [r]$, and say that $\rho$ is good otherwise.
	
	Let $\mc R^{(0)}$ be the random restriction that does nothing.
	For $t=1,2,\dots,T$, define a generalized random restriction $\mc R^{(t)}$ as follows: 
	\begin{enumerate}
		\item Choose $\rho\sim \mc R^{(t-1)}$.
		\item If $\rho$ is good,  output the restriction $\rho$.
		\item\label{step:increment_step} Else, consider an arbitrary $i\in [r]$ such that $\rho$ is bad for $i$. Let $\mc R_\rho$ be the $(m(\rho)^c, m(\rho)^{-0.01})$-generalized random restriction on $\Sigma^{m(\rho)}$ obtained by applying Lemma~\ref{lemma:uniformization_increment} to the function $(g_i)_\rho$. Choose $\rho'\sim \mc R_\rho$ and output $\rho'\circ \rho$.
	\end{enumerate}
	
	By induction, it is easily verified that for each $t=0,1,\dots,T$, the random restriction $\mc R^{(t)}$ is a $(m^{(t)},\epsilon^{(t)})$-generalized random restriction on $\Sigma^n$, with \[ m^{(t)} = \brac{m^{(t-1)}}^c = n^{c^t} \geq n^{c^T} \geq 2^{\sqrt{\log n}}.\] and
	\begin{align*}
		\epsilon^{(t)} = \epsilon^{(t-1)} + \brac{m^{(t-1)}}^{-0.01} \leq t\cdot n^{-\frac{1}{100}c^{t-1}} \leq T\cdot n^{-\frac{1}{100}c^{T-1}} \leq T\cdot n^{-4c^T} \leq n^{-c^T} =:\epsilon.
	\end{align*}
	In the last inequality above, we used that $\frac{100r}{\delta\gamma^3} \leq n^{c^T}$, which follows from the choice of parameters.
	Note that this also implies that $\gamma \geq n^{-c^T}$, and this is at least $m(\rho)^{-c}$ whenever Lemma~\ref{lemma:uniformization_increment} is applied above in Step~\ref{step:increment_step}, as needed in the assumption of the lemma.
	
	Now, suppose at some point we found some $\rho$ that is bad for $i\in [r]$, and applied Step~\ref{step:increment_step} above to get $\rho'\sim \mc R_\rho$; then it holds:
	\[ \E_{\rho'\sim \mc R_\rho}\sqbrac{ \abs{\mu\brac{(g_i)_{\rho'\circ \rho}}}^2} \geq \abs{\mu\brac{(g_i)_{\rho}}}^2 + \frac{\gamma^3}{16},\]
	and for each $j\not=i$,
	\[ \E_{\rho'\sim \mc R_\rho}\sqbrac{ \abs{\mu\brac{(g_j)_{\rho'\circ \rho}}}^2} \geq \abs{\E_{\rho'\sim \mc R_\rho}\sqbrac{ \mu\brac{(g_j)_{\rho'\circ \rho}}}}^2 \geq \abs{\mu\brac{(g_j)_{\rho}}}^2 - 2 \epsilon, \]
	where we used Cauchy-Schwarz, the second property in Defintiion~\ref{defn:grr}, and Lemma~\ref{lemma:complex_num_inequality}.	
	In particular, it holds
	\[ \sum_{j=1}^r\E_{\rho'\sim \mc R_\rho}\sqbrac{ \abs{\mu\brac{(g_j)_{\rho'\circ \rho}}}^2} \geq \sum_{j=1}^r \abs{\mu\brac{(g_j)_{\rho}}}^2 + \brac{\frac{\gamma^3}{16} - 2(r-1) \epsilon} \geq \sum_{j=1}^r \abs{\mu\brac{(g_j)_{\rho}}}^2 + \frac{\gamma^3}{32}. \]
	In the last inequality, we used $\gamma^3 \geq 64r\epsilon$, which follows from the inequality $\frac{100r}{\delta\gamma^3} \leq n^{c^T}$ that we used earlier.
	With the above, we get that for every $t=1,\dots,T$,
	\[ \sum_{j=1}^r\E_{\rho\sim \mc R^{(t)}}\sqbrac{\mu((g_j)_\rho)^2} \geq \sum_{j=1}^r\E_{\rho\sim \mc R^{(t-1)}}\sqbrac{\mu((g_j)_\rho)^2} + \frac{\gamma^3}{32} \cdot \Pr_{\rho\sim \mc R^{(t-1)}}\sqbrac{\rho \text{ is bad}}. \]
	
	Suppose, for the sake of contradiction that $\Pr_{\rho\sim \mc R^{(t-1)}}\sqbrac{\rho \text{ is bad}} \geq \delta$ for each $t\in [T]$.
	Then, we have
	\[ r\geq \sum_{j=1}^r\E_{\rho\sim \mc R^{(T)}}\sqbrac{\mu((g_j)_\rho)^2} \geq \sum_{j=1}^r\mu(g_j)^2 + \frac{T\delta\gamma^3}{32} \geq 0 + \frac{50}{32}\cdot r > r,\]
	which is a contradiction.
	Hence, for some $t=0,1,\dots,T-1$, the generalized random restriction $\mc R^{(t)}$ satisfies the statement of the lemma.
\end{proof}

\begin{corollary}\label{corr:uniformization}
		Let $(\Sigma,\mu)$ be a probability space with $\abs{\Sigma}\geq 2$, and let $n\in \N$.
	Let $g_1,\dots,g_r:\Sigma^n\to \C$ be a collection of  1-bounded functions, with $r=O(1)$.
	
	Let $0 < \gamma \leq 1$ be such that $\frac{1}{\gamma} \leq o\brac{\log\log n}^{1/4}$.
	Then, there exists a $(\frac{1}{\eta}, \eta)$-generalized random restriction $\mc R$ on $\Sigma^n$, with $\eta=n^{-\exp(-1/\gamma^4)}$, 
	such that the following holds with probability at least $1-\gamma$ over $\rho\sim \mc R$: 
	For every $i\in [r]$, the function $(g_i)_\rho:\Sigma^{m(\rho)}\to \C$ is such that $(g_i)_\rho-\mu((g_i)_\rho)$ is $(\sqrt{m(\rho)}, \gamma)$-product pseudorandom.
	
	The constant in the $\exp$ depends only on $r, \abs{\Sigma}$.
\end{corollary}
\begin{proof}
	This follows from Proposition~\ref{prop:uniformization} by choosing $\delta=\gamma$.
\end{proof}


\subsection{Conditional Generalized Restrictions}

We shall also be interested in a conditional notion of generalized restrictions, defined as follows:

\begin{definition}\label{defn:cond_grr}
	Let $(\Sigma, \mu)$ be a probability space, and let $n\in \N$.
	Let $\mc R$ be a $(m,\epsilon)$-generalized random restriction (see Definition~\ref{defn:grr}) over $\Sigma^n$, and let $E\subseteq \Sigma^n$ be an event such that $\Pr[E]>\epsilon$.
	
	We define the corresponding conditional generalized restriction, denoted $\mc R|E$, as the distribution over generalized random restrictions $\rho$ given by:
	\[ (\mc R|E) [\rho] = \frac{\Pr[E|E_\rho]\cdot \mc R[\rho]}{\E_{\rho'\sim \mc R}\sqbrac{\Pr\sqbrac{E|E_{\rho'}}}}.\]
	Note that the denominator in the above expression satisfies $\E_{\rho'\sim \mc R}\sqbrac{\Pr\sqbrac{E|E_{\rho'}}} \geq \Pr[E]-\epsilon > 0$ by the assumption $\Pr[E]>\epsilon$.
\end{definition}

These satisfy the following properties, which informally say that as long as $\Pr[E]$ is much larger than $\eps$, then conditional generalized restrictions preserve probabilities and distributions on average.

\begin{lemma}\label{lemma:cond_grr_prop}
	Under the setting of Definition~\ref{defn:cond_grr}, we have
	\begin{enumerate}
		\item For any $\rho$, it holds that \[ \abs{(\mc R|E) [\rho] -   \frac{\Pr[E|E_\rho]\cdot \mc R[\rho]}{\Pr[E]}}  \leq (\mc R|E) [\rho]\cdot \frac{\epsilon}{\Pr[E]}. \]
		\item The distribution obtained by first sampling $\rho\sim \mc R|E$ and then sampling $x$ from $\mu^{\otimes n}|E$ conditioned on the restriction $\rho$, is close (in $\ell_1$-norm) to the distribution $\mu^{\otimes n}|E$. Formally, \[ \Norm{\E_{\rho\sim \mc R|E}\sqbrac{\mu^{\otimes n}|E_\rho,E} - \brac{\mu^{\otimes n}|E}}_1 \leq \frac{2\epsilon}{\Pr[E]}. \]
	\end{enumerate}
\end{lemma}
\begin{proof}
	By Definition~\ref{defn:grr}, it follows that $\abs{\E_{\rho'\sim \mc R}\sqbrac{\Pr\sqbrac{E|E_{\rho'}}} - \Pr[E]} \leq \epsilon$. Hence, for every $\rho$, it holds that
	\[ \abs{(\mc R|E) [\rho] -   \frac{\Pr[E|E_\rho]\cdot \mc R[\rho]}{\Pr[E]}} = (\mc R|E) [\rho]\cdot  \frac{\abs{\E_{\rho'\sim \mc R}\sqbrac{\Pr\sqbrac{E|E_{\rho'}}}-\Pr[E]}} { \Pr[E]} \leq (\mc R|E) [\rho]\cdot \frac{\epsilon}{\Pr[E]}. \]
	
	Using this, we also have 
	\begin{align*}
		\Norm{\E_{\rho\sim \mc R|E}\sqbrac{\mu^{\otimes n}|E_\rho,E} - \brac{\mu^{\otimes n}|E}}_1 &= \sum_{x\in \Sigma^n}\abs{\brac{\E_{\rho\sim \mc R|E}\sqbrac{\mu^{\otimes n}|E_\rho,E}}[x] - \Pr[x|E]}
		\\&= \sum_{x\in E}\abs{\sum_{\rho}(\mc R|E)[\rho]\cdot \Pr[x|E,E_\rho] -\Pr[x|E]}
		\\&\leq \sum_{x\in E}\abs{\sum_\rho \frac{\Pr[E|E_\rho]\cdot \mc R[\rho]}{\Pr[E]}\cdot \Pr[x|E,E_\rho] - \Pr[x|E]} + \frac{\epsilon}{\Pr[E]}
		\\&=\sum_{x\in E}\abs{\E_{\rho\sim \mc R} \frac{\Pr[x,E|E_\rho]}{\Pr[E]}  - \Pr[x|E]} + \frac{\epsilon}{\Pr[E]}
		\\&=\frac{\sum_{x\in E}\abs{\E_{\rho\sim \mc R}\Pr[x|E_\rho] - \Pr[x]}}{\Pr[E]}   + \frac{\epsilon}{\Pr[E]}
		\\&\leq \frac{2\epsilon}{\Pr[E]}. \qedhere
	\end{align*}
\end{proof}


\section{Multiplayer Games}\label{sec:multiplayer_games}
We formally define notions associated to multiplayer games, and establish some notation.

\begin{definition}\label{defn:multiplayer_games} (Multiplayer Game)
	A $k$-player game $\mc G$ is a tuple $\mc G = (\mc X, \mc A, Q, V)$, where the question set $\mc X = \mc X^1 \times\dots\times \mc X^k$ and the answer set $\mc A= \mc A^1 \times\dots\times \mc A^k$ are finite sets, $Q$ is a probability distribution over $\mc X$, and $V:\mc X\times \mc A\to\set{0,1}$ is a predicate.
\end{definition}

The game $\mc G$ proceeds as follows: A verifier samples questions $X = (X^1, \dots, X^k)\sim Q$; then, for each $j\in [k]$, the verifier sends the question $X^j\in \mc X^j$ to the $j\textsuperscript{th}$ player, to which the player responds back with answer $A^j\in \mc A^j$.
Finally, the verifier declares that the players win if and only if $V\brac{X^1,\dots,X^k,A^1,\dots,A^k} = 1$.

\begin{definition} (Game Value)
	Let $\mc G = (\mc X, \mc A, Q, V)$ be a $k$-player game.
	
	For a sequence $\brac{f^j:\mc X^j\to \mc A^j}_{j\in[k]}$ of functions, define the function $f = f^1\times\dots\times f^k : \mc X \to \mc A$ by $f(x^1,\dots,x^k) = \brac{f^1(x^1), \dots, f^k(x^k)}$. 
	We use the term \emph{product functions} to denote functions $f$ defined in this manner, and the functions $(f^j)_{j\in[k]}$ are called \emph{player strategies}.
	
	The value $\val(\mc G)$ of the game $\mc G$ is defined as 
	\[\val(\mc G) = \max_{f = f^1\times \dots\times f^k}\ \Pr_{X\sim Q}\sqbrac{V(X,f(X))=1},\]
	where the maximum is over all product functions $f = f^1\times \dots\times f^k$.
    \label{defn:game_value}
\end{definition}

\begin{fact}\label{fact:rand_no_help}
	The value of the game is unchanged even if we allow the player strategies to be randomized; that is, we allow the strategies to depend on some additional shared and private randomness.
	This is because there always exists an optimal fixed value for the randomness. 
\end{fact}

Next, we define the parallel repetition of a $k$-player game, which corresponds to playing $n$ independent copies of the game in parallel.

\begin{definition}\label{defn:game_parrep} (Parallel Repetition)
	Let $\mc G = (\mc X, \mc A, Q, V)$ be a $k$-player game. We define its $n$-fold repetition as $\mc G^{\otimes n} = (\mc X^{\otimes n}, \mc A^{\otimes n}, Q^{\otimes n}, V^{\otimes n})$.
	The sets $\mc X^{\otimes n},\ \mc A^{\otimes n}$ are defined to be the $n$-fold product of the sets $\mc X,\mc A$ with themselves respectively.\footnote{We use the notation $\mc X^{\otimes n}$ instead of the usual $\mc X^n$ so as to avoid confusion with the sets $\mc X^1, \dots, \mc X^k$.}
	The distribution $Q^{\otimes n}$ is the $n$-fold product of the distribution $Q$ with itself, that is, $Q^{\otimes n}[x] = \prod_{i=1}^n Q[x_i]$ for each $x\in \mc X^{\otimes n}$.
	The predicate $V^{\otimes n}$ is defined as $V^{\otimes n}(x, a) = \bigwedge_{i=1}^n V(x_i, a_i)$.
\end{definition}

\textbf{Notation:} We use subscripts to denote the coordinates in the parallel repetition, and superscripts to denote the players. 
That is, for any $S\subseteq [n], T\subseteq [k]$, we shall use $x_S^T$ to denote the questions in coordinates $S$, that the players in set $T$ receive.
For example, for $i\in[n]$ and $j\in[k]$, we will use $x_i^j$ to refer to the question to the $j\ts{th}$ player in the $i\ts{th}$ repetition of the game.
Similarly, $x_i$ will refer to the vector of questions to the $k$ players in the $i\ts{th}$ repetition, and $x^j$ will refer to the vector of questions received by the $j\ts{th}$ player over all repetitions.
We use $x^{-j}$ to refer to the questions to all players except the $j\ts{th}$ player, and use $x_{-i}$ to refer to the questions in all coordinates except the $i\ts{th}$ coordinate.

\section{Games with No Abelian Embeddings}\label{sec:ab_emb}

\begin{theorem}\label{thm:noabemb_parrep}
	Let $\mc G = (\mc X, \mc A, Q, V)$ be a $k$-player game such that the distribution $Q$ has no-Abelian-embeddings (see Definition~\ref{defn:ab_emb}), and such that $\val(\mc G) < 1$.
	Then, there exists a constant $C\in \N,\ C\leq k^{O(k)}$, such that for every sufficiently large $n\in \N$, \[\val(\mc G^{\otimes n}) \leq \frac{1}{\log \log\cdots\log n},\]
	where the number of logarithms is $C$.
\end{theorem}

In this section, we shall prove this theorem.
Let $\mc G = (\mc X, \mc A, Q, V)$ be a $k$-player game such that $Q$ has no-Abelian-embeddings, and such that $\val(\mc G) < 1$.
For some sufficiently large $n$, consider the game  $\mc G^{\otimes n} = (\mc X^{\otimes n}, \mc A^{\otimes n}, P=Q^{\otimes n}, V^{\otimes n})$, and fix an optimal strategy for the $k$ players in this game.
Let $X= (X^1,\dots,X^k)$ be the random variable denoting the questions to the $k$ players in the game $\mc G^{\otimes n}$, and let $A= (A^1,\dots,A^k)$ be the random variable denoting the answers of the players using these strategies.
For each $i\in [n]$, let $\Win_i$ (resp. $\Lose_i$) be the event that $V(X_i, A_i)=1$ (resp. $V(X_i, A_i)=0$); that is, the players win (resp. lose) the $i$\textsuperscript{th} coordinate of the game.

We introduce some parameters that will be useful: 
\[\delta = (\log n)^{-1/3},\ T = \left\lceil \frac{1}{\delta^2} \right\rceil,\ \epsilon = \frac{1}{\underbrace{\log\log\cdots\log}_{C \text{ times}}n},\ \alpha = \sqrt{\epsilon},\]
where $2\leq C\leq k^{O(k)}$ is a constant so that $(\epsilon, \delta)$ satisfy Corollary~\ref{corr:bklm} with respect to the probability space $(\mc X, Q)$; recall that the distribution $Q$ has no-Abelian-embeddings.

Let $E=E^1\times \dots \times E^k \subseteq (\mc X^1)^{\otimes n}\times\dots\times (\mc X^k)^{\otimes n} = \mc X^{\otimes n}$ be an arbitrary product event with $\Pr_{Q^{\otimes n}}[E]\geq \alpha$.
We prove the following lemma:

\begin{lemma}\label{lemma:main_rand_coord_hard}
    For parameters chosen as above and $\Pr[E] \ge \alpha$, it holds that
	\[ \E_{i\sim [n]}\sqbrac{\Pr[\Win_i\st E]} \leq \val(\mc G) + o_n(1), \]
	where the expectation is over $i\in [n]$ chosen uniformly at random.
\end{lemma}

Assuming this, the main theorem follows directly:
\begin{proof}[Proof of Theorem~\ref{thm:noabemb_parrep}]	
	This follows by combining a standard inductive parallel repetition proof (see Lemma~\ref{lemma:prod_set_hard_coor}) and Lemma~\ref{lemma:main_rand_coord_hard}, and using the fact that $\val(\mc G) < 1$;
	note that the theorem works with the constant $C+1\leq k^{O(k)}$.
\end{proof}

The remainder of this section is devoted to proving Lemma~\ref{lemma:main_rand_coord_hard}.

\subsection{Strategy for a Single Copy of the Game}

The above lemma is proven via an \emph{embedding argument}, where we try to embed a single copy of the game $\mc G$ into the $i$\textsuperscript{th} coordinate of the game $\mc G^{\otimes n}$ (while conditioning on $E$).
Formally, we construct the following randomized strategy for the game $\mc G$:
\begin{enumerate}
	\item The verifier samples $\tilde{X}\sim Q$, and for each $j\in [k]$, gives player $j$ the input $\tilde{X}^j$.
	\item Using shared randomness, the players sample:
	\begin{enumerate}
		\item $i\sim [n]$ uniformly at random.
		\item $p\in (0,1]$ uniformly at random from the set $\set{1,\delta,\delta^2,\dots,\delta^{T-1}}$.
		\item $I \sim_p [n]\setminus \set{i} $; let $I' = [n]\setminus\brac{I\cup\set{i}}$.
		\item $Z\sim P_{X_{I'}|E} = Q^{\otimes I'}\st E$.
	\end{enumerate}
	\item For each $j\in [k]$, player $j$ does the following: 
	 
	We say that $\tilde{X}^j, Z^j$ are consistent with $E^j$ if 
	\[\set{x^j\in (\mc X^j)^{\otimes n} : x^j\in E^j,\ x^j_i=\tilde{X}^j,\ x^j_{I'} = Z^j} \not=\emptyset.\] 
	\begin{enumerate}
		\item If $\tilde{X}^j, Z^j$ are not consistent with $E^j$, output an arbitrary answer from $\mc A^j$; for example, we may assume the output is the first element of $\mc A^j$ under some ordering.
		\item Else, using private randomness, output $\tilde{A}^j \sim P_{A_i^j \st X^j\in E^j,\ X^j_i=\tilde{X}^j,\ X^j_{I'}=Z^j}$.
	\end{enumerate}	
	\item Let $\tilde{A}=(\tilde{A}^1,\tilde{A}^2,\dots,\tilde{A}^k)$; the players win if and only if $V(\tilde{X}, \tilde{A})=1$.
\end{enumerate}

Let $L$ be the event that $V(\tilde X, \tilde A) = 0$; that is, the players lose the game $\mc G$ when using the above strategy.
We shall prove:

\begin{lemma}\label{lemma:embed_strat_analysis}
	\[  \Pr[L] \leq \E_{i\sim [n]}\sqbrac{\Pr[\Lose_i\st E]} + \frac{4\abs{\mc A}}{\alpha}\cdot\brac{\frac{k}{\delta T} + \epsilon}+\sqrt{\frac{2}{\delta^T n}\cdot \log_2\brac{\frac{1}{\alpha}}}. \]
\end{lemma}

Assuming this, we can easily complete the proof of Lemma~\ref{lemma:main_rand_coord_hard}:

\begin{proof}[Proof of Lemma~\ref{lemma:main_rand_coord_hard}]
	It must hold that $\Pr[L]\geq 1-\val(\mc G)$.
	Hence, by Lemma~\ref{lemma:embed_strat_analysis}, we get
	\[ \E_{i\sim [n]}\sqbrac{\Pr[\Win_i\st E]}\leq \val(\mc G) + \frac{4\abs{\mc A}}{\alpha}\cdot\brac{\frac{k}{\delta T} + \epsilon}+\sqrt{\frac{2}{\delta^T n}\cdot \log_2\brac{\frac{1}{\alpha}}}. \]
	By our choice of parameters, we have
	\begin{enumerate}
		\item $\frac{4\abs{\mc A}k}{\alpha \delta T} \leq O\brac{\frac{\delta}{\alpha}} \leq O\brac{\frac{\epsilon}{\alpha}} = O\brac{\sqrt{\epsilon}} \leq o_n(1).$
		\item $\frac{4\abs{\mc A}\epsilon}{\alpha} \leq O\brac{ \frac{\epsilon}{\alpha}} = O\brac{\sqrt{\epsilon}} \leq o_n(1).$
		\item $\delta^T \geq \frac{1}{\sqrt{n}}$, and so $\frac{1}{\delta^T n }\log_2\brac{1/\alpha} \leq \frac{\log_2 n}{\sqrt n} \leq o_n(1)$.
	\end{enumerate}
	Hence, $\E_{i\sim [n]}\sqbrac{\Pr[\Win_i\st E]}\leq \val(\mc G) + o_n(1)$, as desired.
\end{proof}

Now, we focus on proving Lemma~\ref{lemma:embed_strat_analysis}.
We break up the losing probability into two terms, and bound them individually, as follows:

\begin{lemma}\label{lemma:distr_change_in_i}
	\[ \E_{i,p,I,Z} \Norm{P_{X_i|E,X_{I'}=Z}-Q}_1 \leq \sqrt{\frac{2}{\delta^T n}\cdot \log_2\brac{\frac{1}{\alpha}}}.\]
\end{lemma}

\begin{lemma}\label{lemma:fixed_i_rr_term}
	For every fixed $i\in [n]$, it holds that
	\[ \E_{p,I,Z} \E_{\tilde{X}\sim P_{X_i|E,X_{I'}=Z}} \sqbrac{\Pr[L\st \tilde{X},i,p,I,Z]} \leq \Pr\sqbrac{\Lose_i| E} + \frac{4\abs{\mc A}}{\alpha}\cdot\brac{\frac{k}{\delta T} + \epsilon}. \]
\end{lemma}

Assuming these, the lemma follows easily:
\begin{proof}[Proof of Lemma~\ref{lemma:embed_strat_analysis}]
	
	We can write the losing probability as
	\begin{align*}
		\Pr[L] &= \E_{\tilde{X}\sim Q}\E_{i,p,I,Z} \sqbrac{\Pr[L\st \tilde{X},i,p,I,Z]}
		\\&= \E_{i,p,I,Z} \E_{\tilde{X}\sim Q} \sqbrac{\Pr[L\st \tilde{X},i,p,I,Z]}
		\\&\leq \E_{i,p,I,Z}\Norm{P_{X_i|E,X_{I'}=Z}-Q}_1 + \E_{i,p,I,Z} \E_{\tilde{X}\sim P_{X_i|E,X_{I'}=Z}} \sqbrac{\Pr[L\st \tilde{X},i,p,I,Z]}. 
	\end{align*}
	Now the result follows by Lemma~\ref{lemma:distr_change_in_i} and Lemma~\ref{lemma:fixed_i_rr_term}.
\end{proof}

We complete our proof by proving Lemma~\ref{lemma:distr_change_in_i} in Section~\ref{sec:distr_change_in_i} and Lemma~\ref{lemma:fixed_i_rr_term} in Section~\ref{sec:fixed_i_rr_term}.

\subsection{Analysis of the First Term}\label{sec:distr_change_in_i}

We wish to analyze the following term, i.e., how much the random restriction $I'$ and event $E$ can affect the marginal distribution of the $i$-th question:

\[ \E_{i,p,I,Z} \Norm{P_{X_i|E,X_{I'}=Z}-Q}_1. \]

To analyze this, we first fix a value of $p$, and show the following using Lemma \ref{lemma:pdt_distr_cond}.
\begin{lemma}\label{lemma:lemma:distr_change_in_i_fixp}
	Conditioned on any fixed $p\in (0,1]$, it holds that
	\[ \E_{i\sim [n]}\E_{I\sim_p [n]\setminus \set{i}}\E_{Z\sim P_{X_{I'}|E}} \Norm{P_{X_i|E,X_{I'}=Z}-Q}_1 \leq \sqrt{\frac{2}{pn}\cdot \log_2\brac{\frac{1}{\alpha}}}.\]
\end{lemma}
\begin{proof}
	We have
	\begin{align*}
		&\E_{i\sim [n]}\E_{I\sim_p [n]\setminus \set{i}}\E_{Z\sim P_{X_{I'}|E}} \Norm{P_{X_i|E,X_{I'}=Z}-Q}_1
		\\&\qquad\qquad=\ \sum_{i\in [n]}\sum_{I\subseteq [n], I\not\ni i} \frac{1}{n}\cdot p^{\abs{I}} (1-p)^{n-1-\abs{I}} \E_{Z\sim P_{X_{\overline{I\cup\set{i}}}|E}} \Norm{P_{X_i|E,X_{\overline{I\cup\set{i}}}=Z}-Q}_1
		\\&\qquad\qquad=\ \sum_{i\in [n]}\sum_{I\subseteq [n], I\ni i} \frac{1}{n}\cdot p^{\abs{I}-1} (1-p)^{n-\abs{I}} \E_{Z\sim P_{X_{\overline{I}}|E}} \Norm{P_{X_i|E,X_{\overline{I}}=Z}-Q}_1
		\\&\qquad\qquad=\ \sum_{I\subseteq [n], I\not=\emptyset}\sum_{i\in I} \frac{1}{pn}\cdot p^{\abs{I}} (1-p)^{n-\abs{I}} \E_{Z\sim P_{X_{\overline{I}}|E}} \Norm{P_{X_i|E,X_{\overline{I}}=Z}-Q}_1
		\\&\qquad\qquad=\ \E_{I\sim_p [n]}\sum_{i\in I}\sqbrac{ \frac{1}{pn}\cdot  \E_{Z\sim P_{X_{\overline{I}}|E}} \Norm{P_{X_i|E,X_{\overline{I}}=Z}-Q}_1}
		\\& \qquad\qquad=\ \E_{I\sim_p [n]}\sqbrac{\frac{1}{pn}\cdot \E_{Z\sim P_{X_{\overline{I}}|E}} \sum_{i\in I} \Norm{P_{X_i|E,X_{\overline{I}}=Z}-Q}_1}.
	\end{align*}
	
	Now, for any fixed $I,Z$, we can apply Lemma~\ref{lemma:pdt_distr_cond}, and get that the above at most
	\begin{align*}
		&\E_{I\sim_p [n]}\sqbrac{\frac{1}{pn}\cdot \E_{Z\sim P_{X_{\overline{I}}|E}} \sqrt{2\abs{I}\cdot \log_2\brac{\frac{1}{\Pr\sqbrac{E|X_{\overline{I}}=Z}}}}}
		\\&\qquad\qquad\qquad =\E_{I\sim_p [n]}\sqbrac{\frac{\sqrt{2\abs{I}}}{pn}\cdot \E_{Z\sim P_{X_{\overline{I}}|E}} \sqrt{\log_2\brac{\frac{1}{\Pr\sqbrac{E|X_{\overline{I}}=Z}}}}}
	\end{align*}
	Note that the function $\sqrt{\log_2(\cdot)}$ is concave over $[1,\infty)$; hence, by Jensen's inequality, the above is at most
	\begin{align*}
		\E_{I\sim_p [n]}\sqbrac{\frac{\sqrt{2\abs{I}}}{pn}\cdot \sqrt{\log_2\brac{\E_{Z\sim P_{X_{\overline{I}}|E}} \frac{1}{\Pr\sqbrac{E|X_{\overline{I}}=Z}}}}}
		&\leq \E_{I\sim_p [n]}\sqbrac{\frac{\sqrt{2\abs{I}}}{pn}\cdot \sqrt{\log_2\brac{\frac{1}{\Pr[E]}}} }
		\\&\leq \frac{1}{pn}\sqrt{2\cdot \log_2\brac{\frac{1}{\Pr[E]}}}\cdot \E_{I\sim_p [n]}\sqbrac{\sqrt{\abs{I}}}.
	\end{align*}
	
	Finally, by Cauchy-Schwarz inequality, we get $\E_{I\sim_p [n]}\sqbrac{\sqrt{\abs{I}}} \leq \sqrt{\E_{I\sim_p [n]}\abs{I}} = \sqrt{pn}$, and plugging this above completes the proof.
\end{proof}

With the above, the lemma we wish to prove follows:
\begin{proof}[Proof of Lemma~\ref{lemma:distr_change_in_i}]
	In the randomized strategy for $\mc G$, the players choose $p \sim \set{1,\delta,\dots,\delta^{T-1}}$, and hence $p\geq \delta^T$ almost surely.
	Plugging this into Lemma~\ref{lemma:lemma:distr_change_in_i_fixp} gives the desired result.
\end{proof}

\subsection{Analysis of the Second Term}\label{sec:fixed_i_rr_term}

For the remainder of this section, fix an index $i\in [n]$.
We want to analyze the following quantity:
\[ \E_{p,\ I\sim_p [n]\setminus \set{i}}\E_{Z\sim P_{X_{I'}|E}} \E_{\tilde{X}\sim P_{X_i|E,X_{I'}=Z}} \sqbrac{\Pr[L\st \tilde{X},i,p,I,Z]}.\]

Observe that when $Z\sim P_{X_{I'}|E}$ and $\tilde{X}\sim P_{X_i|E,X_{I'}=Z}$, the consistency condition (in the definition of the embedding strategy) holds almost surely for each player $j\in [k]$, and so they answer $\tilde{A}^j \sim P_{A_i^j \st X^j\in E^j,\ X^j_i=\tilde{X}^j,\ X^j_{I'}=Z^j}$.

Hence, we have
\begin{align*}
	& \E_{p, I\sim_p [n]\setminus \set{i}}\E_{Z\sim P_{X_{I'}|E}} \E_{\tilde{X}\sim P_{X_i|E,X_{I'}=Z}} \sqbrac{\Pr[L\st \tilde{X},i,p,I,Z]}
	\\&\quad=\ \E_{p, I\sim_p [n]\setminus \set{i}}\E_{Z\sim P_{X_{I'}|E}} \E_{\tilde{X}\sim P_{X_i|E,X_{I'}=Z}} \sum_{\tilde a \in \mc A: V(\tilde{X},\tilde a)=0} \sqbrac{ \prod_{j=1}^k \Pr\sqbrac{A_i^j=\tilde{a}^j | X^j\in E^j,\ X_i^j=\tilde{X}^j,\ X_{I'}^j=Z^j} }
	\\&\quad=\ \E_{\tilde{X}\sim P_{X_i|E}} \sum_{\tilde a \in \mc A: V(\tilde{X},\tilde a)=0}  \E_{p,I\sim_p [n]\setminus \set{i}} \E_{Z\sim P_{X_{I'}|E,X_i=\tilde{X}}}  \sqbrac{ \prod_{j=1}^k \Pr\sqbrac{A_i^j=\tilde{a}^j | X^j\in E^j,\ X_i^j=\tilde{X}^j,\ X_{I'}^j=Z^j} }.
\end{align*}

With the above expression in mind, we prove the following lemma:
\begin{lemma}\label{lemma:ques_ans_rr_term}
	For every $\tilde{x}\in \mc X, \tilde{a}\in \mc A$ such that $\set{x\in \mc X^{\otimes n}:x_i=\tilde{x},x\in E}\not=\emptyset$, it holds that
	\begin{align*}
		&\E_{p,I\sim_p [n]\setminus \set{i}} \E_{Z\sim P_{X_{I'}|E,X_i=\tilde{x}}}  \sqbrac{ \prod_{j=1}^k \Pr\sqbrac{A_i^j=\tilde{a}^j | X^j\in E^j,\ X_i^j=\tilde{x}^j,\ X_{I'}^j=Z^j} }
		\\&\qquad\qquad\leq\ \Pr\sqbrac{A_i=\tilde{a} | X\in E, X_i=\tilde{x}} + \frac{4}{\Pr[E| X_i=\tilde{x}]}\cdot \brac{\frac{k}{\delta T} + \epsilon}.
	\end{align*}
\end{lemma}

Before proving this lemma, we show the proof of Lemma~\ref{lemma:fixed_i_rr_term} assuming this.

\begin{proof}[Proof of Lemma~\ref{lemma:fixed_i_rr_term}]
	Fix any $i\in [n]$; as before, we have
	\begin{align*}
	& \E_{p, I\sim_p [n]\setminus \set{i}}\E_{Z\sim P_{X_{I'}|E}} \E_{\tilde{X}\sim P_{X_i|E,X_{I'}=Z}} \sqbrac{\Pr[L\st \tilde{X},i,p,I,Z]}
	\\&\quad=\ \E_{p, I\sim_p [n]\setminus \set{i}}\E_{Z\sim P_{X_{I'}|E}} \E_{\tilde{X}\sim P_{X_i|E,X_{I'}=Z}} \sum_{\tilde a \in \mc A: V(\tilde{X},\tilde a)=0} \sqbrac{ \prod_{j=1}^k \Pr\sqbrac{A_i^j=\tilde{a}^j | X^j\in E^j,\ X_i^j=\tilde{X}^j,\ X_{I'}^j=Z^j} }
	\\&\quad=\ \E_{\tilde{X}\sim P_{X_i|E}} \sum_{\tilde a \in \mc A: V(\tilde{X},\tilde a)=0}  \E_{p,I\sim_p [n]\setminus \set{i}} \E_{Z\sim P_{X_{I'}|E,X_i=\tilde{X}}}  \sqbrac{ \prod_{j=1}^k \Pr\sqbrac{A_i^j=\tilde{a}^j | X^j\in E^j,\ X_i^j=\tilde{X}^j,\ X_{I'}^j=Z^j} }
	\\&\quad\leq\ \E_{\tilde{X}\sim P_{X_i|E}} \sum_{\tilde a \in \mc A: V(\tilde{X},\tilde a)=0} \sqbrac{\Pr\sqbrac{A_i=\tilde{a} | E, X_i=\tilde{x}} + \frac{4}{\Pr[E| X_i=\tilde{X}]}\cdot \brac{\frac{k}{\delta T} + \epsilon}}
	\\&\quad=\ \Pr\sqbrac{\Lose_i| E} + \E_{\tilde{X}\sim P_{X_i|E}} \sum_{\tilde a \in \mc A: V(\tilde{X},\tilde a)=0}\sqbrac{\frac{4}{\Pr[E| X_i=\tilde{X}]}\cdot \brac{\frac{k}{\delta T} + \epsilon}}
	\\&\quad\leq\ \Pr\sqbrac{\Lose_i| E} + \sum_{\tilde a \in \mc A} 4 \cdot \brac{\frac{k}{\delta T} + \epsilon}\cdot \E_{\tilde{X}\sim P_{X_i|E}} \sqbrac{\frac{1}{\Pr[E| X_i=\tilde{X}]} }
	\\&\quad\leq \Pr\sqbrac{\Lose_i| E} + 4 \cdot \brac{\frac{k}{\delta T} + \epsilon}\cdot \frac{1}{\Pr[E]}\cdot \abs{\mc A}. 
	\qedhere
\end{align*}
\end{proof}

\subsubsection{Approximate Independence  Under Random Restriction}

In the remaining part of this section, we prove Lemma~\ref{lemma:ques_ans_rr_term}.
For this, we fix some $\tilde{x}\in \mc X, \tilde{a}\in \mc A$ satisfying $\set{x\in \mc X^{\otimes n}:x_i=\tilde{x},x\in E}\not=\emptyset$.
Recall that we want to analyze:
\[ \E_{p,I\sim_p [n]\setminus \set{i}} \E_{Z\sim P_{X_{I'}|E,X_i=\tilde{x}}}  \sqbrac{ \prod_{j=1}^k \Pr\sqbrac{A_i^j=\tilde{a}^j | X^j\in E^j,\ X_i^j=\tilde{x}^j,\ X_{I'}^j=Z^j} } .\]

We now define a set of functions that we want to have small noise stability. We will ensure this by taking random restrictions and using Lemma \ref{lemma:random_rr_noise_stab}. In words, the function $F^j$ is a function on $n-1$ coordinates, specifically $(\mc{X}^j)^{\otimes (n-1)}$, and is the indicator function of when the $j$-th player fills in the $i$-th coordinate with $\tilde{x}^j$, whether the full vector satisfies the event $E^j$. Similarly, $f^j$ is the indicator of $F^j$ being true, and player $j$ answering $\tilde{a}^j$ on coordinate $i$ on the full input.

\begin{definition}\label{defn:abemb_fns_pseudo}
	For each $j\in [k]$, define the functions:
	\begin{enumerate}
		\item Let $F^j:(\mc X^j)^{\otimes n-1}\to \set{0,1}$ be the function given by
		\[ F^j(x_{-i}^j) = \ind[(x_{-i}^j, \tilde{x}^j)\in E^j], \]
		where $(x_{-i}^j,\tilde{x}^j)$ is the vector with $\tilde{x}^j$ in the $i$\textsuperscript{th} coordinate.
		\item Let $f^j:(\mc X^j)^{\otimes n-1}\to \set{0,1}$, be the function given by:
	\[ f^j(x_{-i}^j) = F^j(x_{-i}^j)\cdot \ind[\textnormal{Player }j \textnormal{ has } i\textsuperscript{th}\textnormal{ answer }\tilde{a}^j\textnormal{ on input }(x_{-i}^j,\tilde{x}^j)\in \mc X^{\otimes n} \textnormal{ in }\mc G^{\otimes n} ].\]
	Note that this definition relies on the previously fixed strategy for $\mc G^{\otimes n} $.
	\end{enumerate}		
\end{definition}

\begin{definition} (Good random restriction)
	Let $\Lambda(I,Z)$ be the event that for each $j\in [k]$:
	\[ \Stab_{1-\delta}^{Q^j}[ (F^j)_{I'\to Z^j}-Q^j((F^j)_{I'\to Z^j}) ] < \delta, \]
	\[ \Stab_{1-\delta}^{Q^j}[ (f^j)_{I'\to Z^j}-Q^j((f^j)_{I'\to Z^j}) ] < \delta, \]
	where the functions $(F^j)_{I'\to Z^j}$(resp. $(f^j)_{I'\to Z^j}$) are the restrictions of the functions $F^j$ (resp. $f^j$) with $Z^j$ plugged into coordinates $I'=([n]\setminus\set{i}) \setminus I$.
	We used $Q^j$ to denote the marginal of the query distribution $Q$ on the $j$\textsuperscript{th} player.
\end{definition}

Next, we prove some useful lemmas.
First, we show that the event $\Lambda(I,Z)$ occurs with high probability.

\begin{lemma}\label{lemma:use_rrstab}
	\[ \E_{p,I\sim_p [n]\setminus \set{i}} \E_{Z\sim P_{X_{I'}|E,X_i=\tilde{x}}} \sqbrac{\ind\sqbrac{\lnot\Lambda(I,Z)}} \leq  \frac{1}{\Pr[E| X_i=\tilde{x}]}\cdot \frac{4k}{\delta T}.\]
\end{lemma}
\begin{proof}
	Using Lemma~\ref{lemma:random_rr_noise_stab} (with $\eta=\delta/2$), and by a union bound (over $2k$ functions), we get
	\[ \E_{p,I\sim_p [n]\setminus \set{i}} \E_{Z\sim Q^{\otimes I'}} \sqbrac{\ind\sqbrac{\lnot\Lambda(I,Z)}} \leq  \frac{4k}{\delta T}.\]
	Now, for every $z$, it holds that
	\[ \Pr\sqbrac{X_{I'}=z \st E, X_i=\tilde{x}}\leq  \frac{\Pr\sqbrac{X_{I'}=z, X_i=\tilde{x}}}{\Pr\sqbrac{E, X_i=\tilde{x}}} = \frac{\Pr\sqbrac{X_{I'}=z }}{\Pr\sqbrac{E| X_i=\tilde{x}}} = \frac{Q^{\otimes I'}[z]}{\Pr\sqbrac{E| X_i=\tilde{x}}}. \]
	Hence we have
	\[ \E_{p,I\sim_p [n]\setminus \set{i}} \E_{Z\sim P_{X_{I'}|E,X_i=\tilde{X}}} \sqbrac{\ind\sqbrac{\lnot\Lambda(I,Z)}} \leq  \frac{1}{\Pr[E| X_i=\tilde{x}]}\cdot \frac{4k}{\delta T}. \qedhere\]
\end{proof}

We observe that under the event $\Lambda(I,Z)$, our functions satisfy an approximate independence property:

\begin{lemma}\label{lemma:using_bklm_corr}
	Let $I,Z$ be such that the event $\Lambda(I,Z)$ holds.
	Then,
	\[ \prod_{j=1}^k \Pr\sqbrac{A_i^j = \tilde{a}^j, X^j\in E^j | X_{i}^j=\tilde{x}^j,\ X_{I'}^j=Z^j} \leq \Pr\sqbrac{A_i=\tilde{a},\ X\in E| X_i=\tilde{x},\ X_{I'}=Z} + \epsilon, \]
	\[ \prod_{j=1}^k \Pr\sqbrac{X^j\in E^j | X_{i}^j=\tilde{x}^j,\ X_{I'}^j=Z^j} \geq \Pr\sqbrac{X\in E| X_i=\tilde{x},\ X_{I'}=Z} - \epsilon. \]
\end{lemma}
\begin{proof}
	Let $I,Z$ be such that the event $\Lambda(I,Z)$ holds.
	By Corollary~\ref{corr:bklm} and the definition of the event $\Lambda(I,Z)$, it holds that
	 \[ \prod_{j=1}^k \E_{Y^j\sim (Q^j)^{\otimes I}}\sqbrac{ (f^j)_{I'\to Z^j}(Y^j)} \leq \E_{Y\sim Q^{\otimes I}}\sqbrac{ \prod_{j=1}^k(f^j)_{I'\to Z^j}(Y^j)} + \epsilon, \]
	\[ \prod_{j=1}^k \E_{Y^j\sim (Q^j)^{\otimes I}}\sqbrac{ (F^j)_{I'\to Z^j}(Y^j)} \geq \E_{Y\sim Q^{\otimes I}}\sqbrac{ \prod_{j=1}^k(F^j)_{I'\to Z^j}(Y^j)} - \epsilon. \]
	Now, by the definitions of the functions $(f^j)_{j\in [k]},\ (F^j)_{j\in [k]}$, we have
	\[ \E_{Y\sim Q^{\otimes I}}\sqbrac{ \prod_{j=1}^k(f^j)_{I'\to Z^j}(Y^j)} = \Pr\sqbrac{A_i=\tilde{a},\ X\in E\st  X_i=\tilde{x},\ X_{I'}=Z},\]
	\[ \E_{Y\sim Q^{\otimes I}}\sqbrac{ \prod_{j=1}^k(F^j)_{I'\to Z^j}(Y^j)} = \Pr\sqbrac{X\in E\st X_i=\tilde{x},\ X_{I'}=Z},\]
	and for every $j\in [k]$,
	\[ \E_{Y^j\sim (Q^j)^{\otimes I}}\sqbrac{ (f^j)_{I'\to Z^j}(Y^j)} =  \Pr\sqbrac{A_i^j = \tilde{a}^j, X^j\in E^j | X_{i}^j=\tilde{x}^j,\ X_{I'}^j=Z^j}, \]
	\[ \E_{Y^j\sim (Q^j)^{\otimes I}}\sqbrac{ (F^j)_{I'\to Z^j}(Y^j)} =  \Pr\sqbrac{X^j\in E^j | X_{i}^j=\tilde{x}^j,\ X_{I'}^j=Z^j}. \]
	Plugging these into the above inequalities, we obtain the desired result.	
\end{proof}

Now, we are ready to complete the proof:

\begin{proof}[Proof of Lemma~\ref{lemma:ques_ans_rr_term}]
	We have
	\begin{align*}
		&\E_{p,I\sim_p [n]\setminus \set{i}} \E_{Z\sim P_{X_{I'}|E,X_i=\tilde{x}}}  \sqbrac{ \prod_{j=1}^k \Pr\sqbrac{A_i^j=\tilde{a}^j | X^j\in E^j,\ X_i^j=\tilde{x}^j,\ X_{I'}^j=Z^j} }
		\\&\quad\leq \E_{p,I\sim_p [n]\setminus \set{i}} \E_{Z\sim P_{X_{I'}|E,X_i=\tilde{x}}}  \sqbrac{ \prod_{j=1}^k \Pr\sqbrac{A_i^j=\tilde{a}^j | X^j\in E^j,\ X_i^j=\tilde{x}^j,\ X_{I'}^j=Z^j} \cdot \ind\sqbrac{\Lambda(I,Z)}}
		\\&\hspace{17em}+\ \E_{p,I\sim_p [n]\setminus \set{i}} \E_{Z\sim P_{X_{I'}|E,X_i=\tilde{x}}}\sqbrac{\ind\sqbrac{\lnot\Lambda(I,Z)}}.
	\end{align*}
	
	By Lemma~\ref{lemma:use_rrstab}, the second term above is at most 
	$\frac{1}{\Pr[E| X_i=\tilde{x}]}\cdot \frac{4k}{\delta T}$.
	By Lemma~\ref{lemma:using_bklm_corr} and Lemma~\ref{lemma:math_division}, we can bound the first term as:
	\begin{align*}
		&\E_{p,I\sim_p [n]\setminus \set{i}} \E_{Z\sim P_{X_{I'}|E,X_i=\tilde{x}}}  \sqbrac{ \prod_{j=1}^k \Pr\sqbrac{A_i^j=\tilde{a}^j | X^j\in E^j,\ X_i^j=\tilde{x}^j,\ X_{I'}^j=Z^j} \cdot \ind\sqbrac{\Lambda(I,Z)}}
		\\&\quad= \E_{p,I\sim_p [n]\setminus \set{i}} \E_{Z\sim P_{X_{I'}|E,X_i=\tilde{x}}}  \sqbrac{ \frac{\prod_{j=1}^k \Pr\sqbrac{A_i^j=\tilde{a}^j,\  X^j\in E^j\ |\ X_i^j=\tilde{x}^j,\ X_{I'}^j=Z^j}}{\prod_{j=1}^k \Pr\sqbrac{X^j\in E^j\ |\  X_i^j=\tilde{x}^j,\ X_{I'}^j=Z^j}}  \cdot \ind\sqbrac{\Lambda(I,Z)}}
		\\&\quad\leq \E_{p,I\sim_p [n]\setminus \set{i}} \E_{Z\sim P_{X_{I'}|E,X_i=\tilde{x}}}\sqbrac{\frac{\Pr\sqbrac{A_i=\tilde{a},\ X\in E\ |\ X_i=\tilde{x}, X_{I'}=Z} + 4\epsilon}{\Pr\sqbrac{X\in E\ |\ X_i=\tilde{x}, X_{I'}=Z}}}
		\\&\quad= \E_{p,I\sim_p [n]\setminus \set{i}} \E_{Z\sim P_{X_{I'}|E,X_i=\tilde{x}}}\sqbrac{\Pr\sqbrac{A_i=\tilde{a}\ |\ X\in E, X_i=\tilde{x}, X_{I'}=Z}
		+ \frac{4\epsilon}{\Pr\sqbrac{X\in E\ |\ X_i=\tilde{x}, X_{I'}=Z}}}
		\\&\quad= \Pr\sqbrac{A_i=\tilde{a}\ |\ X\in E,\ X_i=\tilde{x}} + \frac{4\epsilon}{\Pr\sqbrac{X\in E\ |\ X_i=\tilde{x}}}.
	\end{align*}
	Combining the two terms completes the proof.
\end{proof}

\subsection{Some Remarks}\label{sec:nae_remarks}

We remark that the same proof leads to even better bounds on parallel repetition, in the cases we know better CSP inverse theorems for games with no-Abelian-embeddings.
Formally, the same choice of the parameters $\delta = \log(n)^{-1/3}, T=\lceil 1/\delta^2\rceil, \epsilon=\epsilon(\delta), \alpha=\sqrt{\epsilon}$ works; here $\epsilon=\epsilon(\delta) \geq \delta$ is chosen so as to satisfy the CSP inverse theorem in Corollary \ref{corr:bklm}.
This leads to the bound
\[\val(\mc G^{\otimes n}) \leq \epsilon\brac{ \frac{1}{\sqrt[3]{\log n}}}^{\Omega(1)}.\]

This implies the following bounds:
\begin{enumerate}
	\item For a connected game: $\val(\mc G^{\otimes n})\leq(\log n)^{-\Omega(1)}$, via \cite{Mossel10}. In particular, we obtain this bound for all 2-player games, since they are connected without loss of generality.
	\item For a 3-player game with no-Abelian-embeddings: $\val(\mc G^{\otimes n})\leq(\log \log n)^{-\Omega(1)}$, via \cite{BKM2}.
	\item As before, for a $k$-player game with no-Abelian-embeddings: $\val(\mc G^{\otimes n})\leq (\log\cdots \log n)^{-\Omega(1)}$, where number of logarithms is at most $k^{O(k)}$, via \cite{BKLM24b}.
\end{enumerate}


\section{Pairwise Connected Games with No Marginal Abelian Embeddings}\label{sec:paircon_parrep}

The main result of this section is a parallel repetition theorem for pairwise-connected distributions with no-marginal-Abelian embeddings, i.e., Theorem \ref{thm:main}.
Throughout the remainder of this section, we fix a $k$-player game $\mc G = (\mc X, \mc A, Q, V)$ with $\val(\mc G)<1$, and such that the distribution $Q$ is pairwise-connected with no-marginal-Abelian-embeddings.

The proof done in a series of steps, as follows:

\begin{enumerate}
	\item In Section~\ref{sec:pairwise_answer_function_partitions}, given any sufficiently large $n\in \N$, a strategy for the game $\mc G^{\otimes n}$, and a product event $E\subseteq \mc X^{\otimes n}$, we define generalized random restrictions $\mc R_i$ on $\mc X^{\otimes n}$, one for each coordinate $i\in [n]$, that make some relevant functions (corresponding to the set $E$ and the answer functions for this coordinate) product-pseudorandom.
	\item In Section~\ref{sec:pairwise_one_coord_hard_assum}, we show that for any sufficiently large $n\in \N$, a strategy for the game $\mc G^{\otimes n}$, and any product event $E$ of large measure, it is hard for the players to win a coordinate $i\in [n]$, conditioned on the inputs being drawn from $E$, under a certain pseudorandomness assumption. Namely, we want that conditioned on the event $E$, a random restriction $\rho\sim \mc R_i$ does not give too much information on the inputs to the players in coordinate $i$. We note that the proof of this part is similar in spirit to the proof of Theorem~\ref{thm:noabemb_parrep}.
	\item In Section~\ref{sec:pairwise_rr_information}, we show how to achieve the pseudorandom assumption above via an iterative process. More formally, for any sufficiently large $n\in \N$, a strategy for the game $\mc G^{\otimes n}$, and any product event $E$ of large measure, we show that there exists a generalized random restriction $\mc R$ on $\mc X^{\otimes n}$, such that the assumption above is satisfied with high probability when the inputs to the game are drawn conditioned on the restriction $\rho\sim \mc R$.
	\item In Section~\ref{sec:pairwise_combining}, we combine the results in the two sections above. More formally, for any sufficiently large $n\in \N$, a strategy for the game $\mc G^{\otimes n}$, and any product event $E$ of large measure, we show that there exists a generalized random restriction $\mc R$ on $\mc X^{\otimes n}$, such that the game has many hard coordinates when the inputs are drawn condition a restriction $\rho\sim \mc R$.
	\item Finally, in Section~\ref{sec:pairwise_induction}, we use the result of the above section along with an inductive argument to complete the proof of Theorem~\ref{thm:pairconn_parrep}. We note that the proof of this part is similar in spirit to Lemma~\ref{lemma:prod_set_hard_coor}.
\end{enumerate}


\subsection{Pseudorandom Partitions For Each Coordinate}\label{sec:pairwise_answer_function_partitions}

For any sufficiently large $n\in \N$, consider the repeated game $\mc G^{\otimes n}$, and let $(h_i^j:(\mc X^j)^{\otimes n}\to \mc A^j)_{i\in [n],j\in [k]}$ be any fixed strategies for the $k$ players, for each of the $n$ coordinates.
Let $E =E^1\times\dots\times E^k\subseteq \mc X^{\otimes n}$ be a product event.
We define the following set of functions, which are the same functions in Definition~\ref{defn:abemb_fns_pseudo}, except for all possible values of $\tilde{x} \in \mathcal{X}$ and $\tilde{a} \in \mathcal{A}$ (in Definition \ref{defn:abemb_fns_pseudo} we fixed some $\tilde{x} \in \mathcal{X}$ and $\tilde{a} \in \mathcal{A}$ beforehand).

\begin{definition}\label{defn:pairwise_fns_to_make_pseudo}
	Consider any $i\in [n]$. 
	
	For each $\tilde{x}\in \mc X, \tilde{a}\in \mc A, j\in [k]$, define the following functions:
	\begin{enumerate}
		\item The function $F_{i,\tilde{x}^j}^j: (\mc X^j)^{\otimes n-1}\to \set{0,1}$ is given by 
		\[ F_{i,\tilde{x}^j}^j(x_{-i}^j) = \ind[(x_{-i}^j, \tilde{x}^j)\in E^j], \]
		where $(x_{-i}^j,\tilde{x}^j)$ is the vector with $\tilde{x}^j$ in the $i$\textsuperscript{th} coordinate.
		
		\item The function $f_{i,\tilde{x}^j,\tilde{a}^j}^j:(\mc X^j)^{\otimes n-1}\to \set{0,1}$ is given by:
	\[ f_{i,\tilde{x}^j,\tilde{a}^j}^j(x_{-i}^j) = F_{i,\tilde{x}^j}^j(x_{-i}^j)\cdot \ind[h_i^j(x_{-i}^j,\tilde{x}^j)=\tilde{a}^j ].\]
	\end{enumerate}
\end{definition}

Observe that for every $i\in [n]$, the total number of functions in Definition~\ref{defn:pairwise_fns_to_make_pseudo} is $2k\abs{\mc X} \abs{\mc A}$, which is a constant (depending on the base game $\mc G$).
Thus using Corollary \ref{corr:uniformization} we can find a generalized random restriction $\mc{R}_i$ depending on $i \in [n]$ to make all the functions defined in Definition \ref{defn:pairwise_fns_to_make_pseudo} product pseudorandom.

\begin{lemma}\label{lemma:pairwise_pseud_fn_dist}
	Let $0<\gamma\leq 1$ be such that $\frac{1}{\gamma} \leq o(\log\log n)^{1/4}$, and let $i\in [n]$.
	
	Then, there exists a $(\frac{1}{\eta},\eta)$-generalized random restriction $\mc R_i$ on $\mc X^{\otimes n-1}$ (with respect to coordinates $[n]\setminus \set{i}$),\footnote{We shall also think of $\mc R_i$ as a generalized random restriction (with the same parameters) on $\mc X^{\otimes n}$ which always leaves coordinate $i$ untouched.} for $\eta = n^{-\exp(-1/\gamma^4)}$ such that with probability $1-\gamma$ over $\rho\sim \mc R_i$, each function in Definition~\ref{defn:pairwise_fns_to_make_pseudo} (with respect to $i$) is $(\sqrt{m(\rho)},\gamma)$-product pseudorandom under the restriction $\rho$.

		The constant in the $\exp$ depends only on parameters of the base game $\mc G$.	
\end{lemma}
\begin{proof}
	This follows by applying Corollary~\ref{corr:uniformization} to the relevant functions.
	
	Formally, we apply the corollary with the probability space $(\mc X, Q)$, after extending each function in Definition~\ref{defn:pairwise_fns_to_make_pseudo} to a function $\mc X^{\otimes n-1}\to \set{0,1}$; for example, a function corresponding to player $j$ will only depend on the inputs from $(\mc X^j)^{\otimes n-1}$, and ignore the inputs corresonding to the other players. This ensures that the pseudorandomness condition for this function finally holds with respect to the correct marginal $Q^j$.
\end{proof}

In later sections, we shall also be interested in knowing how the parallel repetition of a multiplayer game behaves under generalized restrictions. Generalized random restrictions effectively turn a multiplayer game into the same game on less coordinates. Formally, we define the following:

\begin{definition}\label{defn:game_rep_gr} (Multiplayer game under generalized restriction)
	Let $n\in \N$; consider the game  $\mc G^{\otimes n} = (\mc X^{\otimes n}, \mc A^{\otimes n}, Q^{\otimes n}, V^{\otimes n})$, and fix any strategy for the $k$ players in this game.
	Let $E=E^1\times \dots \times E^k \subseteq  \mc X^{\otimes n}$ be a product event.
	
	Let $\rho = (T_1,T_2,\dots,T_m,I,z)$ be a generalized restriction $\mc X^{\otimes n}$ with $m=m(\rho) \leq n$ free coordinates.
	For any $x'\in \mc X^{\otimes m}$, let $x'^{(\rho)}\in \mc X^{\otimes n}$ be its relevant extension to $\mc X^{\otimes n}$; formally, we have $x'^{(\rho)}_i = \begin{cases} x'_j, & i\in T_j,j\in [m] \\ z_i, & i\in I \end{cases}.$
	Then, we have:
	
	\begin{enumerate}
		\item Consider the game $\mc G^{\otimes n}$, with the inputs to the $k$ players drawn conditioned on $E_\rho$; this input distribution is the same as $Q^{\otimes m}$, the input distribution of the game $\mc G^{\otimes m}$.
		\item Under this identification, we define a restricted event $E'\subseteq \mc X^{\otimes m}$ for the game $\mc G^{\otimes m}$ by \[E':=\set{x'\in \mc X^{\otimes m}:x'^{(\rho)}\in E}.\]
		This is a product event $E'=E'^1\times\dots\times E'^k$ with respect to the $k$ players.
		\item Given the strategy for $\mc G^{\otimes n}$, we can define a restricted strategy for $\mc G^{\otimes m}$ as follows: Fix indices $i_1\in T_1,i_2\in T_2,\dots, i_m\in T_m$.
			Now, on input $x'\in \mc X^{\otimes m}$, the players extend it to an input $x'^{(\rho)}\in \mc X^{\otimes n}$ of $\mc G^{\otimes n}$, and output the answers to coordinates $i_1,i_2,\dots,i_m$ respectively.
			
			Note that the players win the game $\mc G^{\otimes m}$ on input $x'\in \mc X^{\otimes m}$ with the above strategy if they win the game $\mc G^{\otimes n}$ on input $x'^{(\rho)}\in \mc X^{\otimes n}$.
	\end{enumerate}
\end{definition}

\subsection{Embedding for a Single Copy of the Game}\label{sec:pairwise_one_coord_hard_assum}

In this subsection we show that under a certain pseudorandomness condition (that the generalized random restrictions $\mc{R}_i$ do not change the distribution of $X_i$ too much when conditioning on $E$, see \eqref{eq:pseudo_beta}), an embedding argument shows hardness for coordinates of the game $\mc G^{\otimes n}$ (while conditioning on $E$).

For some sufficiently large $n$, consider the game  $\mc G^{\otimes n} = (\mc X^{\otimes n}, \mc A^{\otimes n}, P=Q^{\otimes n}, V^{\otimes n})$, and fix any strategy for the $k$ players in this game.
Let $X= (X^1,\dots,X^k)$ be the random variable denoting the questions to the $k$ players in the game $\mc G^{\otimes n}$, and let $A= (A^1,\dots,A^k)$ be the random variable denoting the answers of the players using these strategies.
For each $i\in [n]$, let $\Win_i$ (resp. $\Lose_i$) be the event that $V(X_i, A_i)=1$ (resp. $V(X_i, A_i)=0$); that is, the players win (resp. lose) the $i$\textsuperscript{th} coordinate of the game.

We introduce some parameters:
\begin{enumerate}
	\item $\beta \in [0,1]$ is any parameter satisfying $\beta \leq o(1)$.
	\item $\gamma\in [0,1]$ is any real number such that $(\log\log\log n)^{-1} \leq \gamma \leq (\log\log\log\log\log  n)^{-1}$.
	\item $\eta = \eta(n,\gamma)= n^{-\exp(-1/\gamma^4)} $ is as in Lemma~\ref{lemma:pairwise_pseud_fn_dist} with respect to the parameter $\gamma$.
	\item  $\delta := (\log\log\log\log\log n)^{-1}$.
	\item Let $1\leq C\leq k^{O(k)}$ be a constant so that Corollary~\ref{corr:pair_conn_bklm} with respect to the probability space $(\mc X, Q)$ holds with $C$ logarithms.
	Let \[\epsilon := \frac{1}{\log \log \cdots \log 1/\delta} = \frac{1}{\log \log \cdots \log n},\] where the number of logarithms is $C$ in the first expression and $C+5$ in the second expression; this is chosen so that $(\epsilon, \delta)$ satisfy Corollary~\ref{corr:pair_conn_bklm}.
	\item  $\alpha := \sqrt{\epsilon} = \frac{1}{\brac{\log \log \cdots \log n}^{1/2}}$, where the number of logarithms is $C+5$.
\end{enumerate}
These satisfy the following inequalities, which we shall use later:
\[ \frac{1}{\gamma} \leq o (\log\log n)^{1/4},\qquad \eta\leq 2^{-\sqrt{\log n}}, \qquad \delta \geq \gamma, \delta \geq \sqrt{\eta}, \qquad \eta, \gamma, \epsilon \leq o(\alpha).\]

%
The main result of this subsection is the following:
\begin{proposition}\label{prop:pairwise_hard_coor_under_assum}
	Let the parameters $\alpha, \beta, \gamma, \eta$ be as above.
	Let $E=E^1\times \dots \times E^k \subseteq(\mc X^k)^{\otimes n} = \mc X^{\otimes n}$ be a product event with $\Pr_{Q^{\otimes n}}[E]\geq \alpha$.
	
	Consider any $i\in [n]$.
	Let $\mc R_i$ be the $(\frac{1}{\eta}, \eta)$-generalized random restriction as in Lemma~\ref{lemma:pairwise_pseud_fn_dist} (with respect to $E$, the fixed player strategies, and the parameter $\gamma \geq \omega(\log\log n)^{-1/4}$; and $\eta = n^{-\exp(-1/\gamma^4)}$), and suppose that it satisfies
	  \begin{equation} \E_{\rho\sim \mc R_i|E} \Norm{P_{X_i|E,E_\rho}-Q}_1 \leq \beta \leq o(1).\footnote{The conditional random restriction $\mc R_i|E$ is as in Definition~\ref{defn:cond_grr}; this is well-defined since $\eta< \alpha$.} \label{eq:pseudo_beta}
      \end{equation}
	Then, it holds that 
	\[ \Pr\sqbrac{\Win_i\st E} \leq \val(\mc G) + o_n(1).\]
\end{proposition}

For the remainder of the section, we fix some $i\in [n]$ and prove the above lemma.
We construct the following randomized strategy for the game $\mc G$:
\begin{enumerate}
	\item The verifier samples $\tilde{X}\sim Q$, and for each $j\in [k]$, gives player $j$ the input $\tilde{X}^j$.
	\item Using shared randomness, the players sample $\rho\sim \mc R_i|E$.
	\item For each $j\in [k]$, player $j$ does the following: 
	 
	We say that $\tilde{X}^j, \rho$ are consistent with $E^j$ if 
	\[\set{x^j\in (\mc X^j)^{\otimes n} : x^j\in E^j\cap E_\rho^j,\ x^j_i=\tilde{X}^j}\not=\emptyset,\]
	where $E_\rho^j\subseteq (\mc X^j)^{\otimes n}$ is the projection of the set $E_\rho$ on player $j$.\footnote{Recall that this simply ensures that certain coordinates (in $[n]$) are equal, and certain coordinates have some fixed value. Note it also holds that $E_\rho = \prod_{j=1}^k E_\rho^j$.} 
	\begin{enumerate}
		\item If the above consistency condition does not hold, output an arbitrary answer from $\mc A^j$; for example, we may assume the output is the first element of $\mc A^j$ under some ordering.
		\item Else, using private randomness, output $\tilde{A}^j \sim P_{A_i^j \st X^j\in E^j\cap E_\rho^j,\ X^j_i=\tilde{X}^j}$.
	\end{enumerate}	
	\item Let $\tilde{A}=(\tilde{A}^1,\tilde{A}^2,\dots,\tilde{A}^k)$; the players win if and only if $V(\tilde{X}, \tilde{A})=1$.
\end{enumerate}

Let $L$ be the event that $V(\tilde X, \tilde A) = 0$; that is, the players lose the game $\mc G$ when using the above strategy.
We prove that:

\begin{lemma}\label{lemma:pairwise_embed_strat_analysis} Under the hypotheses of Proposition \ref{prop:pairwise_hard_coor_under_assum}, it holds that
	\[  \Pr[L] \leq {\Pr[\Lose_i\st E]} + o(1). \]
\end{lemma}

Assuming this, we immediately have:
\begin{proof}[Proof of Proposition~\ref{prop:pairwise_hard_coor_under_assum}]
	This follows by Lemma~\ref{lemma:pairwise_embed_strat_analysis} and the fact that $\Pr[L] \geq 1-\val(\mc G)$.
\end{proof}

Now, we focus on proving Lemma~\ref{lemma:pairwise_embed_strat_analysis}.
First, we show that using our assumption on $\mc R_i$, it suffices to bound the losing probability assuming the input $\tilde{X}$ came from a different distribution, which is the one conditioned on $E, E_\rho$.
Formally, we show:

\begin{lemma}\label{lemma:pairwise_fixed_i_rr_term} Under the hypotheses of Proposition \ref{prop:pairwise_hard_coor_under_assum}, it holds that
	\[ \E_{\rho\sim\mc R_i|E}\E_{\tilde{X}\sim P_{X_i|E,E_\rho}} \sqbrac{\Pr[L\st \tilde{X},\rho]} \leq \Pr\sqbrac{\Lose_i| E} + o(1). \]
\end{lemma}
Assuming these, the lemma follows easily:
\begin{proof}[Proof of Lemma~\ref{lemma:pairwise_embed_strat_analysis}]
	We can write
	\begin{align*}
		\Pr[L] &= \E_{\tilde{X}\sim Q}\E_{\rho\sim\mc R_i|E} \sqbrac{\Pr[L\st \tilde{X},\rho]}
		\\&= \E_{\rho\sim\mc R_i|E} \E_{\tilde{X}\sim Q}\sqbrac{\Pr[L\st \tilde{X},\rho]}
		\\&\leq \E_{\rho\sim\mc R_i|E} \Norm{P_{X_i|E,E_\rho}-Q}_1 + \E_{\rho\sim\mc R_i|E}\E_{\tilde{X}\sim P_{X_i|E,E_\rho}} \sqbrac{\Pr[L\st \tilde{X},\rho]}
		\\&\leq \beta +  \E_{\rho\sim\mc R_i|E}\E_{\tilde{X}\sim P_{X_i|E,E_\rho}} \sqbrac{\Pr[L\st \tilde{X},\rho]}
		\\&\leq o(1) +  \E_{\rho\sim\mc R_i|E}\E_{\tilde{X}\sim P_{X_i|E,E_\rho}} \sqbrac{\Pr[L\st \tilde{X},\rho]}.
	\end{align*}
	Now the result follows by Lemma~\ref{lemma:pairwise_fixed_i_rr_term}.
\end{proof}

Next, we complete the proof of Lemma~\ref{lemma:pairwise_fixed_i_rr_term}.
We want to upper bound the following quantity:
\[  \E_{\rho\sim\mc R_i|E}\E_{\tilde{X}\sim P_{X_i|E,E_\rho}} \sqbrac{\Pr[L\st \tilde{X},\rho]}.\]
To analyze the above, we prove two lemmas.

First, we show that we may assume (upto a small error) that the chosen random restriction $\rho\sim \mc R_i|E$ is good, in the sense it makes the relevant functions pseudorandom.
This is formally defined as:
\begin{definition}
	For a random restriction $\rho \in \supp(\mc R_i)$, we denote by $\Lambda(\rho)$ the event that all the $2k\abs{\mc X}\abs{\mc A}$ functions in Definition~\ref{defn:pairwise_fns_to_make_pseudo} with respect to coordinate $i$ are $(\sqrt{m(\rho)},\gamma)$-product pseudorandom.
\end{definition}

\begin{lemma}\label{lemma:pairwise_rand_is_good}
	\[ \Pr_{\rho\sim \mc R_i|E}\sqbrac{\lnot \Lambda(\rho)} \leq \frac{\gamma}{\alpha} + \frac{\eta}{\alpha} \leq o(1).\]
\end{lemma}
\begin{proof}
	We know by Lemma~\ref{lemma:pairwise_pseud_fn_dist} that $\Pr_{\rho\sim \mc R_i}\sqbrac{\lnot \Lambda(\rho)} \leq \gamma$.
	Hence, by Lemma~\ref{lemma:cond_grr_prop} we have
	\begin{align*}
		\Pr_{\rho\sim \mc R_i|E}\sqbrac{\lnot \Lambda(\rho)} &\leq  \sum_{\rho} \frac{\mc R[\rho]\cdot\Pr[E|E_\rho]}{\Pr[E]}\cdot \ind\sqbrac{\lnot \Lambda(\rho)} + \frac{\eta}{\Pr[E]}
		\\&\leq \frac{1}{\Pr[E]} \cdot \E_{\rho\sim \mc R}\sqbrac{\ind\sqbrac{\lnot \Lambda(\rho)}} + \frac{\eta}{\alpha}
		\leq \frac{\gamma}{\alpha} + \frac{\eta}{\alpha} \leq o(1). \qedhere
	\end{align*}
\end{proof}

Second, we show that the distribution obtained by sampling $\rho\sim\mc R_i|E$ and $\tilde{X}\sim P_{X_i|E,E_\rho}$ is essentially the same as the distribution obtained by sampling $\tilde{X}\sim P_{X_i|E}$ and $\rho\sim\mc R_i|(E,X_i=\tilde{X})$. Formally, we have:
\begin{lemma}\label{lemma:pairwise_swap_dist}
	For every $\tilde{x}\in \supp(Q)$, it holds that $\Pr\sqbrac{E,X_i=\tilde{x}} \geq \Omega(\alpha) > 2\eta$; in particular, this implies the distribution $\mc R_i|(E,X_i=\tilde{x})$ is well-defined.
	
	Moreover, we have
	\[ \sum_{\substack{\rho\in \supp(\mc R_i),\\ \tilde{x}\in \supp(Q)}} \Big|\  (\mc R_i|E)[\rho]\cdot \Pr[X_i=\tilde{x}\st E,E_\rho] - \Pr[X_i=\tilde{x}\st E]\cdot (\mc R_i|E,X_i=\tilde{x})[\rho]\ \Big| \leq\frac{2\abs{\mc X}\cdot \eta}{\alpha} \leq o(1) .\]
\end{lemma}
\begin{proof}
Observe that by Lemma~\ref{lemma:cond_grr_prop}, we have \[\Norm{\E_{\rho\sim \mc R_i|E}P_{X_i|E,E_{\rho}} - P_{X_i|E} }_1 \leq \frac{2\eta}{\alpha} .\]
This implies
\[
	\Norm{P_{X_i|E}-Q}_1 \leq \frac{2\eta}{\alpha} + \E_{\rho\sim \mc R_i|E}\Norm{P_{X_i|E,E_{\rho}}-Q} \leq \frac{2\eta}{\alpha}+\beta \leq o(1).
\]
In particular, this implies that for every $\tilde{x}\in \supp(Q)$,
\[ \Pr\sqbrac{X_i=\tilde{x}, E} \geq \alpha\cdot \brac{\Pr_Q[\tilde{x}] -o(1)} = \Omega(\alpha) > 2\eta. \]
Hence, the distribution $\mc R_i|(E,X_i=\tilde{X})$ is well-defined for each $\tilde{x}\in \supp(Q)$.

Now, by Lemma~\ref{lemma:cond_grr_prop}, we have
\begin{align*}
		&\sum_{\substack{\rho\in \supp(\mc R_i),\\ \tilde{x}\in \supp(Q)}} \Big|\  (\mc R_i|E)[\rho]\cdot \Pr[X_i=\tilde{x}\st E,E_\rho] - \Pr[X_i=\tilde{x}\st E]\cdot (\mc R_i|E,X_i=\tilde{x})[\rho]\ \Big| 
		\\& \leq  \sum_{\rho,\tilde{x}} \abs{ \frac{\mc R_i[\rho]\cdot \Pr[E|E_\rho]}{\Pr[E]}\cdot \Pr[X_i=\tilde{x}\st E,E_\rho] - \Pr[X_i=\tilde{x}\st E]\cdot \frac{\mc R_i[\rho]\cdot \Pr[E,X_i=\tilde{x}|E_\rho]}{\Pr[E,X_i=\tilde{x}]}}
		\\&\qquad+\sum_{\rho,\tilde{x}} (\mc R_i|E)[\rho]\cdot \frac{\eta}{\Pr[E]}\cdot \Pr[X_i=\tilde{x}\st E,E_\rho]
		\\&\qquad\qquad+\sum_{\rho,\tilde{x}} (\mc R_i|E,X_i=\tilde{x})[\rho]\cdot \frac{\eta}{\Pr[E,X_i=\tilde{x}]}\cdot \Pr[X_i=\tilde{x}\st E] 
		\\&\leq 0 + \frac{\eta}{\Pr[E]} + \frac{\abs{\mc X}\cdot \eta}{\Pr[E]}
		\leq \frac{2\abs{\mc X}\cdot \eta}{\alpha}. \qedhere
\end{align*}
\end{proof}

With the above two lemmas, we are ready to complete the proof of Lemma~\ref{lemma:pairwise_fixed_i_rr_term}, assuming the following \emph{approximate-independence lemma}, which we shall prove later:\footnote{Note that the conditional generalized random restriction in this lemma is well-defined by Lemma~\ref{lemma:pairwise_swap_dist}.}

\begin{lemma}\label{lemma:pairwise_ques_ans_rr_term}
	For every $\tilde{x}\in \mc X,\ \tilde{a}\in \mc A$ such that $\set{x\in \mc X^{\otimes n}:x_i=\tilde{x},x\in E}\not=\emptyset$, it holds that
	\begin{align*}
		& \E_{\rho\sim \mc R_i|E,X_i=\tilde{x}}   \sqbrac{ \prod_{j=1}^k \Pr\sqbrac{A_i^j=\tilde{a}^j | X^j\in E^j\cap E_\rho^j,\ X_i^j=\tilde{x}^j} \cdot \ind\sqbrac{\Lambda(\rho)}} 
		\\&\qquad\qquad\leq\ \Pr\sqbrac{A_i=\tilde{a} \st E,\ X_i=\tilde{x}} + \frac{2\eta+8\epsilon}{\Pr\sqbrac{ E,X_i=\tilde{x}}}.
	\end{align*}
\end{lemma}

\begin{proof}[Proof of Lemma~\ref{lemma:pairwise_fixed_i_rr_term}]
	Observe that  when $\rho\sim\mc R_i|E$ and $\tilde{X}\sim P_{X_i|E,E_\rho}$, the consistency condition (in the definition of the embedding strategy) holds almost surely for each player $j\in [k]$, and so they answer $\tilde{A}^j \sim P_{A_i^j \st X^j\in E^j\cap E_\rho^j,\ X^j_i=\tilde{X}^j}$.
	This, along with Lemma~\ref{lemma:pairwise_rand_is_good}, gives that
	\begin{align*}
	& \E_{\rho\sim \mc R_i|E} \E_{\tilde{X}\sim P_{X_i|E,E_\rho}} \sqbrac{\Pr[L\st \tilde{X},\rho]}
	\\&\quad=\ \E_{\rho\sim \mc R_i|E} \E_{\tilde{X}\sim P_{X_i|E,E_\rho}} \sum_{\tilde a \in \mc A: V(\tilde{X},\tilde a)=0} \sqbrac{ \prod_{j=1}^k \Pr\sqbrac{A_i^j=\tilde{a}^j | X^j\in E^j\cap E_\rho^j,\ X_i^j=\tilde{X}^j} }
	\\&\quad\leq\ \E_{\rho\sim \mc R_i|E} \E_{\tilde{X}\sim P_{X_i|E,E_\rho}} \sum_{\tilde a \in \mc A: V(\tilde{X},\tilde a)=0} \sqbrac{ \prod_{j=1}^k \Pr\sqbrac{A_i^j=\tilde{a}^j | X^j\in E^j\cap E_\rho^j,\ X_i^j=\tilde{X}^j}\cdot \ind\sqbrac{\Lambda(\rho)} } + o(1).
	\end{align*}
	Now, by Lemma~\ref{lemma:pairwise_swap_dist}, the above is at most
	\begin{align*}
	&\ \E_{\tilde{X}\sim P_{X_i|E}}\E_{\rho\sim \mc R_i|E,X_i=\tilde{X}}  \sum_{\tilde a \in \mc A: V(\tilde{X},\tilde a)=0} \sqbrac{ \prod_{j=1}^k \Pr\sqbrac{A_i^j=\tilde{a}^j | X^j\in E^j\cap E_\rho^j,\ X_i^j=\tilde{X}^j}\cdot \ind\sqbrac{\Lambda(\rho)} } + o(1)
	\\&\quad\leq\ \E_{\tilde{X}\sim P_{X_i|E}}\sum_{\tilde a \in \mc A: V(\tilde{X},\tilde a)=0}\E_{\rho\sim \mc R_i|E,X_i=\tilde{X}}   \sqbrac{ \prod_{j=1}^k \Pr\sqbrac{A_i^j=\tilde{a}^j | X^j\in E^j\cap E_\rho^j,\ X_i^j=\tilde{X}^j}\cdot \ind\sqbrac{\Lambda(\rho)} } + o(1)
	\end{align*}
	Using Lemma~\ref{lemma:pairwise_ques_ans_rr_term}, this is at most
	\begin{align*}
	& \E_{\tilde{X}\sim P_{X_i|E}} \sum_{\tilde a \in \mc A: V(\tilde{X},\tilde a)=0} \sqbrac{\Pr\sqbrac{A_i=\tilde{a} \st E,\ X_i=\tilde{x}} + \frac{2\eta+8\epsilon}{\Pr[E,X_i=\tilde{X}]}} + o(1)
	\\&\quad= \Pr\sqbrac{\Lose_i| E} + \sum_{\tilde{x}\in \mc X} \sum_{\tilde a \in \mc A: V(\tilde{x},\tilde a)=0}\Pr[X_i=\tilde{x}\st E]\cdot \frac{2\eta+8\epsilon}{\Pr\sqbrac{ E,X_i=\tilde{x}}} + o(1)
	\\&\quad\leq \Pr\sqbrac{\Lose_i| E} +\frac{(2\eta+8\epsilon)\cdot \abs{\mc X}\cdot \abs{\mc A}}{\Pr\sqbrac{ E}} + o(1).
	\\&\quad\leq  \Pr\sqbrac{\Lose_i| E} + o(1).
	\end{align*}
	In the last inequality, we used $\Pr[E]\geq \alpha$, and $\eta,\epsilon\leq o(\alpha)$.
\end{proof}

\subsubsection{Approximate Independence Under Random Restriction
}
Finally, in the remainder of this subsection, we prove Lemma~\ref{lemma:pairwise_ques_ans_rr_term}.
For this, we fix some $\tilde{x}\in \mc X,\ \tilde{a}\in \mc A$ satisfying $\set{x\in \mc X^{\otimes n}:x_i=\tilde{x},x\in E}\not=\emptyset$.

We start by observing that under the event $\Lambda(\rho)$, our functions satisfy an approximate independence property, by the inverse theorem.

\begin{lemma}\label{lemma:pairwise_using_bklm_corr}
	Let $\rho\in \supp(\mc R_i)$ be such that the event $\Lambda(\rho)$ holds.
	Then,
	\[ \prod_{j=1}^k \Pr\sqbrac{A_i^j = \tilde{a}^j, X^j\in E^j | X_{i}^j=\tilde{x}^j,\ X^j\in E_\rho^j} \leq \Pr\sqbrac{A_i=\tilde{a},\ X\in E| X_i=\tilde{x},\ X\in E_\rho} + \epsilon, \]
	\[ \prod_{j=1}^k \Pr\sqbrac{X^j\in E^j | X_{i}^j=\tilde{x}^j,\ X^j\in E_\rho^j} \geq \Pr\sqbrac{X\in E| X_i=\tilde{x},\ X\in E_\rho} - \epsilon. \]
\end{lemma}
\begin{proof}
	Let $\rho$ be such that the event $\Lambda(\rho)$ holds, and let $m=m(\rho)\geq n^{\exp(-1/\gamma^4)}$  be the number of free coordinates in $\rho$.
	By Corollary~\ref{corr:pair_conn_bklm} and the definition of the event $\Lambda(\rho)$, it holds that
	 \[ \prod_{j=1}^k \E_{Y^j\sim (Q^j)^{\otimes m}}\sqbrac{ (f_{i,\tilde{x}^j,\tilde{a}^j}^j)_{\rho}(Y^j)} \leq \E_{Y\sim Q^{\otimes m}}\sqbrac{ \prod_{j=1}^k (f_{i,\tilde{x}^j,\tilde{a}^j}^j)_{\rho}(Y^j)} + \epsilon, \]
	\[ \prod_{j=1}^k \E_{Y^j\sim (Q^j)^{\otimes m}}\sqbrac{ (F_{i,\tilde{x}^j}^j)_{\rho}(Y^j)} \geq \E_{Y\sim Q^{\otimes m}}\sqbrac{ \prod_{j=1}^k(F_{i,\tilde{x}^j}^j)_{\rho}(Y^j)} - \epsilon. \]
	We used that each of the functions  is $(\sqrt{m}, \gamma)$-product pseudorandom, and hence also $(\delta m, \delta)$-product pseudorandom, since $\sqrt{m}\leq \delta m$ and $\gamma \leq \delta$.

	By the definitions of the functions $(f_{i,\tilde{x}^j,\tilde{a}^j}^j)_{j\in [k]},\ (F_{i,\tilde{x}^j}^j)_{j\in [k]}$, we have
	\[ \E_{Y\sim Q^{\otimes m}}\sqbrac{ \prod_{j=1}^k(f_{i,\tilde{x}^j,\tilde{a}^j}^j)_{\rho}(Y^j)} = \Pr\sqbrac{A_i = \tilde{a}, X\in E | X_{i}=\tilde{x},\ X\in E_\rho},\]
	\[ \E_{Y\sim Q^{\otimes m}}\sqbrac{ \prod_{j=1}^k(F_{i,\tilde{x}^j}^j)_{\rho}(Y^j)} = \Pr\sqbrac{X\in E| X_i=\tilde{x},\ X\in E_\rho},\]
	and for every $j\in [k]$,
	\[ \E_{Y^j\sim (Q^j)^{\otimes m}}\sqbrac{ (f_{i,\tilde{x}^j,\tilde{a}^j}^j)_{\rho}(Y^j)} =  \Pr\sqbrac{A_i^j = \tilde{a}^j, X^j\in E^j | X_{i}^j=\tilde{x}^j,\ X^j\in E_\rho^j}, \]
	\[ \E_{Y^j\sim (Q^j)^{\otimes m}}\sqbrac{ (F_{i,\tilde{x}^j}^j)_{\rho}(Y^j)} =  \Pr\sqbrac{X^j\in E^j | X_{i}^j=\tilde{x}^j,\ X^j\in E_\rho^j}. \]
	Plugging these into the above inequalities, we obtain the desired result.	
\end{proof}

Next, we complete the proof:

\begin{proof}[Proof of Lemma~\ref{lemma:pairwise_ques_ans_rr_term}]
	We have
	\begin{align*}
		 &\E_{\rho\sim \mc R_i|E,X_i=\tilde{x}}   \sqbrac{ \prod_{j=1}^k \Pr\sqbrac{A_i^j=\tilde{a}^j | X^j\in E^j\cap E_\rho^j,\ X_i^j=\tilde{x}^j} \cdot \ind\sqbrac{\Lambda(\rho)}} 
		 \\&\quad=  \E_{\rho\sim \mc R_i|E,X_i=\tilde{x}}   \sqbrac{ \frac{\prod_{j=1}^k \Pr\sqbrac{A_i^j=\tilde{a}^j, X^j\in E^j | X^j\in E_\rho^j,\ X_i^j=\tilde{x}^j}}{\prod_{j=1}^k \Pr\sqbrac{X^j\in E^j | X^j\in  E_\rho^j,\ X_i^j=\tilde{x}^j}} \cdot \ind\sqbrac{\Lambda(\rho)}}
	\end{align*}
	
	By Lemma~\ref{lemma:pairwise_using_bklm_corr}, Lemma~\ref{lemma:math_division}, and Lemma~\ref{lemma:cond_grr_prop} the above is at most
	\begin{align*}
		& \E_{\rho\sim \mc R_i|E,X_i=\tilde{x}}\sqbrac{\frac{\Pr\sqbrac{A_i=\tilde{a},\ X\in E\ |\ X_i=\tilde{x}, X\in E_\rho} + 4\epsilon}{\Pr\sqbrac{X\in E\ |\ X_i=\tilde{x}, X\in E_\rho }}}
		\\&\quad= \E_{\rho\sim \mc R_i|E,X_i=\tilde{x}}\sqbrac{\Pr\sqbrac{A_i=\tilde{a}\ | X_i=\tilde{x}, E, E_\rho}
		+ \frac{4\epsilon}{\Pr\sqbrac{E\ |\ X_i=\tilde{x},E_\rho }}}
		\\&\quad\leq \Pr\sqbrac{A_i=\tilde{a}\ |\ E, X_i=\tilde{x}} + \frac{2\eta}{\Pr\sqbrac{ E,X_i=\tilde{x}}}+ \E_{\rho\sim \mc R_i|E,X_i=\tilde{x}}\sqbrac{\frac{4\epsilon}{\Pr\sqbrac{E\ |\ X_i=\tilde{x},  E_\rho }}}.
	\end{align*}
	The last term can now be bounded as
	\begin{align*}
		&\E_{\rho\sim \mc R_i|E,X_i=\tilde{x}}\sqbrac{\frac{4\epsilon}{\Pr\sqbrac{E\ |\ X_i=\tilde{x},  E_\rho }}}
		\\&\qquad = \E_{\rho\sim \mc R_i}\sqbrac{\Pr[E,X_i=\tilde{x}\st E_\rho]\cdot \frac{4\epsilon}{\Pr\sqbrac{E\ |\ X_i=\tilde{x},  E_\rho }} }\cdot \frac{1}{ \E_{\rho'\sim \mc R_i}\sqbrac{\Pr[E,X_i=\tilde{x}\st E_{\rho'}]}}
		\\&\qquad = 4\epsilon\cdot \frac{ \E_{\rho\sim \mc R_i}\sqbrac{\Pr[X_i=\tilde{x}\st E_{\rho}]}}{ \E_{\rho'\sim \mc R_i}\sqbrac{\Pr[E,X_i=\tilde{x}\st E_{\rho'}]}}
		\\&\qquad\leq  4 \epsilon \cdot \frac{1}{\Pr[E,X_i=\tilde{x}] - \eta}.
		\\&\qquad\leq  8 \epsilon \cdot \frac{1}{\Pr[E,X_i=\tilde{x}]}.
	\end{align*}
	We used Lemma~\ref{lemma:pairwise_swap_dist} to say that $\Pr[E,X_i=\tilde{x}] \geq \Omega(\alpha) > 2\eta$.
\end{proof}


\subsection{Ensuring That Random Restrictions Don't Give Too Much Information}\label{sec:pairwise_rr_information}

In this subsection, we show how to obtain the pseudorandomness assumption (see \eqref{eq:pseudo_beta}) in the above section, via a generalized random restriction.
The following lemma says that if some $\mc{R}_i$ changes the distribution of $X_i$ conditioned on $E$ by a lot, then in fact conditioning on $\rho\sim \mc{R}_i,\ X_i\sim Q$ increases the $\ell_2$ energy of $E$.

\begin{lemma}\label{lemma:pairwise_rr_information_iterate}
	For some sufficiently large $n$, consider the game  $\mc G^{\otimes n} = (\mc X^{\otimes n}, \mc A^{\otimes n}, P=Q^{\otimes n},  V^{\otimes n})$, and let $X= (X^1,\dots,X^k)$ be the random variable denoting the questions to the $k$ players in the game $\mc G^{\otimes n}$.
Let $E=E^1\times \dots \times E^k \subseteq \mc X^{\otimes n}$ be a product event with $\alpha:= \Pr_{Q^{\otimes n}}[E]$.

	Let $i\in [n]$, and let $\mc R_i$ be any $(m,\epsilon)$-generalized random restriction on $\mc X^{\otimes n-1}$ (on coordinates $[n]\setminus \set{i}$),\footnote{We shall also regard $\mc R_i$ as a random restriction on $\mc X^{\otimes n}$ which does nothing to coordinate $i$.}, with $\epsilon<\alpha$, and such that 
	\[ \E_{\rho\sim \mc R_i|E}\Norm{P_{X_i|E,E_\rho}-Q}_1 \geq \beta,\footnote{Note that the conditional random restriction $\mc R_i|E$ is well-defined when $\epsilon < \alpha$.}\]
	for some $0< \beta \leq 1$.
	
	Let $\mc R$ be the generalized random restriction on $\mc X^{\otimes n}$ defined as follows:  choose $\rho\sim \mc R_i$, $\tilde{x}\sim Q$; output the generalized restriction $\rho'$ which performs $\rho$ on coordinates $[n]\setminus \set{i}$, and fixes the input value in coordinate $i$ to $\tilde{x}$.
	Then, it holds:
	\begin{enumerate}
		\item $\mc R$ is a $(m,\epsilon)$-generalized random restriction on $\mc X^{\otimes n}$.
		\item The conditional mass of $E$ under $E_{\rho'}$ has non-trivially increased variance, i.e., \[ \E_{\rho'\sim 
		\mc R}\sqbrac{\Pr[E|E_{\rho'}]^2} \geq \alpha^2\brac{1+\beta^2 - \frac{6\epsilon}{\alpha}}.\]
	\end{enumerate}
\end{lemma}
Before we give the formal proof, we explain why this should hold. 
Since $\rho\sim \mc R_i$ does not act on coordinate $i$, we know that just conditioning on $\rho$ (ignoring $E$) should not affect the marginal distribution of $X_i$.
However, adding in $E$ does change it.
This means that restricting by $\rho\sim \mc{R}_i,\ X_i\sim Q$ should split\footnote{The term \emph{split} is justified by Property 2 of Definition~\ref{defn:grr}.} the mass of $E$ in an uneven manner, which then increases the $\ell_2$ energy.

\begin{proof}
	It follows by definitions that $\mc R$ is a $(m,\epsilon)$-generalized random restriction on $\mc X^{\otimes n}$.
	
	By Lemma~\ref{lemma:cond_grr_prop}, we have
	\begin{align*}
		\beta &\leq \E_{\rho\sim \mc R_i|E}\Norm{P_{X_i|E,E_\rho}-Q}_1
		\\&\leq \sum_{\rho} \mc R_i[\rho]\cdot \frac{\Pr\sqbrac{E|E_\rho}}{\Pr[E]}\cdot \Norm{P_{X_i|E,E_\rho}-Q}_1 + \frac{2\epsilon}{\Pr[E]}
		\\&= \sum_{\tilde{x},\rho} \mc R_i[\rho]\cdot \frac{\Pr\sqbrac{E|E_\rho}}{\alpha}\cdot \abs{\Pr\sqbrac{X_i=\tilde{x}\st E, E_\rho}-Q[\tilde{x}]} + \frac{2\epsilon}{\alpha}
		\\&= \frac{1}{\alpha}\cdot \E_{\rho\sim \mc R_i}\E_{\tilde{x}\sim Q} \abs{\frac{\Pr\sqbrac{E, X_i=\tilde{x}|E_\rho}}{ Q[\tilde{x}]} - \Pr\sqbrac{E|E_\rho}}+ \frac{2\epsilon}{\alpha}.
	\end{align*}
	Since $\rho$ only acts on coordinates $[n]\setminus\set{i}$, it holds that $\Pr\sqbrac{E, X_i=\tilde{x}|E_\rho} = \Pr[E|E_\rho, X_i=\tilde{x}]\cdot Q[\tilde{x}]$, and hence the above gives
	\begin{align*}
		\beta\alpha - 2\epsilon \leq \E_{\rho\sim \mc R_i}\E_{\tilde{x}\sim Q}\abs{\Pr[E|E_\rho, X_i=\tilde{x}]-\Pr[E|E_\rho]}.
	\end{align*}
	Using Cauchy-Schwarz, we get
	\begin{align*}
		\brac{\beta\alpha - 2\epsilon}^2 &\leq \E_{\rho\sim \mc R_i}\E_{\tilde{x}\sim Q}\abs{\Pr[E|E_\rho, X_i=\tilde{x}]-\Pr[E|E_\rho]}^2
		\\&= \E_{\rho\sim \mc R_i}\sqbrac{\E_{\tilde{x}\sim Q}\Pr[E|E_\rho, X_i=\tilde{x}]^2 + \Pr[E|E_\rho]^2  - 2\E_{\tilde{x}\sim Q} \Pr[E|E_\rho, X_i=\tilde{x}]\cdot \Pr[E|E_\rho]}
		\\&= \E_{\rho\sim \mc R_i}\sqbrac{\E_{\tilde{x}\sim Q}\Pr[E|E_\rho, X_i=\tilde{x}]^2 - \Pr[E|E_\rho]^2}
		\\&\leq \E_{\rho'\sim 
		\mc R}\sqbrac{\Pr[E|E_{\rho'}]^2} - \brac{\E_{\rho\sim \mc R_i}\Pr[E|E_\rho]}^2
		\\&\leq \E_{\rho'\sim 
		\mc R}\sqbrac{\Pr[E|E_{\rho'}]^2} - \brac{\alpha-\epsilon}^2.
	\end{align*}
	Rearranging, we get
	\[ \E_{\rho'\sim 
		\mc R}\sqbrac{\Pr[E|E_{\rho'}]^2} \geq (\beta\alpha-2\epsilon)^2+(\alpha-\epsilon)^2 \geq \alpha^2(1+\beta^2)-6\epsilon\alpha. \qedhere\]
\end{proof}

With the above, we state and prove the main result of this subsection. This follows by repeatedly iterating Lemma \ref{lemma:pairwise_rr_information_iterate} while some $\mc{R}_i$ violates \eqref{eq:pseudo_beta}.

\begin{lemma}\label{lemma:pairwise_rr_information}
	Let $n\in \N$ be sufficiently large.
	Let $\gamma\in (0,1)$ be such that $\gamma \geq \frac{1}{\log\log\log\log n}$;
	let $\alpha,\beta,\kappa\in (0,1)$ be parameters, such that $\kappa, \alpha, \beta \geq \gamma$.
	
	Consider the game  $\mc G^{\otimes n} = (\mc X^{\otimes n}, \mc A^{\otimes n}, Q^{\otimes n}, V^{\otimes n})$, and fix any strategy for the $k$ players in this game.
	Let $E=E^1\times \dots \times E^k \subseteq  \mc X^{\otimes n}$ be a product event with $\Pr_{Q^{\otimes n}}[E]\geq \alpha$.
	Then, there exists a $(n^{\exp(-1/\gamma^{10})}, n^{-\exp(-1/\gamma^{10})})$-generalized random restriction $\mc R$ on $\mc X^{\otimes n}$, such that with probability at least $1-\kappa$ over $\rho\sim \mc R|E$,\footnote{This conditional random restriction is well defined as $ n^{-\exp(-1/\gamma^{10})} < 2^{-\sqrt{\log n}} < \gamma\leq \alpha \leq \Pr[E]$.} it holds that:

	Let $m=m(\rho)$ be the number of free coordinates in $\rho$.
	Consider the game $\mc G^{\otimes n}$ when the inputs to the $k$-players are conditioned to be in the set $E_\rho$; this is the same as $\mc G^{\otimes m}= (\mc X^{\otimes m}, \mc A^{\otimes m},P=Q^{\otimes m}, V^{\otimes m})$.
	Let $E'\subseteq \mc X^{\otimes m}$ be the corresponding restricted event, and also consider restricted strategies for the game $\mc G^{\otimes m}$, as in Definition~\ref{defn:game_rep_gr}.
	Let $X$ be the random variable denoting the questions to the $k$ players in this game $\mc G^{\otimes m}$.
	For every coordinate $i\in [m]$ in this game, let $\mc R_i$ be the $(\frac{1}{\eta}, \eta)$-generalized random restriction as in Lemma~\ref{lemma:pairwise_pseud_fn_dist} (with respect to $E'$, the above player strategies, and the parameter $\gamma$;\footnote{Note that this is well-defined since $m\geq n^{-\exp(-1/\gamma^{10})}\geq  2^{\sqrt{\log n}}$, and hence $1/\gamma \leq o(\log \log m)^{1/4}$.} and $\eta = \eta(m,\gamma)= m^{-\exp(-1/\gamma^4)}$).
	Then, it holds that
	\begin{enumerate}
		\item $\Pr_{Q^{\otimes m}}[E'] = \Pr_{Q^{\otimes n}}[E|E_\rho]\geq \frac{\kappa}{4}\cdot \Pr_{Q^{\otimes n}}[E]$.
		\item For every $i\in [m]$ it holds:
			\begin{equation}\label{eqn:pairwise_rr_info_good}
				\E_{\rho'\sim \mc R_i|E'}\Norm{P_{X_i|E',E_{\rho'}}-Q}_1 \leq \beta.\footnote{This conditional random restriction is well-defined as $\eta\leq 2^{-\sqrt{\log n}}$ and $\Pr[E'] \geq \frac{\kappa}{4}\cdot \Pr[E] \geq \frac{\gamma^2}{4} > \eta$ .}
			\end{equation}
	\end{enumerate}
\end{lemma}
\begin{proof}
	The proof proceeds via an iterative argument; we start with the generalized random restriction $\mc R^{(0)}$ on $\mc X^{\otimes n}$ that does nothing, and in each step refine it.
	We shall use a progress measure defined as follows: for any generalized random restriction $\mc R$ on $\mc X^{\otimes n}$, define \[ \mc Z(\mc R): = \E_{\rho\sim \mc R}\sqbrac{\Pr\sqbrac{E \st E_\rho}^2}. \]
	This satisfies $\mc Z\brac{\mc R^{(0)}} = \Pr[E]^2 \geq \alpha^2$.
		
	Let $T = \big\lceil \frac{8}{\gamma^6} \big\rceil$; for $t=1,2,\dots, T$, we define the random restriction $\mc R^{(t)}$ in the following manner:
	\begin{enumerate}
		\item Choose $\rho\sim \mc R^{(t-1)}$. Let $m=m(\rho)$ be the number of free coordinates in $\rho$, let $\eta = \eta(m(\rho), \gamma) = m^{-\exp(-1/\gamma^4)}$ be as in Lemma~\ref{lemma:pairwise_pseud_fn_dist}, let $E'\subseteq \mc X^{\otimes m}$ be the restriction of $E$ corresponding to $\rho$, and also consider restricted strategies for the game $\mc G^{\otimes m}$ as in Definition~\ref{defn:game_rep_gr}.
		\item We say that $\rho$ is bad if $\Pr[E'] =\Pr[E|E_\rho] \geq \frac{\kappa}{4}\cdot \Pr[E]$, and Equation~\ref{eqn:pairwise_rr_info_good} fails for some coordinate $i\in [m]$; else, we say it is good.
		\begin{enumerate}
			\item If $\rho$ is good, we do nothing and output $\rho$.
			\item\label{step:increment_step_rrinfo} Otherwise, $\Pr[E']\geq \frac{\kappa}{4}\cdot \Pr[E]$ and there exists $i\in [m]$ such that Equation~\ref{eqn:pairwise_rr_info_good} fails. Now, we apply Lemma~\ref{lemma:pairwise_rr_information_iterate} with respect to $\mc R_i$ to find a relevant $(\frac{1}{\eta}, \eta)$-generalized random restriction $\mc R_\rho$ on $\mc X^{\otimes m}$. Choose $\rho'\sim \mc R_\rho$ and output $\rho'\circ \rho$.
		\end{enumerate}
\end{enumerate}
	By induction, it is verified that for each $t=0,1,\dots,T$, the random restriction $\mc R^{(t)}$ is a $(m^{(t)},\epsilon^{(t)})$-generalized random restriction on $\Sigma^n$, with
	\[ m^{(t)} = \frac{1}{\eta(m^{(t-1)}, \gamma)}= \brac{m^{(t-1)}}^{\exp(-1/\gamma^4)} = n^{\exp(-t/\gamma^4)} \geq n^{\exp(-1/\gamma^{10})} \geq 2^{\sqrt{\log n}},\]
	and
	\begin{align*}
		\epsilon^{(t)} = \epsilon^{(t-1)} + \brac{m^{(t-1)}}^{-\exp(-1/\gamma^4)} \leq t\cdot n^{-\exp(-t/\gamma^4)} \leq n^{-\exp(-1/\gamma^{10})}\leq 2^{-\sqrt{\log n}}:=\epsilon.
	\end{align*}
	Note that this satisfies $\frac{1}{\gamma^4} \leq o(\log \log m^{(T)})^{1/4}\leq o(\log \log m(\rho))^{1/4}$ whenever Lemma~\ref{lemma:pairwise_pseud_fn_dist} is applied, as needed in the assumption of the lemma.
	Also, $\epsilon < \alpha$, so the distribution $\mc R^{(t)}|E$ is well-defined for every $t=0,1,\dots,T$.
	
	Now, suppose that at some point we found some $\rho$	 that is bad, and applied Step~\ref{step:increment_step_rrinfo} above to get $\mc R_\rho$; then, for  $\eta=\eta(m(\rho), \gamma)$, by Lemma~\ref{lemma:pairwise_rr_information_iterate}, we have
	\[ \E_{\rho'\sim \mc R_\rho}\sqbrac{ \Pr[E'|E_{\rho'}]^2} \geq \Pr[E']^2\brac{1+\beta^2-\frac{6\eta}{\Pr[E']}}.\]
	Since $\Pr[E'] \geq \frac{\kappa}{4}\cdot \Pr[E] \geq \frac{\gamma^2}{4}$ (as $\rho$ is bad), and $\eta \leq \epsilon =  2^{-\sqrt{\log n}}$, we have $\frac{6\eta}{\Pr[E']} \leq \frac{24\epsilon}{\gamma^2} \leq \frac{\gamma^2}{2} \leq \frac{\beta^2}{2}$, and
	\[ \E_{\rho'\sim \mc R_\rho}\sqbrac{ \Pr[E'|E_{\rho'}]^2} \geq \Pr[E']^2\brac{1+\frac{\beta^2}{2}}.\]
	This implies that for any $t=1,\dots, T$, we have
	\begin{align*}
	\mc Z(\mc R^{(t)}) &= \E_{\rho\sim \mc R^{(t)}}\sqbrac{\Pr[E|E_{\rho}]^2}
	\\&\geq \E_{\rho\sim \mc R^{(t-1)}}\sqbrac{\Pr[E|E_\rho]^2 \cdot \brac{1+\frac{\beta^2}{2}\cdot \ind\sqbrac{\rho\text{ is bad}} } } 
	\\&= \mc Z(\mc R^{(t-1)}) +\frac{\beta^2}{2}\cdot\E_{\rho\sim \mc R^{(t-1)}}\sqbrac{\Pr[E|E_\rho]^2 \cdot \ind\sqbrac{\rho\text{ is bad}} }
	\end{align*}

	Suppose, for the sake of contradiction, that $\mc R^{(t)}$ does not satisfy the statement of the lemma, for any $t=0,1,\dots,T$.
	By Lemma~\ref{lemma:cond_grr_prop}, we get
	\begin{align*}
		\kappa &\leq \Pr_{\rho\sim \mc R^{(t)}|E}\sqbrac{\Pr[E|E_\rho] < \frac{\kappa}{4}\cdot \Pr[E]\text{ or }\rho \text{ is bad}}
		\\&\leq \E_{\rho\sim \mc R^{(t)}}\sqbrac{\frac{\Pr[E|E_\rho]}{\Pr[E]}\cdot \brac{\ind\sqbrac{\Pr[E|E_\rho] < \frac{\kappa}{4}\cdot \Pr[E]} + \ind\sqbrac{\rho\text{ is bad}} }} + \frac{\epsilon}{\Pr[E]}
		\\&\leq \frac{\frac{\kappa}{4}\cdot \Pr[E]}{\Pr[E]} + \frac{1}{\Pr[E]}\cdot \E_{\rho\sim \mc R^{(t)}}\sqbrac{ \Pr[E|E_\rho]\cdot  \ind\sqbrac{\rho\text{ is bad}}} + \frac{\epsilon}{\Pr[E]}
		\\&\leq \frac{1}{\alpha}\cdot \E_{\rho\sim \mc R^{(t)}}\sqbrac{ \Pr[E|E_\rho]\cdot  \ind\sqbrac{\rho\text{ is bad}}} + \frac{\kappa}{4}+\frac{\epsilon}{\alpha}
		\\&\leq  \frac{1}{\alpha}\cdot \E_{\rho\sim \mc R^{(t)}}\sqbrac{ \Pr[E|E_\rho]\cdot  \ind\sqbrac{\rho\text{ is bad}}} + \frac{\kappa}{2}
	\end{align*}
	In the last inequality, we used $\kappa, \alpha \geq \gamma$ and $\epsilon \leq 2^{-\sqrt{\log n}}$.
	By Cauchy-Schwarz, this implies
	\[ \E_{\rho\sim \mc R^{(t)}}\sqbrac{\Pr[E|E_\rho]^2 \cdot \ind[\rho\text{ is bad}]} \geq \frac{\alpha^2\kappa^2}{4}.\]
	Plugging into the above, we get that for any $t=1,\dots,T$ it holds that
	\[ \mc Z(\mc R^{(t)}) \geq \mc Z(\mc R^{(t-1)}) + \frac{\alpha^2\beta^2\kappa^2}{8} \geq\mc Z(\mc R^{(0)})+\frac{\alpha^2\beta^2\kappa^2 t}{8} \geq \alpha^2+\frac{\alpha^2\beta^2\kappa^2 t}{8} \geq \alpha^2+\frac{\gamma^6 t}{8} .\]
	This is a contradiction for $t=T$.
	Hence, for some $t$, the generalized random restriction $\mc R^{(t)}$ satisfies the statement of the lemma.	
\end{proof}


\subsection{Combining Together: Hard Coordinates under Random Restrictions}\label{sec:pairwise_combining}

We combine the results in the above sections, and prove the following lemma.
It shows that in the game $\mc G^{\otimes n}$, conditioned on a large product event $E$, we can find hard coordinates after a generalized random restriction.

\begin{lemma}\label{lemma:pairwise_hard_coor_after_rr}
	Let $1\leq C\leq k^{O(k)}$ be a constant so that Corollary~\ref{corr:pair_conn_bklm} with respect to the probability space $(\mc X, Q)$ holds with $C$ logarithms.
	Let $n\in \N$ be sufficiently large; let $\alpha, \gamma\in (0,1)$ be such that $(\log\log\log\log n)^{-1}\leq \gamma \leq (\log\log\log \log \log n)^{-1}$, and $\alpha \geq \frac{1}{\log\log\cdots\log n}$ where the number of logarithms is $C+6$.
	Let $\kappa\in (0,1)$ be such that $\kappa\geq \alpha$.
	 
	Consider the game  $\mc G^{\otimes n} = (\mc X^{\otimes n}, \mc A^{\otimes n}, Q^{\otimes n}, V^{\otimes n})$, and fix any strategy for the $k$ players in this game.
	Let $E=E^1\times \dots \times E^k \subseteq  \mc X^{\otimes n}$ be a product event with $\Pr_{Q^{\otimes n}}[E]\geq \alpha$.
	Then, there exists a $\brac{n', \frac{1}{n'}}$-generalized random restriction $\mc R$ on $\mc X^{\otimes n}$, with $n' \geq n^{\exp(-1/\gamma^{10})}$, such that with probability at least $1-\kappa$ over $\rho \sim \mc R|E$,\footnote{This conditional random restriction is well-defined as $\frac{1}{n'} < \alpha$.} it holds that:
	
	Let $m=m(\rho)$ be the number of free coordinates in $\rho$.
	Consider the game $\mc G^{\otimes n}$ when the inputs to the $k$-players are conditioned to be in the set $E_\rho$; this is the same as $\mc G^{\otimes m}= (\mc X^{\otimes m}, \mc A^{\otimes m},P=Q^{\otimes m}, V^{\otimes m})$.
	Let $E'\subseteq \mc X^{\otimes m}$ be the corresponding restricted event, and also consider restricted strategies for the game $\mc G^{\otimes m}$, as in Definition~\ref{defn:game_rep_gr}.
	Let $X$ be the random variable denoting the questions to the $k$ players in this game $\mc G^{\otimes m}$, and let $A$ be the random variable denoting their answers with respect to the above strategies.
	For each $i\in [m]$, let $\Win_i$ be the event $V(X_i,A_i)=1$, denoting that the players win coordinate $i$ of the game.
	Then, we have:
	\begin{enumerate}
		\item $\Pr_{Q^{\otimes m}}[E'] = \Pr_{Q^{\otimes n}}[E|E_\rho]\geq \frac{\kappa}{4}\cdot \Pr_{Q^{\otimes n}}[E]$.
		\item For each $i\in [m]$,
		\[ \Pr_{Q^{\otimes m}}\sqbrac{\Win_i | E'} \leq \val(\mc G)+o(1) \leq 1-\Omega(1). \]
	\end{enumerate}	
\end{lemma}
\begin{proof}
	Let $\mc R$ be the $(n',\frac{1}{n'})$-generalized random restriction, for $n'=n^{\exp(-1/\gamma^{10})}$, as in Lemma~\ref{lemma:pairwise_rr_information}, with the choice $\beta = \gamma$;\footnote{It is fine to choose $\beta$ to be any anything that satisfies $\gamma\leq \beta\leq o(1)$.} note that the assumption $\alpha, \kappa\geq \gamma$ is satisfied.
	This satisfies that with probability $1-\kappa$ over $\rho\sim \mc R|E$:
	
	Let $m:=m(\rho)$, and let $E'$ and player strategies for $\mc G^{\otimes m}$ be defined as in the lemma statement. Then, we know
	\begin{enumerate}
		\item $\Pr[E'] = \Pr[E|E_\rho]\geq \frac{\kappa}{4}\cdot \Pr[E]$.
		\item For each $i\in [m]$, let $\mc R_i$ be the $(\frac{1}{\eta}, \eta)$-generalized random restriction as in Lemma~\ref{lemma:pairwise_pseud_fn_dist} (with respect to $E'$, the restricted player strategies, and the parameter $\gamma$; and $\eta = \eta(m,\gamma)= m^{-\exp(-1/\gamma^4)}$).
			Then, \[\E_{\rho'\sim \mc R_i|E'}\Norm{P_{X_i|E',E_{\rho'}}-Q}_1 \leq \beta.\]		
	\end{enumerate}	
	Now, the result will follow by applying Proposition~\ref{prop:pairwise_hard_coor_under_assum} on the game $\mc G^{\otimes m}$, with respect to the event $E'$. We verify that the parameter assumptions in Section~\ref{sec:pairwise_one_coord_hard_assum} hold, as follows:
	\begin{enumerate}
		\item $C$ is the constant for Corollary~\ref{corr:pair_conn_bklm} with respect to the space $(\mc X, Q)$.
		\item $\beta \leq o(1)$ holds as $\beta=\gamma$, and $m\to\infty$ as $n\to \infty$.
		\item We have $(\log\log\log\log n)^{-1}\leq \gamma \leq (\log\log\log \log \log n)^{-1}$, and $n\geq m\geq n^{\exp(-1/\gamma^{10})} \geq 2^{\sqrt{\log n}}$. Hence, it holds that $(\log\log\log m)^{-1} \leq \gamma \leq (\log\log\log\log\log m)^{-1}$.
		\item $\eta = \eta(m,\gamma) = m^{-\exp(-1/\gamma^4)} $ is as in Lemma~\ref{lemma:pairwise_pseud_fn_dist}, as required.
		\item We know $\alpha \geq \frac{1}{\log\log\cdots\log n}$ where the number of logarithms is $C+6$, and $\Pr[E] \geq \alpha$.
		Hence, 
		\[ \Pr[E']\geq \frac{\kappa}{4}\cdot \Pr[E] \geq \frac{\alpha^2}{4} \geq \frac{1}{4\cdot (\underbrace{\log\log\cdots\log}_{C+6} n)^2}\geq \frac{1}{(\underbrace{\log\log\cdots\log}_{C+5} m)^{1/2}}.\qedhere \]
	\end{enumerate}	
\end{proof}


\subsection{Final Induction}\label{sec:pairwise_induction}

In this section, we finally prove the main theorem via an inductive argument.
The goal is to prove that for sufficiently large $N\in \N$, it holds that $\val(\mc G^{\otimes N}) \leq \frac{1}{\underbrace{\log\log\cdots\log}_{O(1)} N}$.
For this purpose, we shall fix some large enough $N\in \N$ for the remainder of this section.
Also, we fix the following parameters:
\[ \alpha := \frac{1}{\underbrace{\log\log\cdots\log}_{C+7} N},\quad \gamma := \frac{1}{\log\log\log\log\log N},\quad c=\frac{1-\val(\mc G)}{2} \geq \Omega(1),\]
where $1\leq C\leq k^{O(k)}$ is a constant so that Corollary~\ref{corr:pair_conn_bklm} with respect to the space $(\mc X, Q)$ holds with $C$ logarithms.
We make the following definition:

\begin{definition}\label{defn:pairwise_game_val}
	For integer $1\leq n\leq N$, and $\alpha\in [0,1]$, we define the quantity $\val(n, \alpha)$ as \[ \val(n, \alpha):= \max_{n\leq n'\leq N} \max_E \val\brac{\mc G^{\otimes n'}|E}, \]
	where the second maximum is over product events $E = E^1\times \dots \times E^k \subseteq \mc X^{\otimes n'}$ satisfying $\Pr_{Q^{\otimes n'}}[E]\geq \alpha$, and $\mc G^{\otimes n'}|E:= (\mc X^{\otimes n'}, \mc A^{\otimes n'}, Q^{\otimes n'}|E, V^{\otimes n'})$ is the game $\mc G^{\otimes n'}$ with the questions to the $k$ players drawn conditioned on $E$.

	Note that by definition it holds that $\val(\mc G^{\otimes N}) = \val(N, 1)$.
\end{definition}

We start by proving a simple lemma, which was (implicitly) used in Lemma~\ref{lemma:prod_set_hard_coor}.
\begin{lemma}\label{lemma:pairwise_induction_condition_on_i}
	For any integer $n\leq N$, consider the game $\mc G^{\otimes n} = (\mc X^{\otimes n}, \mc A^{\otimes n}, Q^{\otimes n}, V^{\otimes n})$, and
	fix any strategy for the $k$ players in this game.
	For each $i\in [n]$, let $\Win_i$ be the event that this strategy wins the $i$\textsuperscript{th} coordinate of the game.
	Let $\Win = \land_{i=1}^n \Win_i$.
	
	Let $E=E^1\times \dots \times E^k \subseteq \mc X^{\otimes n}$ be a product event, and let $i\in [n]$.
	Then, for any $\theta\in [0,1]$, we have 
	\[ \Pr[\Win \st E] \leq  \Pr[\Win_i\st E]\cdot  \val(n, \theta \Pr[E]) + \abs{\mc X}\abs{\mc A}\theta.\]
\end{lemma}
\begin{proof}
	Let $X$ be the random variable denoting the questions to the players in $\mc G^{\otimes n}$, and let $A$ be the random variable denoting their answers with respect to the fixed strategy.
	
	Fix some $i\in [n]$, and let $Z = (X_i,A_i)$ be the random variable denoting the tuple of questions and answers in coordinate $i$.
	Let $\mc T = \set{z \in \mc X\times \mc A: V(z)=1}$ be the set of winning question and answer pairs in the base game $\mc G$, and let $\mc T'\subseteq \mc T$ consist of $z$ such that $\Pr[Z=z \st E] \geq \theta$.
	Then, it holds that
	\begin{align*}
			\Pr[\Win \st  E] &= \Pr[\Win \land\Win_i \st E]
			\\&= \sum_{z\in \mc T} \Pr[\Win \land (Z=z) \st E]
			\\&\leq \sum_{z\in \mc T'} \Pr[\Win \land (Z=z) \st E] + \abs{\mc X}\abs{\mc A}\theta 
			\\&= \sum_{z\in \mc T'}\Pr[Z=z\st E]\cdot  \Pr[\Win  \st E,Z=z] + \abs{\mc X}\abs{\mc A}\theta 
			\\&\leq  \sum_{z\in \mc T'}\Pr[Z=z\st E]\cdot  \val(n,\ \theta \Pr[E]) + \abs{\mc X}\abs{\mc A}\theta
			\\&\leq \Pr[\Win_i\st E]\cdot  \val(n, \theta \Pr[E]) + \abs{\mc X}\abs{\mc A}\theta.\qedhere
	\end{align*}
	We used that for every $z\in\mc T'$, it holds that $E,Z=z$ is a product event with measure $\Pr[E, Z=z]=\Pr[E]\cdot \Pr[Z=z|E]\geq \Pr[E]\cdot \theta$.
\end{proof}

Now, using Lemma~\ref{lemma:pairwise_hard_coor_after_rr}, we prove the following inductive lemma:

\begin{lemma}\label{lemma:pairwise_induction}
	Let $\theta\in (0,1)$ be a paramter satisfying $4\theta^3 \geq \alpha$.
	
	Let $n\in \N,\ \mu\in [0,1]$ be such that $2^{\sqrt{\log N}}\leq n\leq N$ and $\mu \geq \alpha$. Then, it holds that
	\[ \val(n, \mu) \leq (1-c)\cdot \val\brac{\big\lceil n^{\exp(-1/\gamma^{10})} \big\rceil, \theta^6\mu} + 6\abs{\mc X}\abs{\mc A}\cdot\theta^3.\]
\end{lemma}
\begin{proof}
	Let $\theta$ be as in the lemma statement.
	Let $2^{\sqrt{\log N}}\leq n\leq N$, and $\mu \geq \alpha$.
	Consider the game $\mc G^{\otimes n} = (\mc X^{\otimes n}, \mc A^{\otimes n}, Q^{\otimes n}, V^{\otimes n})$, and let $E = E^1\times \dots \times E^k \subseteq \mc X^{\otimes n}$ be a product event such that $\Pr_{Q^{\otimes n}}[E]\geq \mu$.
	Fix any strategy for the $k$ players for the game $\mc G^{\otimes n}$, and let $W$ be the event that the players win this game.
	It suffices to show\footnote{Formally, to actually prove the lemma one can apply this claim to an arbitrary $n'$ satisfying $n\leq n'\leq N$, and $E\subseteq \mc X^{\otimes n'}$, and then use that $\big\lceil n'^{\exp(-1/\gamma^{10})} \big\rceil \geq \big\lceil n^{\exp(-1/\gamma^{10})} \big\rceil$.} that 
	\[ \Pr[W|E]  \leq (1-c)\cdot \val\brac{\big\lceil n^{\exp(-1/\gamma^{10})} \big\rceil, \theta^6\mu} + 6\abs{\mc X}\abs{\mc A}\cdot\theta^3.\]
	
	First, we apply Lemma~\ref{lemma:pairwise_hard_coor_after_rr} (with $\kappa=4\theta^3$), to find a $\brac{\frac{1}{\eta}, \eta}$-generalized random restriction $\mc R$ on $\mc X^{\otimes n}$, with $\eta \leq n^{-\exp(-1/\gamma^{10})} \leq 2^{-(\log N)^{1/4}}$, such that with probability at least $1-4\theta^3$ over $\rho \sim \mc R|E$:	
	Let $m=m(\rho)$ be the number of free coordinates in $\rho$.
	Consider the game $\mc G^{\otimes n}$ when the inputs to the $k$-players are conditioned to be in the set $E_\rho$; this is the same as $\mc G^{\otimes m}= (\mc X^{\otimes m}, \mc A^{\otimes m},P=Q^{\otimes m}, V^{\otimes m})$.
	Let $E'\subseteq \mc X^{\otimes m}$ be the corresponding restricted event, and also consider restricted strategies for the game $\mc G^{\otimes m}$, as in Definition~\ref{defn:game_rep_gr}.
    Let $X$ be the random variable denoting the questions to the $k$ players in this game $\mc G^{\otimes m}$, and let $A$ be the random variable denoting their answers with respect to the above strategies.
	For each $i\in [m]$, let $\Win_i$ be the event $V(X_i,A_i)=1$ denoting that the players win coordinate $i$ of the game.
	Then, we have:
	\begin{enumerate}
		\item $\Pr[E'] = \Pr[E|E_\rho]\geq \theta^3\cdot \Pr[E] \geq \theta^3\mu$.
		\item For each $i\in [m]$,
		$\Pr_{Q^{\otimes m}}\sqbrac{\Win_i | E'} \leq \val(\mc G)+o(1) \leq 1-c. $
	\end{enumerate}	
 	Note that the lemma is applicable by the choice of parameters as:
 	\begin{enumerate}
 		\item $\kappa=4\theta^3 \geq \alpha$ by the lemma hypothesis.
 		\item Note that $2^{\sqrt{\log N}}\leq n\leq N$, and so $\frac{1}{2}\log\log N\leq \log\log n \leq \log\log N$.
 		\begin{enumerate}
 			\item $\alpha := \frac{1}{\log\log\cdots\log N}$ where the number of logarithms is $C+7$, and hence $\Pr[E]\geq \mu \geq \alpha\geq \frac{1}{\log\log\cdots\log n}$ where the number of logarithms is $C+6$.
 			\item $\gamma = \frac{1}{\log\log\log\log\log N}$ and hence $(\log\log\log\log n)^{-1}\leq \gamma \leq (\log\log\log \log \log n)^{-1}$.
 		\end{enumerate}
 	\end{enumerate}
 	
 	Now, let $\Gamma$ be the set of all $\rho$ for which the above properties hold.
 	Consider any $\rho\in \Gamma$ and let $m=m(\rho)$.
 	Denoting $\Win=\land_{i=1}^m \Win_i$, we have by Lemma~\ref{lemma:pairwise_induction_condition_on_i} (with parameters $m, \theta^3$), that
 	\begin{align*}
 		\Pr_{Q^{\otimes m}}[\Win \st E'] &\leq \Pr_{Q^{\otimes m}}\sqbrac{\Win_i | E'}\cdot \val(m,\ \theta^3\Pr[E']) + \abs{\mc X}\abs{\mc A}\theta^3
 		\\&\leq (1-c)\cdot \val(m,\  \theta^6\mu) + \abs{\mc X}\abs{\mc A}\theta^3
 		\\&\leq (1-c)\cdot \val\brac{\big\lceil n^{\exp(-1/\gamma^{10})} \big\rceil,\ \theta^6\mu} + \abs{\mc X}\abs{\mc A}\theta^3
 	\end{align*}
	Finally, by Lemma~\ref{lemma:cond_grr_prop}, in the game $\mc G^{\otimes n}$ we have that
	\begin{align*}
		\Pr_{Q^{\otimes n}}[W|E] &\leq \E_{\rho\sim \mc R|E}\sqbrac{\Pr[W\st E, E_{\rho}]} + \frac{2 \cdot 2^{-(\log N)^{1/4}}}{\Pr[E]}
		\\&\leq \E_{\rho\sim \mc R|E}\sqbrac{\Pr[W\st E, E_{\rho}]} + \theta^3
		\\&\leq (1-c)\cdot \val\brac{\big\lceil n^{\exp(-1/\gamma^{10})} \big\rceil,\  \theta^6\mu} + \abs{\mc X}\abs{\mc A}\theta^3 + \Pr_{\rho\sim \mc R|E}[\rho\not\in \Gamma]+ \theta^3
		\\&\leq (1-c)\cdot \val\brac{\big\lceil n^{\exp(-1/\gamma^{10})} \big\rceil,\  \theta^6\mu} + \abs{\mc X}\abs{\mc A}\theta^3 + 4\theta^3+ \theta^3,
		\\&\leq (1-c)\cdot \val\brac{\big\lceil n^{\exp(-1/\gamma^{10})} \big\rceil,\  \theta^6\mu} + 6\abs{\mc X}\abs{\mc A}\cdot\theta^3.
	\end{align*}
	We used that for each $\rho\in \Gamma$, the quantity $\Pr[W\st E, E_\rho]$ is bounded by the quantity $\Pr_{Q^{\otimes m}}[\Win|E']$ analyzed before.
\end{proof}

Next, we complete the proof of the main theorem:
\begin{proof}[Proof of Theorem~\ref{thm:pairconn_parrep}]
	The proof shall follow by applying Lemma~\ref{lemma:pairwise_induction} iteratively.
	We define the parameters
	\[ \theta := \frac{1}{\log\brac{\frac{1}{\alpha}}}=\frac{1}{\underbrace{\log\log\cdots\log}_{C+8} N},\quad T = \Big\lceil{\frac{2}{c} \log\log\brac{\frac{1}{\alpha}}}\Big\rceil. \]
	Note that this satisfies $4\theta^3\geq \alpha$.
	
	We show by induction that for every $t=0,1,2,\dots, T$, it holds that 
	\[ \val(\mc G^{\otimes N}) \leq (1-c)^t\cdot  \val\brac{\big\lceil N^{\exp(-t/\gamma^{10})} \big\rceil,\ \theta^{6t}} + 6\abs{\mc X}\abs{\mc A}\theta^3 t. \]
	The base case $t=0$ follows from the observation that $\val(\mc G^{\otimes N}) = \val(N, 1)$.
	
	For the inductive step, consider any $t=0,\dots, T-1$, and assume that the statement holds for $t$.		
	Let $n = \big\lceil N^{\exp(-t/\gamma^{10})} \big\rceil,\ \mu = \theta^{6t}$; observe that by the choice of parameters, these satisfy $2^{\sqrt{\log N}}\leq n\leq N$ and $\mu\geq \alpha$.
	By the inductive hypothesis, and Lemma~\ref{lemma:pairwise_induction}, we get
	\begin{align*}
		\val(\mc G^{\otimes N}) &\leq (1-c)^t\cdot  \val\brac{\big\lceil N^{\exp(-t/\gamma^{10})} \big\rceil,\ \theta^{6t}} + 6\abs{\mc X}\abs{\mc A}\theta^3 t
		\\&= (1-c)^t \cdot \val(n, \mu) + 6\abs{\mc X}\abs{\mc A}\theta^3 t
		\\&\leq (1-c)^t \cdot \brac{(1-c)\cdot \val\brac{ \big\lceil n^{\exp(-1/\gamma^{10})} \big\rceil,\ \theta^6\mu} + 6\abs{\mc X}\abs{\mc A}\theta^3} + 6\abs{\mc X}\abs{\mc A}\theta^3 t
		\\&\leq (1-c)^{t+1}  \cdot \val\brac{\mc G, \big\lceil N^{\exp(-(t+1)/\gamma^{10})} \big\rceil,\  \theta^{6(t+1)}} + 6\abs{\mc X}\abs{\mc A}\theta^3\cdot  (t+1),
	\end{align*}
as desired.

Finally, applying the above claim for $t=T$, we get
\begin{align*}
	\val(\mc G^{\otimes N})& \leq (1-c)^T\cdot  \val\brac{ \big\lceil N^{\exp(-T/\gamma^{10})} \big\rceil,\ \theta^{6T}} + 6\abs{\mc X}\abs{\mc A}\theta^3 T
	\\&\leq (1-c)^T + 6\abs{\mc X}\abs{\mc A}\cdot\theta^3\cdot T
	\\&\leq \brac{\log\brac{1/\alpha}}^{-2} + 6\abs{\mc X}\abs{\mc A}\cdot \brac{\log\brac{1/\alpha}}^{-3} \cdot \frac{4}{c} \log\log\brac{\frac{1}{\alpha}} 
	\\&\leq \brac{\log\brac{1/\alpha}}^{-1} = \frac{1}{\underbrace{\log\log\cdots\log}_{C+8} N}.
\end{align*}
Hence, the theorem holds with the constant $C+8 \leq k^{O(k)}$.	
\end{proof}

\appendix
\section{Some Useful Lemmas}

The following simple lemma appears as \cite[Lemma 8.4]{BKLM24a}.
\begin{lemma}\label{lemma:small_set_force_equality}
	Let $(\Sigma, \mu)$ be a finite probability space, and let $k\leq n\in \N$.
	For each $T\subseteq [n]$, let $E_T\subseteq \Sigma^n$ be the event defined as:
	\[ E_T = \set{x\in \Sigma^n : x_i=x_j \text{ for each }i,j\in T}.\]
	Now, let $S\subseteq [n]$ be any set, and let $1\leq k\leq \sqrt{\abs{S}}$ be an integer.
	Let $\nu$ be the distribution on $\Sigma^n$ obtained sampling $T\subseteq S, \abs{T}=k$ uniformly at random, and then sampling $x\sim \mu^{\otimes n}|E_T$.
	Then, it holds that \[ \Norm{\nu-\mu^{\otimes n}}_1 \leq \frac{Ck}{\sqrt{\abs{S}}} ,\]
	where $C$ is a constant depending on $\mu$.  
\end{lemma}
\begin{proof}
	For a fixed $a\in \Sigma$, let $\nu'$ be the distribution obtained sampling $T\subseteq S, \abs{T}=k$, and then sampling $x$ from $\mu^{\otimes n}$ conditioned on $x_i=a$ for each $i\in T$.
	It suffices to show $\Norm{\nu'-\mu^{\otimes n}}_1 \leq \frac{Ck}{\sqrt{\abs{S}}}$.
	This is proven by induction on $k$ as follows:
	\begin{enumerate}
		\item For $k=1$, this follows by direct calculation.
		\item For $k>1$, note that sampling $T\subseteq S$ can be done by sampling $T'\subseteq S,\ \abs{T'}=k-1$, and then sampling $t\in S\setminus T'$ and setting $T=T'\cup\set{t}$. The result follows by induction.\qedhere
	\end{enumerate}
\end{proof}

\begin{lemma}\label{lemma:complex_num_inequality}
	Let $x,y \in \C$ be such that $\abs{y}\leq 1$ and $\abs{x-y}\leq \epsilon$.
	Then, $\abs{x}^2 \geq \abs{y}^2 -2\epsilon.$
\end{lemma}
\begin{proof}
	We have
	\[ \abs{x}^2 = \abs{y+(x-y)}^2= \abs{y}^2+\abs{x-y}^2 - (y\cdot \overline{\brac{x-y}}-\overline{y}\cdot \brac{x-y})\geq \abs{y}^2+\epsilon^2 -2\abs{y}\epsilon \geq \abs{y}^2-2\epsilon.\qedhere \]
\end{proof}

\begin{lemma}\label{lemma:math_division}
	Let $x,y,a,b, \epsilon\in [0,1]$ be real numbers such that
	\[ b>0,\ y>0,\quad a\leq b,\ x\leq y,\quad a\leq x+\epsilon,\ b\geq y-\epsilon . \]
	Then,
		\[ \frac{a}{b} \leq \frac{x+4\epsilon}{y}. \]
\end{lemma}
\begin{proof}
	If $y\leq 2\epsilon$, the statement holds since $\frac{a}{b}\leq 1\leq 2\leq \frac{4\epsilon}{y} \leq \frac{x+4\epsilon}{y}$.
	Otherwise, if $y > 2\epsilon$, we have $y-\epsilon \geq \frac{y}{2}>0$, and hence
	\[ \frac{a}{b} - \frac{x}{y} \leq \frac{x+\epsilon}{y-\epsilon}-\frac{x}{y} = \frac{(x+y)\epsilon}{(y-\epsilon)y} \leq \frac{2y\epsilon}{\frac{y}{2}\cdot y} = \frac{4\epsilon}{y}. \qedhere\]
\end{proof}

\begin{fact} \label{fact:chernoff} (Chernoff Bounds, see~\cite{MU05} for reference)
	Let $X_1,\dots, X_n \in \set{0,1}$ be independent random variables each with mean $\mu$, and let $X = \sum_{i=1}^n X_i$. Then, for all $\delta \in (0,1)$, it holds that 
	\[\Pr\sqbrac{X \leq (1-\delta)\mu n} \leq e^{-\frac{\delta^2 \mu n}{2}},\]
	\[\Pr\sqbrac{X \geq (1+\delta)\mu n} \leq e^{-\frac{\delta^2 \mu n}{3}}.\]
\end{fact}


\section{An Inductive Parallel Repetition Criterion}\label{sec:inductive_crit}

The following lemma is an inductive parallel repetition criterion, similar to one in~\cite{Raz98}.
The proof is identical to the proof of Lemma 3.18 in~\cite{GHMRZ22}, and is included here for the sake of completeness.

\begin{lemma}\label{lemma:prod_set_hard_coor}
	Let $\mc G = (\mc X, \mc A, Q, V)$ be a $k$-player game, and consider its $n$-folds repetition $\mc G^{\otimes n} = (\mc X^{\otimes n}, \mc A^{\otimes n}, Q^{\otimes n}, V^{\otimes n})$ for some sufficiently large $n\in \N$.
	Fix an optimal strategy for the $k$ players in this game, and each for $i\in [n]$, let $\Win_i$ be the event that this strategy wins the $i$\textsuperscript{th} coordinate of the game.
	
	Suppose that $\epsilon,\alpha\in (0,1],\ \alpha\geq 2^{-n}$ are such that the following condition holds: For every product event $E=E^1\times \dots \times E^k \subseteq (\mc X^1)^{\otimes n}\times\dots\times (\mc X^k)^{\otimes n} = \mc X^{\otimes n}$ with $\Pr_{Q^{\otimes n}}[E]\geq \alpha$, there exists a coordinate $i\in [n]$ such that $\Pr\sqbrac{\Win_i\st E} \leq 1-\epsilon$.
	Then, it holds that \[\val(\mc G^{\otimes n}) \leq \brac{1-\frac{\epsilon}{2}}^{\frac{\log_2(1/\alpha)}{2\cdot \log_2(4\abs{\mc X} \abs{\mc A})} }.\]
	In particular, if $\epsilon>0$ is a constant, we get $\val(\mc G^{\otimes n}) \leq \alpha^{c}$, for some constant $c = c(\mc G)>0$.
\end{lemma}
\begin{proof}
	Let $X\in \mc X^{\otimes n}$ be the random variable denoting the questions to the players in $\mc G^{\otimes n}$, and let $A\in \mc A^{\otimes n}$ be the random variable denoting their answers with respect to a fixed optimal strategy.
	Define a sequence of random variables $J_1,\dots,J_n \in [n]$, and $Z_1,\dots,Z_n\in \mc X\times \mc A$ as follows: for each $i=1,2,\dots,n$, let $J_i\in [n]\setminus\set{J_1,\dots,J_{i-1}}$ be a coordinate with the lowest winning probability, conditioned on the value of $(Z_1,\dots,Z_{i-1})$; let $Z_i=(X_{J_i}, A_{J_i})$.
	For each $i\in [n]$, let $W_i$ denote the event $\Win_{J_i}$ of the players winning the coordinate $J_i$.
	
	Let $s = \abs{\mc X} \abs{\mc A}$, and $m= \lfloor \frac{\log_2(1/\alpha)}{\log_2(4s)} \rfloor < n$; assume $m\geq 2$ or else the theorem holds trivially from the lemma hypothesis by taking $E=\mc X^{\otimes n}$.
	We claim that for every integer $k\in [m]$, it holds that $\Pr[W_{\leq k}] \leq \brac{1-\epsilon/2}^k$, where $W_{\leq k} = W_1\land W_2\land\dots \land W_k $.
	Then, substituting $k=m$ gives the desired result.
	
	We prove above claim by induction on $k$.
	The base case $k=1$ follows from the lemma hypothesis by taking $E=\mc X^{\otimes n}$.
	For the inductive step, consider any $1\leq k\leq m-1$, and suppose that $\Pr[W_{\leq k}]\leq \brac{1-\epsilon/2}^k$.
	Further, we may assume that $\Pr[W_{\leq k}] \geq 1/2^{k+1}$, or else we already have $\Pr[W_{\leq k+1}] \leq \Pr[W_{\leq k}] \leq 2^{-(k+1)}\leq \brac{1-\epsilon/2}^{k+1}$.
	Now, observe that $W_{\leq k}$ depends deterministically on the random variables $Z_{\leq k} = (Z_1,\dots, Z_k)$.
	Let $\mc T$ denote the set of all such tuples $z_{\leq k}$ satisfying $W_{\leq k}$, and let $\mc T' \subseteq \mc T$ consist of $z_{\leq k}\in \mc T$ such that $\Pr[Z_{\leq k} = z_{\leq k}] \geq \alpha$; note that by the lemma hypothesis, for each $z_{\leq k}\in \mc T'$, it holds that $\Pr[W_{k+1} \st Z_{\leq k}=z_{\leq k}] \leq 1-\epsilon$, since $Z_{\leq k}=z_{\leq k}$ is a product event of measure at least $\alpha$.
	Hence,
	\begin{align*}
		\Pr\sqbrac{W_{k+1}\st W_{\leq k}} &= \sum_{z_{\leq k}\in \mc T} \Pr\sqbrac{W_{k+1}\st Z_{\leq k} = z_{\leq k}}\cdot \frac{\Pr\sqbrac{Z_{\leq k} = z_{\leq k}}}{\Pr\sqbrac{W_{\leq k}}}
		\\&\leq \sum_{z_{\leq k}\in \mc T'} (1-\epsilon)\cdot \frac{\Pr[Z_{\leq k} = z_{\leq k}]}{\Pr[W_{\leq k}]} + \sum_{z_{\leq k}\in \mc T\setminus \mc T'} 1\cdot  \frac{\Pr[Z_{\leq k} = z_{\leq k}]}{\Pr[W_{\leq k}]}
		\\&= (1-\epsilon) + \epsilon \cdot \sum_{z_{\leq k}\in \mc T\setminus \mc T'}  \frac{\Pr[Z_{\leq k} = z_{\leq k}]}{\Pr[W_{\leq k}]}
		\\&\leq (1-\epsilon) + \epsilon \cdot \frac{\alpha}{2^{-(k+1)}}\cdot s^k 
		\\&\leq (1-\epsilon) + \epsilon\cdot \alpha \cdot (2s)^m \leq (1-\epsilon) + \frac{\epsilon}{2}\cdot \alpha\cdot (4s)^m \leq 1-\frac{\epsilon}{2}.
	\end{align*}	
	This implies $\Pr[W_{\leq k+1}]=\Pr[W_{k+1}\st W_{\leq k}]\cdot \Pr[W_{\leq k}] \leq (1-\epsilon/2)\cdot (1-\epsilon/2)^k = (1-\epsilon/2)^{k+1}$.
\end{proof}

\bibliographystyle{alpha}
\bibliography{main.bib}

\end{document}